%% file: ex_article.tex
\documentclass[
onefignum,
onetabnum]{siamart251216}

\input{ex_shared}

\ifpdf
\hypersetup{
  pdftitle={Convex Analysis of Relaxation Dynamics in Chemical Reaction Networks and Generalized Gradient Flows},
  pdfauthor={K. Sugie, D. Loutchko, T. J. Kobayashi}
}
\fi

\myexternaldocument{ex_supplement}

\begin{document}

\maketitle

% REQUIRED
\begin{abstract}
    We obtain bounds on the Kullback--Leibler divergence to equilibrium for mass-action chemical reaction networks (CRNs) with equilibrium. The associated\KSsout{ decay rates} \KS{trajectory-dependent bounds} are characterized in terms of the singular values of the stoichiometric matrix, convexity parameters, and time-integrated activities via deformed-exponential-type functions. We further extend these bounds within a generalized gradient flow framework. We highlight the biological relevance of this framework: the resulting bounds apply to quasi-steady-state regimes, where long transients and plateau-like behavior are common and functionally important. We illustrate the framework using a catalytic CRN exhibiting plateaus, where the bounds\KSsout{ capture} \KS{provide an a posteriori characterization of} slow relaxation induced by local convexity and provide a bound-based approach to quantifying relaxation in CRNs.
    % \DL{[Maybe make the first sentence more general or move to the end, aka ``We demonstrate how the developed framework has important biological applications/implication.'']}Relaxation dynamics of chemical reaction networks (CRNs) can show long transients and plateaus that are relevant to biological quasi-steady \DL{state} regimes. We develop an analytical framework for such relaxation using generalized gradient flows, derive bounds on Bregman divergences along the flow, and obtain deformed exponential type bounds on the Kullback--Leibler divergence in equilibrium mass-action CRNs. The rates in these bounds are expressed through singular values of the stoichiometric matrix, convexity parameters, and time-integrated activities. We illustrate the framework on a catalytic CRN with plateau behavior, where the bounds capture slow relaxation induced by local convexity, providing a bound-based approach for quantifying relaxation in CRNs.
\end{abstract}

% REQUIRED
\begin{keywords}
  Chemical Reaction Networks, Mass-Action Kinetics, Slow Relaxation, Gradient Flows, Convex Analysis, Non-equilibrium Thermodynamics
\end{keywords}

% REQUIRED
\begin{AMS}
  53B12, 80A30, 90C25, 90C90, 92C45
\end{AMS}
% 53B12  Differential geometric aspects of statistical manifolds and information geometry
% 80A30  Chemical kinetics in thermodynamics and heat transfer
% 92C45  Kinetics in biochemical problems (pharmacokinetics, enzyme kinetics, etc.)
% 90C25  Convex programming
% 90C90  Applications of mathematical programming

% \begin{lemma}[test]
%     \begin{sublems}
%         \item \label[assumption]{ass:U} upupup
%         \begin{align}
%             f(x) = g(x)
%         \end{align}
%         \item \label[assumption]{ass:L} lowlowlow
%         \item \label[assumption]{ass:U2} up2
%     \end{sublems}
% \end{lemma}
% \Cref{ass:U}
% \Cref{ass:L}
% \Cref{ass:U2}

\section{Introduction}
This work aims to establish an analytical description of relaxation in chemical reaction networks (CRNs) by fully exploiting their underlying mathematics, beyond numerical approaches. 

CRNs provide a mathematically rich framework for describing interacting particle systems with constraints arising from stoichiometry and thermodynamics. 
They have been extensively studied in both stochastic and deterministic formulations \cite{Horn1972,Feinberg_2019,Kurtz1972,Agazzi2018}, including the chemical master equation, its deterministic limit, and related large-deviation principles.

Within this broad class, equilibrium CRNs form a particularly tractable subclass: their dynamics with mass-action kinetics admit the Kullback–Leibler (KL) divergence as a Lyapunov function, guaranteeing convergence to equilibrium.
Equilibrium CRNs are also reversible, which makes them a natural setting for nonequilibrium thermodynamics.
They underpin fluctuation theorems \cite{Schmiedl2007,Rao2018,Korbel2021}, thermodynamic inequalities \cite{Yoshimura2021,Nagayama2025,Yoshimura2021thermodynamic,Chun2023}, and decompositions of dissipation \cite{Kobayashi2022,Dechant2022,Yoshimura2023}.
More recently, these systems have been reformulated using the language of gradient flows and information geometry for more general kinetics \cite{Kobayashi2023,Ito2018,Mizohata2024,Yoshimura2021,loutchko2025}, revealing a deep geometric structure behind their relaxation dynamics.

Despite this favorable mathematical structure,\KSsout{ analytical studies} \KS{a systematic quantitative understanding} of relaxation in CRNs remain limited.
\KS{Most existing work on thermodynamic bounds employs quantities other than the KL divergence and derives bounds not on relaxation itself but on fluctuations of time-integrated currents \cite{Yoshimura2021thermodynamic,loutchko2023geometrythermodynamicuncertaintyrelations,PhysRevLett.114.158101}, temporal correlations \cite{PhysRevLett.131.077101,Hasegawa2024}, and first-passage times \cite{Gingrich2017,Pal2021}, as well as on instantaneous rates of change of such quantities \cite{Yoshimura2021,Yoshimura2021thermodynamic,Nagayama2025}.}
\KS{Although several studies have used entropy methods to establish exponential decay of the KL divergence in detailed- and complex-balanced reaction--diffusion systems \cite{Fellner2017,Fellner2018},
such global estimates do not generally clarify how relaxation is shaped by the stoichiometric structure, reaction kinetics, and properties of the underlying thermodynamic function. 
Meanwhile,} existing works on slow or anomalous relaxation \cite{Awazu2009,Awazu2010,himeoka2024} have relied primarily on numerical simulations and phenomenological observations, with comparatively little use of the available gradient-flow or thermodynamic formulations.
In particular, \cite{Awazu2009} examined a special class of equilibrium CRNs whose dynamics admit such structures in principle, but without fully exploiting them to derive analytical bounds or quantitative characterizations of relaxation.

In this work, we establish an analytical description of relaxation in equilibrium CRNs by systematically exploiting their underlying generalized gradient flow structure.
We quantify relaxation through rigorous\KS{, generally trajectory-dependent} bounds on the Bregman divergence, thereby providing a framework that goes beyond numerical observation and connects relaxation behavior directly to geometric and thermodynamic properties of the network.

From an applied perspective, such bounds for slow or quasi-steady relaxation is of broad interest in biology and biophysics.
Living systems are typically maintained far from their ultimate terminal state (death); for example, in resource-limited environments, cells will eventually die but can remain in a quasi-steady dormant state for extended periods.
CRNs exhibiting slow relaxation toward steady states therefore offer a minimal theoretical model for such situations, and the factors controlling this behavior have been explored in several biological and chemical contexts \cite{Awazu2009,Awazu2010,himeoka2024,Himeoka2025,HatakeyamaKaneko2020}.
The analytical framework developed here provides a mathematical basis for understanding these phenomena.

The contributions of this work are twofold:

The first contribution is to provide theoretical bounds on the Bregman divergence along generalized gradient flows and, as a special case, on the KL divergence for equilibrium CRNs \KS{with mass-action kinetics}.
% We introduce a corollary that yields an upper bound on the KL divergence.
Informally, one of our results, \Cref{cor:bounds in quadratic mass action CRNs}, states { the following:}
If the positive orthant $\X = \R^N_{>0}$ is positively invariant in an equilibrium CRN, then the KL divergence $D_\mathrm{KL}(\bm x_T \| \bm x_{\mathrm{eq}})$ satisfies the following upper bound:
\begin{align}
    \label{eq:main result for CRNs}
    D_\mathrm{KL}(\bm x_T \| \bm x_{\mathrm{eq}}) \leq 
    D_\mathrm{KL}(\bm x_0 \| \bm x_{\mathrm{eq}}) \cdot \exp \qty[  - 2\sigma_{\rk(S)}^2 \low{\rho^{\phi_\mathrm{KL}}_{\bm x_\mathrm{eq}}} \low{\Omega}_{\mathrm{log}}^\mathrm{KL}(\bm{x}_{0:T}) ],
\end{align}
which is determined by three components:
(i)\KSsout{ a} \KS{the} smallest non-zero singular value $\sigma_{\rk(S)}$ of the stoichiometric matrix,
(ii) a global convexity parameter $\low{\rho^{\phi_\mathrm{KL}}_{\bm x_\mathrm{eq}}}$ of the KL divergence, and 
(iii) a time-integrated minimal activity $\low{\Omega}_{\mathrm{log}}^\mathrm{KL}(\bm{x}_{0:T})$.
The parameter $\low{\rho^{\phi_\mathrm{KL}}_{\bm x_\mathrm{eq}}}$ characterizes the global convexity along the solution orbit; we also introduce its state-dependent analogue $\rho^{\phi_\mathrm{KL}}_{\bm x_\mathrm{eq}}(\bm x_t)$, describing the local convexity.
\KS{These bounds are generally trajectory dependent, as they retain activity and convexity factors evaluated along the solution orbit rather than depending only on time and a priori constants.}
% \DL{[Is extended the right word? Maybe rather specified (from global to local..)]}
% to a state-dependent parameter $\rho^{\phi_\mathrm{KL}}_{\bm x_\mathrm{eq}}(\bm x_t)$ describing the local convexity. 
These results follow from \Cref{thm:bounds of Bregman divergence in generalized equilibirum flow,thm:bounds under locally convexity assumption}, formulated for generalized gradient flows \cite{Kobayashi2023}.

\KS{Although the overall proof strategy resembles the classical energy--dissipation argument used in entropy methods for reaction--diffusion systems \cite{Fellner2017,Fellner2018}, our formulation makes explicit how relaxation is governed by the stoichiometric structure, reaction activities, and the convexity of the thermodynamic potential. 
To this end, we extend PL-type convex-analytic techniques for standard gradient flows to generalized gradient flows endowed with reaction fluxes and forces \cite{Kobayashi2023}, which {are defined} in a reaction space {whose dimension generally differs from that of the state space}. Except in the long-time limit, the resulting bounds depend on activity and convexity factors evaluated along the solution trajectory.
The main technical ingredients are SVD-based estimates relating the reaction-force norm to the gradient of the Bregman divergence, construction of an \emph{effective concentration} that yields nontrivial estimates despite rank deficiency of the stoichiometric matrix, and a deformed-exponential treatment of nonquadratic dissipation potentials.}

The other contribution concerns the behavior of our bounds for CRNs exhibiting plateaus.
For the example in \cite{Awazu2009}, we numerically confirm that\sout{ our} \KS{the trajectory-dependent local} bound can reproduce the plateau behavior\KS{, whereas the correesponding gloval bound does not}.
\KS{
Although it does not provide a fully predictive a priori estimate, the local-convexity bound captures the plateau and serves as an a posteriori diagnostic of the geometric and kinetic origins of slow relaxation.} 
\KSsout{Interestingly, while the bound based on local convexity reproduces the plateaus, the global one fails;
this suggests that plateaus in our bounds can be controlled by the presence or absence of local convexity.}
To the best of our knowledge, this bound-based relaxation analysis is novel itself.

\KSsout{
% In addition to the main theorems, we make ancillary contributions: first, we introduce a novel convex-analytic approach to generalized gradient flows endowed with reaction fluxes and forces, 
% % which live in spaces of dimension different from that of the state space and are 
% absent in standard gradient flows 
In addition to the main theorems, we make ancillary contributions: first, we extend convex-analytic techniques for standard gradient flows of convex functions to generalized gradient flows endowed with reaction fluxes and forces, whose variables live in spaces of dimension $M$, typically different from the state dimension $N$.
% \DL{[Do you mean that we don't work on the tangent/contangent space?]}.
Second, by exploiting their information geometry and introducing an \emph{effective concentration} that avoids rank-deficiency of stoichiometric matrices, we establish a bound on the force norm.
}

% \subsection{Organizations of This Paper}
This work is organized as follows: 
In \Cref{sec:mathematical preliminaries}, we review preliminaries on CRNs and generalized gradient flows in \cite{Kobayashi2023}. 
In \Cref{sec:convergence analysis of relaxation in generalized gradient flows}, we formulate convexity conditions for the divergence and state our main results on bounds and their corollaries for CRNs.
In \Cref{sec:convergence analysis and characterization of slow relaxation in mass-action reaction systems}, we numerically confirm our bounds for a plateau-exhibiting CRN studied in \cite{Awazu2009}.
Finally, we conclude this paper and discuss the remaining topics in \Cref{sec:Conclusion and Discussion}.

\section{Mathematical Preliminaries}
\label{sec:mathematical preliminaries}
The basic notations are deferred to \ref{sec:standard notations}.
% \DL{Should we refer to SM for the basic notations here?}
\subsection{Chemical Reaction Networks: Structure and Dynamics}
This section presents the mathematical foundations of CRNs. We begin with the structural object, called the stoichiometric matrix.
We then introduce the corresponding ODE, known as the chemical rate equation, and state an assumption on its solutions.
\subsubsection{Structural Properties of CRNs}
\label{subsubsec:strucural properties of CRNs}
Let $X_1, \cdots, X_N$ denote $N$ \textit{chemical species}, and $R_1, \cdots, R_M$ denote $M$ \textit{chemical reactions}.
The (reverisble) chemical equation of the reaction $R_r\, (1\leq r \leq M)$ is formally given as
\begin{align}
   R_r:  \sum_{i \in \Spc_N} s_{ir}^+ \cdot X_i  \rightleftharpoons \sum_{i \in \Spc_N} s_{ir}^- \cdot X_i,
\end{align}
where $s_{ir}^+ \in \Z_{\geq 0}$ is the reactant coefficient of the species $i$, and $s_{ir}^- \in \Z$ is $i$-th product coefficient.
The changes in molecular numbers are represented by the \textit{stoichiometric matrix} $S \in \Z^{N \times M}$, whose $(i,r)$-element is defined as follows:
\begin{align}
    \label{eq:stoichiometric matrix}
    S_{ir} := s_{ir}^+ - s_{ir}^-.
\end{align}

The singular value decomposition (SVD) of the stoichiometric matrix determines systemic chemical reactions, or eigen-reactions, and provides a basis for comparative structural analysis of different metabolic networks \cite{FAMILI200387}.
% \cite{FAMILI200387, PALESE2012151}.
The \KS{transposed} stoichiometric matrix $S\KS{^\top}$ is decomposed as $S\KS{^\top} = \hat U \Sigma \hat V^\top$, where $\hat V = (\bm{\hat v}_1, \cdots, \bm{\hat v}_N) \in \R^{N \times N}, \hat U = (\bm{\hat u}_1, \cdots, \bm{\hat u}_M) \in \R^{M \times M}$ are orthonormal matrices (which means $\hat V^\top \hat V = I_N$ and $\hat U^\top \hat U = I_M$), and the diagonal matrix
$\Sigma = \mathrm{diag}\,(\sigma_1, \cdots, \sigma_{\rk(S)}, 0,\cdots, 0) \in \R^{\KS{M \times N}}$ is the singular matrix of $S$ with the singular values $\sigma_1 \geq \sigma_2 \geq \cdots \geq \sigma_{\rk(S)} > 0$.
Here, $\rk(S) (\leq \min(N,M))$ is the rank of the matrix $S$, and $\sigma_{\rk(S)}$ is the smallest\KSsout{ positive eigenvalue} \KS{singular value} of $S$.
\KS{We assume that $S$ has at least one nonzero element; equivalently, $\rk(S) \geq 1$.}
The sets of vectors $\qty{\bm{\hat v}_i \mid i= 1, \cdots, N}$ and $\qty{\bm{\hat u}_r \mid r= 1, \cdots, M}$ are orthonormal bases of $\R^N$ and $\R^M$, respectively.
Each basis vector $\bm{\hat v}_i \in \R^N$ is transformed by $S^\top$ as 
\begin{align}
    \label{eq:basis by stoichiometric matrix}
    S^\top \bm{\hat v}_i = 
    \begin{cases}
        \sigma_i \bm{\hat u}_i & i \leq \rk(S), \\
        \bm 0 & i > \rk(S).
    \end{cases}
\end{align}
The bases $\qty{\bm{\hat v}_i}$ and $\qty{\bm{\hat u}_r}$ associated % with the SVD 
can be used to represent the four linear subspaces $\Im S, \Ker S^\top \subseteq \R^N, \Im S^\top,\Ker S \subseteq \R^M$ as
$\Im S = \mathrm{span}(\qty{\bm{\hat v}_i \mid 1\leq  i\leq \rk(S)})$,
    $\Ker S^\top = \mathrm{span}(\qty{\bm{\hat v}_i  \mid \rk(S) < i \leq n})$,
    $\Im S^\top = \mathrm{span}(\qty{\bm{\hat u}_r \mid 1 \leq  r\leq \rk(S)})$,
    $\Ker S = \mathrm{span}(\qty{\bm{\hat u}_r \mid \rk(S) < r \leq m})$.

\subsubsection{Dynamics on Thermodynamic CRNs}
\label{subsubsec:dynamics on CRNs}
% CRN dynamics
% \DL{[We never use the polynomials from mass action kinetics but the gradient flow formulation with dissipation functions.
% Would it be better to introduce the gradient flow formulation here right away? We can follow [Cosh gradient systems and tilting by
% MA Peletier, A Schlichting] or our recent submission to information geometry, except that we use $S$ from the previous paragraph instead of div and grad.]}
We introduce the concentration dynamics associated with the thermodynamic CRN.
 The $i$-th element $x_i$ of $\bm x \in \R_{\geq 0}^N$ represents the \textit{concentration} of the chemical $X_i$.
The vector $\bm x$ is called the \textit{concentration vector} or just \textit{state}.
Two differentiable functions $j_r^\pm : \R_{\geq 0}^N \to \R_{\geq 0}$ such that 
\begin{align}
    \label{eq:compatibility of rates and stoichiometry}
    j^\pm_r(\bm x) > 0\Longleftrightarrow \mathrm{supp}\,(\bm s^\pm_r) \subseteq \mathrm{supp}\,(\bm x)
\end{align}
are called the forward and backward \textit{reaction rates} of the reaction $r$, respectively.
\KS{Here, $\mathrm{supp}(\bm a)$ denotes the support of the vector $\bm a \in \R^N$, defined as $\mathrm{supp}(\bm a) := \qty{i \mid a_i > 0}$.}
For thermodynamic consistency, all reactions are reversible, i.e, all forward reactions are paired with their backward ones.
We also denote the reaction rate vector as $\bm j^\pm : \R_{\geq 0}^N \to \R_{\geq 0}^M$.
In particular, the reaction rates are called \textit{mass-action reaction rates} if, for any $1\leq r \leq M$, the $j_r^\pm$ have the form
\begin{align}
    j_r^\pm(\bm x) = k_r^\pm \cdot \prod_{i=1}^N x_i^{s_{ir}^\pm}.
\end{align}
Here, $k_j\pm \in \R_{>0}$ are a certain $\bm x$-independent parameters called rate constants. \Cref{eq:compatibility of rates and stoichiometry} holds for mass-action reaction rates.

The state $\bm x_t$ is assumed to follow the deterministic \textit{chemical rate equation} (CRE):
\begin{align}
    \label{eq:CRE}
    \dv{\bm x_t}{t} = -S (\bm j^+(\bm x_t) - \bm j^-(\bm x_t) ).
\end{align}
The concentration $\bm x_t$ of \cref{eq:CRE} is always restricted in the linear subspace $\bm{x}_0 + \Im S \subseteq \R^N$ at any time $t$ because 
% \begin{align}
    % \label{eq:concentration is restricted to the subspace}
    $\bm{x}_t = \bm x_0 - S \int_0^t \qty[ \bm j^+(\bm x_\tau) - \bm j^-(\bm x_\tau) ] \dd  \tau \in \bm{x}_0 + \Im S$.
% \end{align}
The \textit{stoichiometric subspace} of $\R_{>0}^N$ with respect to $\bm{\tilde x}$ is defined as $\Stoich(\bm{\tilde x}) := ( \bm{\tilde x} + \Im S ) \cap \R_{>0}^N$.
Under the positivity assumption, which is introduced in \Cref{subsubsec:Assumptions on Solutions of Chemical Rate Equation}, the concentration $\bm x_t$ is in $\Stoich(\bm x_0)$ at any time $t$ because $\bm{x}_t\in \bm{x}_0 + \Im S$ and the assumption $\bm x_t \in \R_{>0}^N$.
If restricted to the positive orthant $\R_{>0}^N$, the reaction rate vectors $\bm j^\pm|_{\R_{>0}^N}$ are regarded as maps from $\R_{>0}^N$ to $\R_{>0}^M$.

\subsubsection{\KS{Positivity} Assumption on Solutions of the Chemical Rate Equation}
\label{subsubsec:Assumptions on Solutions of Chemical Rate Equation}
The {positive invariance of the positive orthant} is summarized in this section.
% The following three paragraphs discuss the non-negativity and positivity of concentrations, respectively, and the stronger assumptions on the orbits  in \cref{eq:CRE}.
% We adopt the strongest assumption \cref{eq:closed solution space}, which is stronger than the positivity of concentrations \Cref{asm:invariance of positive orthant}.

% non-negativity 
For a autonomous dynamical system $\dot {\bm x}_t = \bm F(\bm x_t)$ on $\R^N$, define the \textit{finite orbit} of a dynamical system up to $T \geq  0$ as 
\begin{align}
    \label{eq:finite orbit}
    \mathcal{O}_T(\bm x_0) := \qty{\bm x_t = \bm x_0 + \int_0^t \bm F(\bm x_\tau) \dd \tau \mid 0\leq t \leq T}.
\end{align}
A subset $A \subseteq \R^N$ is called \textit{positively invariant} if and only if the following holds:
% , for any initial point $\bm x_0 \in A$ and $t\geq 0$, the state at time $t$ is also in $A$; i.e., $\bm x_t \in A$.
% Using the finite orbit of a dynamical system up to $T \geq  0$, $\mathcal{O}_T(\bm x_0) := \qty{\bm x_t = \bm x_0 + \int_0^t f(\bm x_\tau) \dd \tau \mid 0\leq t \leq T}$, the positive invariance of $A$ can be stated equivalently as 
\begin{align}
    \forall \bm x_0 \in A,\,\forall  T\geq 0,\, \mathcal{O}_T(\bm x_0) \subseteq A.
\end{align}
% The condition on the elements, 
% % \begin{align}
% $
%     \label{eq:invariance of nonnegative orthant}
%     x_i = 0 \implies \dv{x_i}{t} \geq 0
% $
% hold for all $i = 1\leq i \leq N$,
% \end{align}
The following condition is equivalent to the positive invariance of the non-negative orthant $\R^N_{\geq 0}$ \cite[Theorem 5.1]{Bernstein1999}: 
for all $1 \leq i \leq N$, $
    \label{eq:invariance of nonnegative orthant}
    x_i = 0 \implies \dv{x_i}{t} \geq 0
$.
The condition \cref{eq:compatibility of rates and stoichiometry} on $\bm j^\pm$ is sufficient for the positive invariance of $\R_{\geq 0}^N$ \cite[Lemma 3.7.1]{Feinberg_2019}.
% \cref{eq:invariance of nonnegative orthant} to hold \cite[Lemma 3.7.1]{Feinberg_2019}.
{Therefore, the non-negative orthant $\R_{\geq 0}^N$ is positively invariant in the system \cref{eq:CRE} whenever the reaction rates satisfy \cref{eq:compatibility of rates and stoichiometry}.}
% In terms of orbits, one can write the positive invariance of $\R_{\geq 0}^N$ as
% \begin{align}
%     \forall \bm x_0 \in \R_{\geq 0}^N,\, \forall T\geq 0,\, \mathcal{O}_T(\bm x_0) \subseteq \R_{\geq 0}^N.
% \end{align}

% positivity
% positivityの仮定を置く
We denote the positive orthant $\X := \R^N_{>0}$ and refer to it as the \textit{concentration space}.
% In other words, the space $\X$ is assumed to be positive invariant, i.e.
% \begin{align}
%     \label{asm:invariance of positive orthant}
%     \forall \bm x_0 \in \X,\, \forall T\geq 0,\, \mathcal{O}_T(\bm x_0) \subseteq \X.
% \end{align}
We assume the positive invariance of the positive orthant in the systems:
\begin{assumption}[Positive Invariance of Concentration Space]\\
    \label{asm:invariance of positive orthant}
    The concentration space $\X$ is positively invariant in \cref{eq:CRE}.
\end{assumption}
\KS{In terms of finite orbits, \Cref{asm:invariance of positive orthant} is equivalent to the condition where $\mathcal{O}_T(\bm x_0) \subseteq \X$ for any $\bm x_0 \in \X$ and $T > 0$.}
\KS{Note that \Cref{asm:invariance of positive orthant} is automatically satisfied in mass-action CRNs; see \Cref{sec:finite-time positivity of mass-action CRNs} for an explicit proof.} 
\KS{For the more general systems considered here,} following other works on the thermodynamic and geometric structure of CRNs \cite{Kobayashi2023,Mizohata2024,Sughiyama2022growing,Sughiyama2022}, we simply assume \Cref{asm:invariance of positive orthant} for systems of interest and thereby avoid intricate problems.
% \footnote{\KSsout{\Cref{asm:invariance of positive orthant} is also known as persistence. Even for systems equipped with a Lyapunov function (e.g., the equilibrium systems introduced later in \cref{eq:generalized gradient flow}), persistence is generally not satisfied. Determining whether a given class of dynamical systems is persistent is a hard problem; in particular, for a certain class of mass-action CRNs, including equilibrium CRNs, this problem is known as the Global Attractor Conjecture \cite{Horn1972,Craciun2013} and remains open.}}
\KSsout{Indeed, this single assumption is sufficient to obtain }\KS{Together with the additional assumption of convergence to a positive equilibrium introduced below, \Cref{asm:invariance of positive orthant} allows us to derive} explicit bounds \KS{with respect to local convexity (\Cref{thm:bounds under locally convexity assumption})}\KSsout{ in equilibrium mass-action CRNs}.

\subsection{Geometric and Thermodynamic Aspects of CRN Dynamics}
When mass-action CRE satisfies a parameter condition known as the equilibrium condition, the system can be rewritten in the form of a generalized gradient flow \cref{eq:generalized gradient flow} \cite{Mielke2017,Mielke2014,Renger2018,loutchko2025}.
This formulation includes a broader class of dynamics beyond mass-action CRNs and the associated Lyapunov function, Kullback-Leibler divergence. 
The main results of this paper are stated within this generalized framework. 

In the following, we present the theoretical components underlying the generalized gradient flow formulation. 
The next two sections introduce the information-geometric relationship between the concentration and potential spaces, as well as the corresponding generalized gradient flows. 
We then describe the thermodynamic structure of equilibrium mass-action CRNs as a special case of this formulation.

\subsubsection{Information-Geometry of Concentration and Potential Spaces}
\label{subsubsec:information-geometric structure between concentration and potential spaces}
We introduce the affine subspaces and their Legendre-transformed images that arise in CRNs.
Although these manifolds are no longer orthogonal in the usual Euclidean sense, they satisfy a generalized information-geometric orthogonality characterized by Bregman divergence.
The definitions in this section are from \cite{Kobayashi2023, Mizohata2024}.

A strictly convex differentiable function 
$\phi: \X \to \R$ is called the \textit{thermodynamic function} if its gradient $\partial \phi: \X \to \R^N, \bm x \mapsto \partial\phi(\bm x) = (\partial_{x_1} \phi(\bm x), \cdots, \partial_{x_n} \phi(\bm x) )^\top$ is\KSsout{ injective} \KS{bijective}.\footnote{In the language of differential geometry \cite{Tu2011}, the differential of a function $\phi$ at $\bm x \in \X$ naturally lives in the cotangent space $T^\ast_{\bm x} \X$. For simplicity, we equip $\X$ with a global coordinate $(x_1, \cdots, x_N)$ and identify every cotangent space $T_{\bm x}^\ast \X$ with $\R^N$.}
Denote the image of $\partial \phi$ as $\M := \partial \phi(\X) = \R^N$, which is called the \textit{potential space}.
% \KS{bijective would imply that its image is automatically all of $R^N$. So injective is sufficient.?? to do}
For the thermodynamic function $\phi$, the \textit{convex conjugate} of $\phi$ is defined as 
\begin{align}
    \label{eq:convex conjugate}
    \phi^\ast(\bm \mu) := \max_{\bm x \in \X} [\expval{\bm x, \bm \mu} - \phi(\bm x)],
\end{align}
where $\expval{\cdot, \cdot}$ denotes the standard inner product under: $\expval{\bm x, \bm \mu} := \sum_{i=1}^N x_i \mu_i$.
The maximum in \cref{eq:convex conjugate} is attained iff $\bm \mu = \partial \phi(\bm x)$ holds. Thus, the equality
\begin{align}
    \label{eq:Young-Fenchel identity}
    \phi(\bm x) +  \phi^\ast(\partial \phi(\bm x)) = \expval{\bm x, \partial \phi(\bm x)} 
\end{align}
holds for any $\bm x \in \X$, which is known as the Young-Fenchel identity.
The gradient of the conjugate $\partial\phi^\ast: \M \to \X, \bm x \mapsto \partial\phi^\ast(\bm \mu) = (\partial_{\mu_1} \phi^\ast(\bm \mu), \cdots, \partial_{\mu_n} \phi^\ast(\bm \mu) )^\top$ is also a bijection, where $\partial\phi^\ast = (\partial\phi)^{-1}$ holds.
The two bijective maps $\partial \phi :\X \to \M$ and $\partial \phi^\ast: \M \to \X$ are called the \textit{Legendre transformations} between $\X$ and $\M$.

In CRNs, the stoichiometric matrix $S$ \cref{eq:stoichiometric matrix} determines the two affine subspaces  on $\X$ and $\M$, respectively.
\begin{itemize}
    \item The \textit{stoichiometric subspace} of $\X$ for $\bm{\tilde x}$: $\Stoich(\bm{\tilde x}) := ( \bm{\tilde x} + \Im S ) \cap \X$
    \item The \textit{equilibrium subspace} of $\M$ for $\bm{\tilde \mu}$: $\Equib(\bm{\tilde \mu} ) := (\bm{\tilde \mu} + \mathrm{Ker\,} S^\top) \cap \M$
\end{itemize}
The stoichiometric subspace is the same as that defined in \cref{subsubsec:dynamics on CRNs}.
The two spaces $\X$ and $\M$ are orthogonal complement to each other in the sense that
$\expval{\bm x- \bm{\tilde x}, \bm \mu - \bm{\tilde \mu}} = 0$ for any $\bm x \in \Stoich(\bm{\tilde x})$ and $\bm \mu \in \Equib(\bm{\tilde \mu})$.
We introduce the images of these affine subspaces via $\partial \phi$ and $\partial \phi^\ast$, respectively.
\begin{itemize}
    \item The \textit{equilibrium submanifold} of $\X$ for $\bm{\tilde \mu}$: $\partial \phi^\ast[\Equib(\bm{\tilde \mu})] \subset \X$
    \item The \textit{stoichiometric submanifold} of $\M$ for $\bm{\tilde x}$: $\partial \phi[\Stoich(\bm{\tilde x})] \subseteq \M$
\end{itemize}

The \textit{Bregman divergence} for $\phi$ between $\bm x, \bm y \in \X$ is defined as 
\begin{align}
    \label{eq:Bregman divergence}
    D_\phi(\bm x \| \bm y) := \phi(\bm x) - \phi(\bm y) - \expval{\partial \phi(\bm y) , \bm x - \bm y}.
\end{align}
The non-negativity of the Bregman divergence follows from the convexity of the thermodynamic function $\phi$.
Moreover, from the strict convexity of $\phi$, it follows that $D_\phi(\bm x\| \bm y) = 0$ holds if and only if $\bm x = \bm y$.
The Bregman divergences for $\phi$ and its convex conjugate $\phi^\ast$ are related as $D_\phi(\bm x \| \bm y) = D_{\phi^\ast}(\partial\phi(\bm y) \| \partial \phi(\bm x))$, for any $\bm x, \bm y \in \X$. 

There is a general information-geometric orthogonality between the stoichiometric subspace $\Stoich(\bm x_1)$ and the equilibrium submanifold $\partial \phi^\ast[\Equib(\partial \phi(\bm x_2))]$ as follows: for any states $\bm{x}_1$ and $\bm x_2$, there exists an unique point $\bm x^\dagger(\bm x_1, \bm x_2) \in \X$ such that $\qty{\bm x^\dagger(\bm x_1, \bm x_2)} = \Stoich(\bm{x}_1) \cap \partial \phi^\ast[\Equib(\partial \phi(\bm{x}_2))]$ \cite[Lemma 1]{Kobayashi2023}. 
The orthogonal complementarity of $\Stoich(\bm{x}_1)$ and $\Equib(\partial \phi(\bm{x}_2))$ implies $\expval{ \partial \phi(\bm x^\dagger(\bm x_1, \bm x_2))  - \partial \phi(\bm x_2), \bm x^\dagger(\bm x_1, \bm x_2)  - \bm x_1} = 0$, which is equivalent to the generalized Pythagorean relation of Bregman divergences:
\begin{align}
    D_\phi(\bm x_1 \| \bm x_2) = D_\phi(\bm x_1 \| \bm x^\dagger(\bm x_1, \bm x_2)) + D_\phi(\bm x^\dagger(\bm x_1, \bm x_2) \| \bm x_2).
\end{align}
% In information-geometric terms, this equality can also be interpreted as the orthogonality of the $m$-geodesic and $e$-geodesic at the point $\bm x^\dagger$ with respect to the metric $\partial^2 \phi$ on $\X$ \cite[Section 1.6]{Amari2016} 
% \KS{Confirm}.

% Furthermore, by Legendre duality, analogous relationship holds for the stoichiometric submanifold $\partial \phi[\Stoich(\partial \phi^\ast(\bm \mu_1))]$ and the equilibrium subspace  $\Equib(\bm{\mu_2} )$: for any $\bm{\mu}_1$ and $\bm \mu_2$, there exists an unique point $\bm \mu^\dagger(\bm \mu_1, \bm \mu_2) \in \M$ such that $\qty{\bm \mu^\dagger(\bm \mu_1, \bm \mu_2)} = \partial \phi[\Stoich(\partial \phi^\ast(\bm \mu_1))]  \cap \Equib(\bm{\mu_2} )$. The generalized Pythagorean relation also holds in $\M$: $    D_{\phi^\ast}(\bm \mu_1 \| \bm \mu_2) = D_{\phi^\ast}(\bm \mu_1 \| \bm \mu^\dagger(\bm \mu_1, \bm \mu_2)) + D_{\phi^\ast}(\bm \mu^\dagger(\bm \mu_1, \bm \mu_2) \| \bm \mu_2)$, which follows directly from the equality of $D_\phi$ and $D_{\phi^\ast}$.

% \KS{skipped the unused notations on orthogonality in dual coordinates }

\subsubsection{Dissipative Structure, Generalized Gradient Flow, \KS{and Convergence Assumption}}
\label{subsubsec:dissipative structure and generalized gradient flow}
This section outlines generalized gradient flows \cite{Kobayashi2023}, a framework consistent with macroscopic fluctuation and thermodynamics, which is a generalization of concentration dynamics in CRNs.
% (\Cref{subsubsec:dynamics on CRNs}). 
In this picture, the gradient of the thermodynamic potential gives forces which drive the system.
A dissipation function gives a Legendre transformation from forces to fluxes \cite{Mielke2017,Mielke2014}, which includes Onsager’s flux–force relation \cite{Onsager19311,Onsager19312}. 

We now dicuss the mathematical details.
Let $\F := \R^M$ be the \textit{force space}\footnote{Formally, each state $\bm x \in \X$ has its own force space $\F_{\bm x}$.\cite{Kobayashi2023}; here we identify them all with $\F$.}
A \textit{dissipation function} at $\bm x \in \X$ is a strictly convex differentiable function $\Psi_{\bm x}^\ast : \F \to \R$ that satisfies the following conditions:
\begin{align}
    \label{eq:dissipation function}
    \frac{\Psi_{\bm x}^\ast(\bm f)}{\norm{\bm f}_2} \to \infty \text{ as } \norm{\bm f}_2 \to \infty,\,
    \Psi_{\bm x}^\ast(\bm f) = \Psi_{\bm x}^\ast(-\bm f),\,
    \Psi_{\bm x}^\ast(\bm 0) = 0.
\end{align}
The gradient of ${\Psi_{\bm x}^\ast(\bm f)}$, $\partial_{\bm f} \Psi_{\bm x}^\ast: \F \to \R^M, \bm f \mapsto (\partial_{f_1} \Psi_{\bm x}^\ast(\bm f), \cdots, \partial_{f_M} \Psi_{\bm x}^\ast(\bm f))^\top$, is bijective. 
The convex conjugate of $\Psi^\ast_{\bm x}$, denoted by $\Psi_{\bm x}$, is also a dissipation function \cite[Proposition 1]{Kobayashi2023}.
For the pair of dissipation functions $\Psi_{\bm x}$ and $\Psi_{\bm x}^\ast$  the Young-Fenchel relation \cref{eq:Young-Fenchel identity} also holds.
% The image of $\partial_{\bm f} \Psi_{\bm x}^\ast$ is called the flux space: $\J := \partial_{\bm f} \Psi_{\bm x}^\ast(\F) = \R^M_{>0}$.
A dissipation function $\Psi_{\bm x}^\ast$ is called \textit{separable} if and only if, it can be represented as 
\begin{align}
    \label{eq:separable}
    \Psi_{\bm x}^\ast(\bm f) = \sum_{r = 1}^M \omega_r( \bm {x}) \psi^\ast(f_r)
\end{align}
with a positive vector-valued function $\bm \omega: \X \to \R_{>0}^M$ and a scalar dissipation function $\psi^\ast: \R \to \R$. $\bm \omega$ is called the \textit{activity} of the separable dissipation function.
The gradient of a separable $\Psi_{\bm x}^\ast(\bm f)$ with respect to {$f_r$} is explicitly written as 
\begin{align}
    \label{eq:gradient of separable}
    \partial_{{f_r}}\Psi_{\bm x}^\ast(\bm f) = \omega_r( \bm {x}) (\psi^\ast)'(f_r),
\end{align}
where $(\psi^\ast)'(f)$ is the derivative of the scalar dissipation function.

We introduce the concentration dynamics, and connect the force-flux structure with the concentration-potential dual structure.
Let $\F^\mathrm{eq} := \Im S^\top $ be the \textit{equilibrium subspace} \KS{\sout{of} in} the force space $\F$.
Consider the $\bm x$-dependent force generated by the gradient of Bregman divergence associated with a thermodynamic function $\phi$, which is called the (\textit{equilibrium}) \textit{force}:
\begin{align}
    \label{eq:force}
    \bm f (\bm x) := S^\top \partial_{\bm x} D_\phi(\bm x \| \bm{\hat x}) \in \F^\mathrm{eq} \subseteq \F,
\end{align}
where $\bm{\hat x} \in \X$ is a fixed parameter. 
The equilibrium submanifold at $\partial \phi(\bm{\hat x})$ is also represented as 
$\partial \phi^\ast (\Equib(\partial \phi(\bm{\hat{x}})) )= \qty{\bm x \in \X \mid \bm f(\bm x) = S^\top \partial_{\bm x} D_\phi(\bm x \| \bm{\hat x}) = \bm 0}$. 
The flux at $\bm x$ is given by the Legendre transformation of $\bm f(\bm x)$ as: $\bm j(\bm x) = \partial_{\bm f} \Psi^\ast_{\bm x} (\bm f(\bm x))$.

The \textit{generalized gradient flow} \cite[Definition 23]{Kobayashi2023} associated with a stoichiometric matrix $S$, a thermodynamic function $\phi$, a dissipation function $\Psi^{\ast}$, and a parameter $\bm{\hat x} \in \X$ is the following ODE:
\begin{align}
    \label{eq:generalized gradient flow}
    \dot{\bm x}_t = - S \bm j(\bm x) = -S \partial_{\bm f} \Psi_{\bm x_t}^\ast (\bm f(\bm x_t)) = -S \partial_{\bm f} \Psi_{\bm x_t}^\ast [S^\top \partial_{\bm x} D_\phi(\bm x_t \| \bm{\hat x})].
\end{align}
The equilibrium submanifold $\partial \phi^\ast (\Equib(\partial \phi(\bm{\hat{x}})) )$ is also the set of the steady states of the system \cref{eq:generalized gradient flow} \cite[Proposition 3]{Kobayashi2023}: for any $\bm x\in \X$,
$
    S \bm j(\bm x) = \bm 0 \Longleftrightarrow \bm x \in \partial \phi^\ast (\Equib(\partial \phi(\bm{\hat{x}})) ).
$
Henceforth, we refer to any steady state of the system \cref{eq:generalized gradient flow} as an \textit{equilibrium state}.

In the generalized gradient flow, the divergence between $\bm x_t$ and $\bm{\hat x}$ is non-increasing over time as follows:
\begin{align}
    % \begin{aligned}
        \dot{D}_{\phi}(\bm x_t \| \bm{\hat x}) &\overset{\cref{eq:generalized gradient flow}}{=}
        -\expval{S \partial_{\bm f} \Psi_{\bm x_t}^\ast (\bm f(\bm x_t)), \partial_{\bm x_t} D_\phi(\bm x_t \| \bm{\hat x})}\\
        &\overset{\phantom{\cref{eq:generalized gradient flow}}}{=} -\expval{\partial_{\bm f} \Psi_{\bm x_t}^\ast (\bm f(\bm x_t)), S^\top\partial_{\bm x_t} D_\phi(\bm x_t \| \bm{\hat x})}\\
        \label{eq:3rd eq.}
        &\overset{\cref{eq:force}}{=} - \expval{\partial_{\bm f} \Psi_{\bm x_t}^\ast (\bm f(\bm x_t)), \bm f(\bm x_t)}\\
        \label{eq:4th eq.}
        &\overset{\phantom{\cref{eq:generalized gradient flow}}}{=} - (\Psi_{\bm x}^\ast[\bm f(\bm x_t)] + \Psi_{\bm x}[\partial_{\bm f} \Psi_{\bm x_t}^\ast (\bm f(\bm x_t))]) \leq 0.
    % \end{aligned}
\end{align}
The transformation from \cref{eq:3rd eq.} to \cref{eq:4th eq.} follows from the Legendre-Fenchel identity \cref{eq:Young-Fenchel identity} and the last inequality is from the non-negativity of dissipation functions \cite{Kobayashi2023, Mielke2017, Mielke2014}.
The equality holds iff $\bm f(\bm x_t) = \bm 0$, i.e., $\bm x_t \in \partial \phi^\ast (\Equib(\partial \phi(\bm{\hat{x}})) )$.
Therefore, the divergence $D_{\phi}(\bm x_t \| \bm{\hat x})$ is always decreasing except at the points in $\partial \phi^\ast (\Equib(\partial \phi(\bm{\hat{x}})) )$ \cite[Proposition 3]{Kobayashi2023}.
For an initial condition $\bm x_0 \in \X$, the trajectory is in the stoichiometric subspace of $\bm x_0$, i.e., $\bm x_t \in \Stoich (\bm x_0)$, and the unique equilibrium state $\bm x_\mathrm{eq}(\bm x_0, \bm{\hat x})$ is in $\Stoich (\bm x_0) \cap \partial \phi^\ast (\Equib(\partial \phi(\bm{\hat{x}})) )$.
The point $\bm x_\mathrm{eq}(\bm x_0, \bm{\hat x})$ is simply denoted by $\bm x_\mathrm{eq}$ in this article.
Since $\partial \phi(\bm{x}_\mathrm{eq}) - \partial \phi(\bm{\hat x}) \in \Ker S^\top$, the force can also be represented by $\bm x_\text{eq}$:
\begin{align}
    \label{equilibirum force with x_eq}
    \bm f(\bm x) = S^\top \partial_{\bm x} D_\phi(\bm x \| \bm{x}_\text{eq}).
\end{align}
\KSsout{which goes to $\bm{0}$ for $t \to \infty$.}
The expression $\dot \Sigma_t := - \dot D_\phi(\bm x_t \| \bm{\hat x}) \geq 0$ is called the entropy production rate (EPR).

\KSsout{In this case, the divergence is decreasing and the concentration $\bm x_t$ is approaching $\bm x_\mathrm{eq}$} 
\KS{
We note that although the divergence decreases except on the equilibrium submanifold, the solution $\bm x_t$ can approach the boundary.
More precisely, in terms of orbits, for each finite $T > 0$, the finite orbit $\mathcal{O}_T(\bm x_0)$ \cref{eq:finite orbit} is compact in $\X$, since it is the continuous image of a compact interval $[0,T]$.
However, the infinite-time orbit,
\begin{align}
    \mathcal{O}_{\infty}(\bm x_0) := \overline{\bigcup_{T\geq 0} \mathcal{O}_{T}(\bm x_0)},
\end{align} 
need not be contained in $\X$.
Here, $\up{\cdot}$ denotes the closure in $\R^N$ equipped with the standard Euclidean topology.
}

\KS{
We consider a convergence assumption on the generalized gradient flow \cref{eq:generalized gradient flow}.
\begin{assumption}[Global Convergence to the positive equilibrium in $\X$]
    \label{asm:convergence to the positive equilibrium}
    For the generalized gradient flow \cref{eq:generalized gradient flow} with a parameter $\hat{\bm x} \in \X$ and  initial state $\bm x_0 \in \X$, the solution $\bm x_t$ converges to the equilibrium state $\bm x_\mathrm{eq}(\bm x_0, \hat{\bm x})$, namely,
    \begin{align}
        \label{eq:infite orbit}
        \lim_{t \to \infty} \bm x_t = \bm x_\mathrm{eq} = \bm x_\mathrm{eq}(\bm x_0, \hat{\bm x}) \in \X.
    \end{align}
    Here, the infinite orbit \cref{eq:infite orbit} can be written as
    \begin{align}
        \mathcal{O}_{\infty}(\bm x_0) = \qty{\bm x_t = \bm x_0 -S \int_0^t \bm j(\bm x_\tau)\dd \tau \mid 0\leq t} \cup \qty{\bm x_\mathrm{eq}} \subseteq \X.
    \end{align}
\end{assumption}
% \Cref{asm:convergence to the positive equilibrium} is automatic for equilibrium mass-action CRNs, and more generally for complex-balanced mass-action CRNs, by the Global Attractor Conjecture \cite{Craciun2015}.
The infinite orbit $\mathcal{O}_\infty(\bm x_0)$ is compact in $\X$ under \Cref{asm:convergence to the positive equilibrium}, which implies that there exists a positive constant $R > 0$ such that $R^{-1} \leq x_i \leq  R$ for $i = 1, \cdots, N$ and $\bm x \in \mathcal{O}_{\infty}(\bm x_0)$.\footnote{\KS{Similar assumptions on orbits can be found in the literature on convergence to equilibrium in reaction-diffusion systems \cite[Theorem 5.1]{Glitzky1997}.}}
\Cref{asm:convergence to the positive equilibrium} implies persistence \cite[Definition 2.7]{Craciun2013} because $\liminf_{t\to \infty} x_{i,t} = \lim_{t\to \infty} x_{i,t} = x_{\mathrm{eq},i} > 0$ for any $\bm x_0 \in \X$ and $1\leq i\leq N$.
Moreover, under \Cref{asm:convergence to the positive equilibrium}, the system is permanent \cite[Definition 2.9--2.10]{Craciun2013} (uniformly persistent) on each stoichiometric subspace $\Stoich(\tilde{\bm x})$ because all initial states $\bm x_0 \in \Stoich(\tilde{\bm x})$ converge to the same equilibrium state $\bm x_\mathrm{eq}(\tilde{\bm x}, \hat{\bm x})$ and $0 < \min_i x_{\mathrm{eq},i} \leq \max_{i} x_{\mathrm{eq},i} < \infty$ holds.
For equilibrium mass-action CRNs, and more generally for complex-balanced mass-action CRNs, \Cref{asm:convergence to the positive equilibrium} follows from the Global Attractor Conjecture \cite{Craciun2015}.\footnote{\KS{For complex-balanced mass-action systems, Horn and Jackson introduced a strict Lyapunov function and established local asymptotic stability of the positive complex-balanced equilibrium in each stoichiometric compatibility class \cite{Horn1972_2,Horn1972}. The remaining difficulty was to exclude the possibility that a positive trajectory approaches the boundary of the positive orthant, where the Lyapunov function alone does not rule out the behavior. This problem became known as the Global Attractor Conjecture, formulated by Horn and proved by Craciun using toric differential inclusions \cite{Craciun2015}.}}
We impose the assumption here because the generalized gradient flow formulation is broader than the class of equilibrium mass-action systems, and the uniqueness and global convergence need not follow from the Lyapunov structure alone.
}

% The concentration space $\X$, the potential space $\M$, and the force space $\F^\mathrm{eq}$ relate as
% \begin{align}
%     \label{eq:sequence from X to F_eq}
%     \X \overset{\partial \phi}{\longleftrightarrow} \M \overset{-\partial \phi(\bm{\hat x})}{\longleftrightarrow} \M - \partial \phi(\bm{\hat x}) \overset{S^\top}{\longrightarrow} \F^{\mathrm{eq}},
% \end{align}
% where $\overset{g}{\longleftrightarrow}$ means that the map $g$ is bijective.
% The same relation hold for the restriction to the stoichiometric subspace of $\qty{\bm x_\mathrm{eq}}$ with $\qty{\bm x_\mathrm{eq}} = \Stoich(\bm x_\mathrm{0}) \cap \partial \phi^\ast[\Stoich(\partial \phi (\bm{\hat x}))]$:
% \begin{align}
%     \label{eq:sequence from S(x_0) to F_eq}
%     \Stoich(\bm x_\mathrm{eq}) \overset{\partial \phi}{\longleftrightarrow} \partial\phi[\Stoich(\bm x_\mathrm{eq}) ] \overset{-\partial \phi(\bm{x}_\mathrm{eq})}{\longleftrightarrow} \partial_{\bm x} D_\phi[\Stoich(\bm x_\mathrm{eq}) \| \bm x_\mathrm{eq}] \overset{S^\top}{\longrightarrow} \F^{\mathrm{eq}}.
% \end{align}

\subsubsection{Thermodynamics of Equilibrium Mass-Action CRNs}
\label{subsubsec:equilibrium thermodynamics on CRNs}
In CRNs with mass action kinetics, the following thermodynamic function\footnote{Physically, $\phi_\mathrm{KL}$ is identified with the partial grand potential density of the ideal gas systems \cite[Appendix A]{Sughiyama2022}.} is consistent with generalized gradient flow structure:
\begin{align}
    \phi_\mathrm{KL}(\bm x) := \sum_{i = 1}^n \qty(x_i \ln x_i - x_i ).
\end{align}
The gradient of $\phi_\mathrm{KL}$ is $\partial \phi_\mathrm{KL}(\bm x) = \ln \bm x$, which is the Ledendre transform between $\X$ and $\M$.
The Bregman divergence associated with $\phi_\mathrm{KL}$ is given by
\begin{align}
    \label{eq:KL divergence}
    D_{\mathrm{KL}}(\bm x \| \bm y) := D_{\phi_\mathrm{KL}}(\bm x \| \bm y) = \sum_{i = 1}^n \qty(x_i \ln \frac{x_i}{y_i} - x_i + y_i),
\end{align}
for $\bm x, \bm y \in \X$.
$D_\mathrm{KL}$ is called the \textit{Kullback-Leibler} (\textit{KL}) \textit{divergence}.
Mass-action CRNs are called \textit{equilibrium mass-action CRNs} if the systems satisfy the following \textit{Wegscheider equilibrium condition},
$\ln ({\bm k^+}/{\bm k^-}) \in \Im S^\top$,
referred to in \cite{Schuster1989} and \cite[Theorem 1]{Sughiyama2022}.
The Wegscheider condition is equivalent to the existence of $\bm{\hat x} \in \X$ such that $\ln({\bm k^+}/{\bm k^-}) = -S^\top \ln \bm{\hat x}$.
The flux $\bm j(\bm x)$ is defined as $\bm j(\bm x) = \bm j^+(\bm x) - \bm j^-(\bm x)$.
The thermodynamic force $\bm f(\bm x)$ at $\bm x \in \X$ of the mass action CRN is given as 
\begin{align}
    \label{eq:force in CRNs}
    \bm f(\bm x) = \ln \frac{\bm j^+(\bm x)}{\bm j^-(\bm x)} = S^\top \ln \frac{\bm x}{\bm{\hat x}} = S^\top \partial_{\bm x} D_{\mathrm{KL}}(\bm x \|  \bm{\hat x}).
\end{align}

The dynamics of the equilibrium mass action CRNs \cref{eq:CRE} are special cases of the generalized gradient flows \cref{eq:generalized gradient flow}.
However, the choice of the dissipation function to map the chemical rate equation \cref{eq:CRE} to \cref{eq:generalized gradient flow} is not unique; formally, there exists an infinite variety of possible choices \cite{Nagayama2025}.
Here, we present two commonly studied dissipation functions.
The \textit{quadratic dissipation function}  of the mass action systems
are defined as the separable dissipation function \cref{eq:separable} with the logarithmic mean activity $\bm \omega_\mathrm{log}$ \cite{VanVu2023}
% \footnote{The logarithmic activity is known as the dynamical mobility in stochastic thermodynamics } 
and the quadratic scalar dissipation function $\psi^{\ast}_\mathrm{quad}$, which are given as
\begin{align}
    \label{eq:quadratic dissipation}
    \bm \omega_\mathrm{log}(\bm x) := \frac{\bm j^+(\bm x) - \bm j^- (\bm x)}{ \ln \bm j^+(\bm x) - \ln \bm j^- (\bm x)}, \,
    \psi^{\ast}_\mathrm{quad}(f) := \frac{1}{2}f^2.
\end{align}
The \textit{hyperbolic dissipation function} is the separable dissipation function with the geometric mean activity $\bm \omega_\mathrm{geo}$ and the cosh scalar dissipation function $\psi^\ast_\mathrm{hyp}$,
\begin{align}
    \label{eq:hyperbolic dissipation}
    \bm \omega_\mathrm{geo}(\bm x) := \sqrt{\bm j^+(\bm x) \cdot \bm j^-(\bm x)}, \,
    \psi^{\ast}_\mathrm{hyp}(f) := 4 \qty[\cosh\qty(\frac{f}{2}) - 1].
\end{align}
The quadratic dissipation functions induce the associated norm 
$\| \cdot \|_{\bm \omega_\mathrm{log}(\bm x)}$, defined as
% \begin{align}
$
    \| \bm f \|_{\bm \omega_\mathrm{log}(\bm x)}^2 := \sum_{r = 1}^M (\omega_{\mathrm{log}}(\bm x))_{r} \cdot f_r^2
$,
% \end{align}
into the force spaces \cite{Yoshimura2023}, while the hyperbolic ones are related to large deviations and macroscopic fluctuation theory \cite{Patterson2019}.
% They directly induce a norm on the force space.
% A Riemannian metric essentially is a norm on the tangent spaces of the force space, so should not be mentioned here.

\begin{KSblock}
\begin{example}[The spatially homogeneous Gray--Scott model]
\label{eg:gray scott model}
We consider the spatially homogeneous reduction of the reversible Gray–Scott model as a representative example of an equilibrium mass-action CRN and our formulation. The model has four species $X_i\,(i=1,\cdots, 4)$ and three pairs of forward and backward reactions, given by
\begin{align}
\begin{aligned}
    \label{eq:gray-scott}
    R_1: X_1 + 2X_2 &\rightleftharpoons 3X_2,\\
    R_2: \phantom{X_1 + 2}X_2 &\rightleftharpoons X_3,\\
    R_3: \phantom{X_1 + 2}\KSTwo{X_1} &\rightleftharpoons X_4.
\end{aligned}
\end{align}
The stoichiometric matrix of this system is given by
\begin{align}
    S = 
    \KSTwo{\begin{pmatrix}
        1 & 0 & 1 \\
        -1 & 1 & 0 \\
        0 & -1 & 0 \\
        0 & 0 & -1
    \end{pmatrix}}.
\end{align}
Each reaction follows mass-action kinetics with rate constants $k_r^\pm \,(r=1,2,3)$.
The flux vector $\bm j(\bm x) = (j_1(\bm x), j_2(\bm x), j_3(\bm x))^\top$ is given by 
\begin{align}
    \begin{aligned}
        j_1(\bm x) = k_1^+x_1x_2^2-k_1^-x_2^3,\quad
        j_2(\bm x) = k_2^+x_2-k_2^-x_3,\quad
        j_3(\bm x) = \KSTwo{k_3^+x_1-k_3^-x_4}.
    \end{aligned}
\end{align}
The dynamics of the concentration vector $\bm x_t = (x_1, x_2, x_3, x_4)^\top$ follows the rate equation \eqref{eq:CRE} and is restricted to the stoichiometric subspace $\Stoich(\bm x_0) = (\bm x_0 + \Im S) \cap \X = \qty{\bm x \in \X \mid \bm 1^\top \bm x = \bm 1^\top \bm x_0}$, where the total concentration $\bm 1^\top \bm x = x_1 + x_2 + x_3 + x_4$ is conserved over time.
Under the equilibrium condition $\ln(\bm k^+/\bm k^-) = -S^\top \ln \bm{\hat x}$ for some $\bm{\hat x} \in \X$, the thermodynamic force is given by $\bm f(\bm x) = S^\top \partial_{\bm x} D_\mathrm{KL}(\bm x \| \bm{\hat x})$.
The reaction-wise components can be written as follows:
\begin{align}
    \begin{aligned}
        f_1(\bm x) &= \ln\frac{k_1^+x_1x_2^2}{k_1^-x_2^3} = \ln\frac{x_1 / {\hat x}_1}{x_2 / {\hat x}_2},\\
        f_2(\bm x) &= \ln\frac{k_2^+x_2}{k_2^-x_3} = \ln\frac{x_2 / {\hat x}_2}{x_3 / {\hat x}_3},\\
        f_3(\bm x) &= \KSTwo{\ln\frac{k_3^+x_1}{k_3^-x_4} = \ln\frac{x_1 / {\hat x}_1}{x_4 / {\hat x}_4}}.
    \end{aligned}
\end{align}
The KL divergence \cref{eq:KL divergence} is a Lyapunov function for the system. 
For example, by choosing a quadratic dissipation \cref{eq:quadratic dissipation}, the system can be written as the generalized gradient flow:
\begin{align}
    \dot{\bm x}_t = - S\, \mathrm{diag}(\bm\omega_{\log}(\bm x_t)) \, \bm f(\bm x_t) = - S \partial_{\bm f} \left.\qty[ \sum_{r=1}^3 \omega_{\log,r}(\bm x_t) \cdot \frac{f_r^2}{2} ]\right|_{\bm f = \bm f(\bm x_t)}.
\end{align}

As shown for the reversible Gray--Scott system \cite{laurencot2025} \footnote{
For this specific model, Laurencot and Walker showed that every nonnegative solution converges to the positive equilibrium unless it starts on the invariant face $x_2=x_3=0$, in which case it converges to the corresponding boundary equilibrium. In the present work, however, we consider initial conditions in the positive concentration space $\X$, and hence this boundary case is excluded.
}, its spatially homogeneous reduction admits a broader class of convex Lyapunov functions than the KL divergence.
As shown below, this broader family naturally yields alternative generalized gradient flow formulations of the same Gray--Scott dynamics.
For a strictly convex smooth function $h:\mathbb R_{>0}\to\mathbb R$, define
\begin{align}
\phi_{h}(\bm x)
:=
\sum_{i=1}^4
{\hat x}_i
h\qty(\frac{x_i}{{\hat x}_i}).
\end{align}
Then $\phi_h$ is nonincreasing along solutions, and its dissipation vanishes only when the normalized concentrations $x_i/\hat x_i$ coincide, which is attained at the unique equilibrium $\bm x_\mathrm{eq}$ determined by the initial condition, characterized in our general formulation as $\{\bm x_\mathrm{eq}\}=\Stoich(\bm x_0)\cap\partial\phi_h^\ast(\Equib(\partial\phi_h(\bm{\hat x})))$.
On the stoichiometric subspace $\Stoich(\bm x_0)$, the Lyapunov functional $\phi_{h}$ coincides with its associated Bregman divergence $D_{\phi_{h}}(\bm x|\bm{\hat x})$ up to an additive constant, since $\partial_{\bm x}\phi_{h}(\bm{\hat x})=h'(1)\bm 1$ and $\bm 1^\top(\bm x-\bm{x}_0)=0$ on $\Stoich(\bm x_0)$.
% The corresponding thermodynamic force and activity are defined as
% $
% \bm f(\bm x) = S^\top
% \partial_{\bm x}
% D_{\phi_{h}}(\bm x|\bm{\hat x})$ and $
% \bm \omega(\bm x) = \bm j(\bm x) / \bm f(\bm x)$, respectively.
% which yields an analogous generalized gradient-flow formulation to that discussed above
% With these definitions, the Gray--Scott dynamics takes the generalized gradient flow form \cref{eq:generalized gradient flow}, with potential $D_{\phi_h}$ and the quadratic dissipation function associated with the activity $\bm\omega$.
For each choice of $h$, define the corresponding thermodynamic force by
\begin{align}
\bm f_h(\bm x)
:=
S^\top
\partial_{\bm x}
D_{\phi_h}(\bm x\|\hat{\bm x}).
\end{align}
Since
$
\partial_{x_i}
D_{\phi_h}(\bm x\|\hat{\bm x})
=
h'(x_i/\hat x_i)-h'(1),
$
the components of the force are given by
\begin{align}
    \begin{aligned}
       f_{h,1}(\bm x)
        &=
        h'(x_1/\hat x_1)-h'(x_2/\hat x_2),\\
        f_{h,2}(\bm x)
        &=
        h'(x_2/\hat x_2)-h'(x_3/\hat x_3),\\
        f_{h,3}(\bm x)
        &=
        \KSTwo{h'(x_1/\hat x_1)-h'(x_4/\hat x_4)}. 
    \end{aligned}
\end{align}
To formulate the system using the quadratic dissipation function, the corresponding $h$-dependent activities are defined by $\bm\omega_h(\bm x):= {\bm j(\bm x)}/{\bm f_h(\bm x)}$.
Because $h$ is strictly convex, $h'$ is strictly increasing, and hence $\omega_{h,r}(\bm x)>0\, (r = 1,2,3)$. 
For instance, the activity of the first reaction is
\begin{align}
\omega_{h,1}(\bm x) = k_1^+x_{\mathrm{eq},1}x_2^2 \cdot \frac{x_1/\hat x_1- x_2/\hat x_2}{h'(x_1/\hat x_1)-h'(x_2/\hat x_2)} >0.
\end{align}
Consequently, analogously to the KL case, the Gray--Scott dynamics can also be expressed as
\begin{align}
\dot{\bm x}_t
=
-S
\diag\qty(\bm\omega_h(\bm x_t))
\bm f_h(\bm x_t)
=
-S
\partial_{\bm f}
\left.
\qty[
\sum_{r=1}^3
\omega_{h,r}(\bm x_t)
\frac{f_r^2}{2}
]
\right|_{\bm f=\bm f_h(\bm x_t)}.
\end{align}
For instance, for $h_p(w) = w^p/p$ with $p > 1$, the corresponding Bregman divergence,
\begin{align}
    \label{eq:p-norm divergence}
    D_{p}(\bm x|\bm{\hat x}) := D_{\phi_{h_p}}(\bm x|\bm{\hat x}) = \sum_{i=1}^4 {\hat x}_i\qty[\frac{1}{p}\qty(\frac{x_i}{{\hat x}_i})^p-\frac{1}{p}-\qty(\frac{x_i}{{\hat x}_i}-1)],
\end{align} also yields a Lyapunov function.
\end{example}
\end{KSblock}

\section{Convergence Analysis of Relaxation in Generalized Gradient Flows}
\label{sec:convergence analysis of relaxation in generalized gradient flows}
% Our objective in this section is to state the main theorems for generalized gradient flows, present its corollaries for CRNs, and provide their proofs. 
% We first formalize convexity conditions of Bregman divergence (\Cref{convexity conditions of divergences}) and bounds on EPR (\Cref{subsec:assumptions on factorized epr bounds}) required to state the main theorems precisely. The following \Cref{subsec:main results} presents the main theorems, which gives bounds on Bregman divergence in generalized gradient flows. We then show, as corollaries, the explicit bounds that the main theorems yields for CRNs (\Cref{sec:corollaries for equilibrium CRNs}). The final \Cref{sec:proof of main results} provides the proofs of the main theorems.
In this section we develop our main results on generalized gradient flows and their corollaries for CRNs. 
We begin by formalizing the convexity conditions on the Bregman divergence and the factorized bounds on EPR, detailed in \Cref{convexity conditions of divergences,subsec:assumptions on factorized epr bounds} respectively, which we will use throughout. 
\Cref{subsec:main results} then states our main theorems, which provide bounds on the Bregman divergence along generalized gradient flows. 
In \Cref{sec:corollaries for equilibrium CRNs} we derive explicit bounds for equilibrium CRNs as corollaries of these theorems. 
Finally, \Cref{sec:proof of main results} collects selected proofs of the theorems.
% \KS{revised for readability}

\subsection{Convexity Conditions on Divergences}
\label{convexity conditions of divergences}
This section introduces convexity conditions for the Bregman divergence $D_\phi(\bm x\| \bm x_\mathrm{eq})$ in a generalized gradient flow \cref{eq:generalized gradient flow}.
We first define global and local convexity conditions in \Cref{subsubsec:global convexity conditions,subsubsec:local convexity conditions}, inspired by convex optimization \cite{Nesterov2018,Polyak2021,Boyd_Vandenberghe_2004}.
In \Cref{subsubsec:convexity on solution orbit}, we describe convexity on solution orbits of \cref{eq:generalized gradient flow}.

% comment out the redefinition of divergence 
\begin{comment}
    Let $\phi: \X \to \R$ be  a thermodynamic function as defined in \Cref{subsubsec:information-geometric structure between concentration and potential spaces}.
    % and $\X'$ be a non-empty convex subset of $\X$.
    With the second argument fixed at $\bm x^\ast \in \X$, we define the Bregman divergence for $\phi$ as
    \begin{align}
        D_\phi(\cdot \| \bm x^\ast): \X \to \R_{\geq 0}, \bm x \mapsto D_\phi(\bm x \| \bm x^\ast).
    \end{align}
    % We simply refer to $D_\phi(\cdot \| \bm x^\ast)$ as a Bregman divergence if the fixed point $\bm x^\ast$ is clear from the context. \KS{Even if one argument is fixed, it is still called Divergence. }
    Note that $\bm x^\ast$ is the minimizer of $D_\phi(\cdot \| \bm x^\ast)$ on $\X$. 
    The divergence $D_\phi(\cdot \| \bm x^\ast)$ is also strictly convex because it shares the same Hessian as $\phi$.
    % Let $\X' \subset \X$ be a bounded closed subset of $\R^N$  
\end{comment}

\subsubsection{Global Convexity Conditions}
\label{subsubsec:global convexity conditions}
We start with two usual definitions of convexity conditions  for $D_\phi(\cdot \| \bm x^\ast)$: the Polyak–Łojasiewicz (PL) condition and the Lipschitz‑continuity (LC) condition.

For a constant $m_{\bm x^\ast} \in \R_{>0}$, $D_\phi(\cdot \| \bm x^\ast)$ is called \textit{$m_{\bm x^\ast}$-Polyak-Lojasiewicz} (\textit{$m_{\bm x^\ast}$-PL}) on $\X' \subseteq \X$ if and only if, for any $\bm x\in \X'$, the following inequality holds: 
\begin{align}
    \label{eq:Polyak-Lojasiewicz}
    D_\phi(\bm x \| \bm x^\ast) \leq \frac{1}{2m_{\bm x^\ast}} \norm{\partial_{\bm x} D_\phi(\bm x \| \bm x^\ast)}^2_2.
\end{align}
The equality of \cref{eq:Polyak-Lojasiewicz} holds iff $\bm x= \bm x^\ast \in \X'$.
A Bregman divergence $D_\phi(\cdot \| \bm x^\ast)$ is called \textit{$L_{\bm x^\ast}$-Lipschitz continuous} (\textit{$L_{\bm x^\ast}$-LC}) on $\X' \subseteq \X$ for $L_{\bm x^\ast} \in \R_{>0}$ if and only if, for any $\bm x\in \X'$, the following inequality holds:
\begin{align}
    \label{eq:Lipschitz continuity}
    D_\phi(\bm x \| \bm x^\ast) \geq \frac{1}{2L_{\bm x^\ast}} \norm{ \partial_{\bm x} D_\phi(\bm x\| \bm x^\ast)}_2^2.
\end{align}
The equality of \cref{eq:Lipschitz continuity} also holds iff $\bm x= \bm x^\ast \in \X'$.
In this paper, the conditions of PL and LC are referred to as ''\textit{global convexity assumptions}''\footnote{We refer to them as ''convexity conditions'' because, especially when $\X'$ is convex, they are related to conditions stronger than convexity. Informally, \textit{strong convexity}, the condition that $f(\bm x) - \frac{\mu}{2}\norm{\bm x}^{2}$ be convex, implies convexity of $f$, and moreover implies the $\mu$-PL inequality of $f$. In addition, \textit{LC} is equivalent to the convexity of $\frac{L}{2}\norm{\bm x}^{2}-f(\bm x)$. For informal statements, see, e.g., \cite[Section 2.1.3, 3.2.6]{Nesterov2018}, \cite[Section 2.]{Polyak2021}, or \cite{Boyd_Vandenberghe_2004}. 
% \KS{[TO DO: cites on local convexity??]}
}. 
% connection to strong convexity on convex subset

\subsubsection{Local Convexity Conditions}
\label{subsubsec:local convexity conditions}
We extend the global assumptions on $D_\phi(\cdot\| \bm x^\ast)$ to the local ones, whose convexity parameters are dependent on the state $\bm x \in \X$.
For a continuous positive function on $\X' \subseteq \X$, $m_{\bm x^\ast} : \X' \to \R_{>0}$, $D_\phi(\cdot \| \bm x^\ast)$ is called \textit{$m_{\bm x^\ast}(\cdot)$-locally Polyak-Lojasiewicz} (\textit{$m_{\bm x^\ast}(\cdot)$-locally PL}) on $\X'$ if, for any $\bm x \in \X'$, the following condition holds:
\begin{align}
    \label{eq:local Polyak-Lojasiewicz}
    D_\phi(\bm x \| \bm x^\ast) \leq \frac{1}{2m_{\bm x^\ast}(\bm x)} \norm{\partial_{\bm x} D_\phi(\bm x\| \bm x^\ast)}_2^2.
\end{align}
Similarly, for a continuous positive function on $\X' \subseteq \X$, $L_{\bm x^\ast} : \X' \to \R_{>0}$,  $D_\phi(\cdot \| \bm x^\ast)$ is called \textit{$L_{\bm x^\ast}(\cdot)$-locally Lipchitz continuous} (\textit{$L_{\bm x^\ast}(\cdot)$-locally LC}) on $\X'$ if, for any $\bm x \in \X'$, the following condition holds:
\begin{align}
    \label{eq:local Lipschitz continuity}
    D_\phi(\bm x \| \bm x^\ast) \geq \frac{1}{2L_{\bm x^\ast}(\bm x)} \norm{\partial_{\bm x} D_\phi(\bm x\| \bm x^\ast)}_2^2.
\end{align}
The two conditions \cref{eq:local Polyak-Lojasiewicz,eq:local Lipschitz continuity} are referred to as "\textit{local convexity assumptions}".

\begin{remark}[Local Convexity implies Global Convexity]
For a $m_{\bm x^\ast}(\cdot)$-locally PL, $L_{\bm x^\ast}(\cdot)$-locally LC $D_\phi(\cdot \| \bm x^\ast)$ on $\X'$, suppose $\low{\mu}{}_{\bm x^\ast}\leq  m_{\bm x^\ast}(\bm x)$ and $\up{L}_{\bm x^\ast} \geq L_{\bm x^\ast}(\bm x)$ hold for $\bm x \in \X'$.  Then, $D_\phi(\cdot \| \bm x^\ast)$ is $\low{\mu}{}_{\bm x^\ast}$-PL and $\up{L}{}_{\bm x^\ast}$-LC because, for any $\bm x \in \X'$, the following inequalities hold:
\begin{align}
    \label{eq:locally PL is PL}
    D_\phi(\bm x \| \bm x^\ast)
    \leq \frac{1}{2m_{\bm x^\ast}(\bm x)} \norm{\partial_{\bm x} D_\phi(\bm x\| \bm x^\ast)}_2^2
    \leq \frac{1}{2\low{\mu}{}_{\bm x^\ast}} \norm{\partial_{\bm x} D_\phi(\bm x\| \bm x^\ast)}_2^2,\\
    \label{eq:locally LC is LC}
    D_\phi(\bm x \| \bm x^\ast) 
    \geq \frac{1}{2L_{\bm x^\ast}(\bm x)} \norm{\partial_{\bm x} D_\phi(\bm x\| \bm x^\ast)}_2^2
    \geq \frac{1}{2\up{L}{}_{\bm x^\ast}} \norm{\partial_{\bm x} D_\phi(\bm x\| \bm x^\ast)}_2^2.
\end{align}
For $\X'$ that does not include the fixed point $\bm x^\ast$, the divergence is locally PL and locally LC with respect to $m_{\bm x^\ast}(\bm x) = L_{\bm x^\ast}(\bm x) = \rho_{\bm x^\ast}^\phi(\bm x):= \frac{\norm{\partial_{\bm x} D_\phi(\bm x\| \bm x^\ast)}_2^2}{2 D_\phi(\bm x\| \bm x^\ast)}$.
The first equalities in \cref{eq:locally PL is PL,eq:locally LC is LC} hold iff $m_{\bm x^\ast}(\bm x) = L_{\bm x^\ast}(\bm x) = \rho_{\bm x^\ast}^\phi(\bm x)$.
For a bounded and closed $\X'$, there exists $\min_{\bm x \in \X'} m_{\bm x^\ast}(\bm x) \in \X'$ and $\max_{\bm x \in \X'} L_{\bm y}(\bm x) \in \X'$, therefore {$D_\phi(\cdot \| \bm x^\ast)$} is PL and LC for the minimum and maximum, respectively.
\end{remark}

\subsubsection{Convexity on Solution Orbits}
\label{subsubsec:convexity on solution orbit}
% topology of trajectory set
% local convexity on finite orbits, and problems
% towards T-dependent parameters for \rho
% lemma
Here, we\KSsout{ define solution orbits in generalized gradient flows and recall their properties. We then} describe the convexity of the Bregman divergence on the solution orbits.

We focus on the most stringent parameter $\rho_{\bm x_\mathrm{eq}}^\phi(\cdot)$. When $\bm x_0 \neq \bm x_\mathrm{eq}$, the Bregman divergence $D_\phi(\cdot \| \bm x_\mathrm{eq})$ is $\rho^\phi_{\bm x_\mathrm{eq}}(\cdot)$-locally PL and $\rho^\phi_{\bm x_\mathrm{eq}}(\cdot)$-locally LC on the finite orbit $\mathcal{O}_{T}(\bm x_0)$. 
On bounded and closed $\mathcal{O}_{T}(\bm x_0)$, the function $\rho_{\bm x_\mathrm{eq}}(\cdot)$ attains its maximum and minimum, although they depend on the time $T$. 
A naive remedy is to take the maximum and minimum over $\mathcal{O}_{\infty}(\bm x_0)$; however, this can fail, because $\rho_{\bm x_\mathrm{eq}}(\bm x_t)$ may not converge to a constant as $\bm x_t \to \bm x_\mathrm{eq}$ and be extendable to a continuous function $\rho_{\bm x_\mathrm{eq}}(\cdot)$ on $\mathcal{O}_{\infty}(\bm x_0)$. 
Nevertheless, for $\rho^\phi_{\bm x_\mathrm{eq}}(\cdot)$ one can choose global convexity parameters that are independent of $T$ but depend on the initial condition $\bm x_0$:

% lemma: there exists global convexity constants for D_\phi on O(\bm x_0) (from 0 to infty)
\begin{lemma}[Global Convexity of Bregman Divergence on Solution Orbits]
    \label{lemma:global convexity of Bregman divergence on solution orbit}
    Suppose that \Cref{asm:invariance of positive orthant,asm:convergence to the positive equilibrium} holds for a generalized gradient flow \cref{eq:generalized gradient flow}. Then there exist positive constants $\low{\rho_{\bm x_\mathrm{eq}}}(\bm x_0), \up{\rho_{\bm x_\mathrm{eq}}}(\bm x_0) > 0$ such that, $\low{\rho_{\bm x_\mathrm{eq}}}(\bm x_0) \leq {\rho_{\bm x_\mathrm{eq}}^\phi}(\bm x) \leq \up{\rho_{\bm x_\mathrm{eq}}}(\bm x_0)$ in $\mathcal{O}_{\infty}(\bm x_0) \setminus \qty{\bm x_\mathrm{eq}}$, and $D_\phi(\cdot \| \bm x_\mathrm{eq})$ is $\low{\rho_{\bm x_\mathrm{eq}}}(\bm x_0)$-PL and $\up{\rho_{\bm x_\mathrm{eq}}}(\bm x_0)$-LC  on $\mathcal{O}_{\infty}(\bm x_0)$.
    Particularly, $D_\phi(\cdot \| \bm x_\mathrm{eq})$ is $\low{\rho_{\bm x_\mathrm{eq}}}(\bm x_0)$-PL and $\up{\rho_{\bm x_\mathrm{eq}}}(\bm x_0)$-LC  on $\mathcal{O}_{T}(\bm x_0)$ for $T\geq 0$.
\end{lemma}
\begin{proof}[Proof of \Cref{lemma:global convexity of Bregman divergence on solution orbit}]
    We have $D_\phi(\bm x \| \bm x_\mathrm{eq}) = D_{\phi^\ast}(\partial \phi(\bm x_\mathrm{eq}) \| \partial \phi(\bm x))$ for any $\bm x \in \X$ by the property of Bregman divergence. 
    Denote $\partial \phi(\bm x), \partial \phi(\bm x_\mathrm{eq})$ as $\bm \mu, \bm \mu_\mathrm{eq}$ for simplicity. We have the integral representation of Bregman divergence,
    \begin{align}
        D_{\phi^\ast}( \bm \mu_\mathrm{eq} \|  \bm \mu) = \int_0^1 \qty{s \cdot ( \bm \mu - \bm \mu_\mathrm{eq})^\top  \partial^2\phi^\ast\qty[\bm \mu_\mathrm{eq} + s( \bm \mu - \bm \mu_\mathrm{eq})] ( \bm \mu - \bm \mu_\mathrm{eq})}  \, \dd s,
    \end{align}
    by Taylor’s theorem for multi-variable functions \cite[Theorem 2.68]{Folland2001}, or by calculating the integral on the right-hand side.
    Using the minimum and maximum eigenvalues $\lambda_{\min}(\partial^2\phi^\ast\qty(\bm{\tilde \mu}))$ and $\lambda_{\max}(\partial^2\phi^\ast\qty(\bm{\tilde \mu}))$ of the Hessian $\partial^2\phi^\ast(\bm{\tilde \mu})$, we have $\lambda_{\min}(\partial^2\phi^\ast\qty(\bm{\tilde \mu})) \cdot I_N \preceq  \partial^2\phi^\ast\qty(\bm{\tilde \mu}) \preceq \lambda_{\max}(\partial^2\phi^\ast\qty(\bm{\tilde \mu})) \cdot I_N$ for any $\bm{\tilde \mu} \in \M$.
    Note that $\lambda_{\min}(\partial^2\phi^\ast\qty(\bm{\tilde \mu}))$ and $\lambda_{\max}(\partial^2\phi^\ast\qty(\bm{\tilde \mu}))$ are positive by the strictly convexity of $\phi^\ast$, respectively.

    % Define a map $\bm M_{\bm x_\mathrm{eq}}(\bm x, s)$ from $\X \times [0,1]$ to $\M$ as $\bm M_{\bm x_\mathrm{eq}}(\bm x, s) := \partial \phi(\bm x) + s \cdot (\partial \phi(\bm x_\mathrm{eq}) - \partial \phi(\bm x))$. The set $\mathcal{O}_{\infty}(\bm x_0) \times [0,1]$ is compact with the product topology, thus its image,
    % \begin{align}
    %     \M(\bm x_0, \bm x_\mathrm{eq}) :=&\, \bm M_{\bm x_\mathrm{eq}}(\mathcal{O}_{\infty}(\bm x_0) \times [0,1]) \\=&\, \qty{\partial \phi(\bm x) + s(\partial \phi(\bm x_\mathrm{eq}) - \partial \phi(\bm x)) \mid \bm x\in \mathcal{O}_{\infty}(\bm x_0), s\in [0,1]},
    % \end{align} is also compact (bounded and closed) in $\M$.

    Since the map $\X \times [0,1] \to \M$, 
    $
        (\bm x,s) \mapsto 
        \partial \phi(\bm x) + s\bigl(\partial \phi(\bm x_\mathrm{eq}) - \partial \phi(\bm x)\bigr)
    $
    is continuous and a set $\mathcal{O}_{\infty}(\bm x_0) \times [0,1]$ is compact with the product topology, the image     $
       \M(\bm x_0, \bm x_\mathrm{eq}) := \qty{(1-s)\partial_{\bm x} \phi(\bm x) + s \partial_{\bm x} \phi(\bm x_\mathrm{eq})  \mid \bm x\in \mathcal{O}_{\infty}(\bm x_0), s\in [0,1]}
        \subseteq \M
    $ is also compact (hence bounded and closed) in $\M$.    
    Taking the minimum and maximum over $\bm{\tilde \mu} \in \M(\bm x_0, \bm x_\mathrm{eq})$, we obtain the following inequality for all $\bm x \in \mathcal{O}_{\infty}(\bm x_0)$: 
    \begin{align}
        \frac{\underset{\bm{\tilde\mu} }{\min}\, \lambda_{\min}(\partial^2\phi^\ast(\bm{\tilde\mu}))}{2} \norm{\bm \mu - \bm \mu_\mathrm{eq}}^2_2 \leq D_{\phi^\ast}( \bm \mu_\mathrm{eq} \|  \bm \mu) \leq \frac{\underset{\bm{\tilde\mu}}{\max}\, \lambda_{\max}(\partial^2\phi^\ast(\bm{\tilde\mu}))}{2} \norm{\bm \mu - \bm \mu_\mathrm{eq}}^2_2.
    \end{align}
    % Here, we consider the minimization of $\bm{\tilde \mu}$ over $\M(\bm x_0, \bm x_\mathrm{eq})$. \KS{to rewrite}
    If $\bm x \neq \bm x_\mathrm{eq}$ and $\bm x\in \mathcal{O}_{\infty}(\bm x_0)$, as $\rho_{\bm x_\mathrm{eq}}^\phi(\bm x) = \frac{\norm{\bm \mu - \bm \mu_\mathrm{eq}}^2_2}{2D_{\phi^\ast}( \bm \mu_\mathrm{eq} \|  \bm \mu)}$, we have 
    \begin{align}
        \label{eq:global convexity parameters of strict case}
        \qty(\underset{\bm{\tilde\mu}}{\max}\, \lambda_{\max} (\partial^2\phi^\ast(\bm{\tilde\mu})))^{-1}\leq \rho_{\bm x_\mathrm{eq}}^\phi(\bm x) \leq \qty( \underset{\bm{\tilde\mu}}{\min}\, \lambda_{\min}(\partial^2\phi^\ast(\bm{\tilde\mu})) )^{-1}.
    \end{align}
    For the leftmost and rightmost constants in \cref{eq:global convexity parameters of strict case}, $D_\phi(\cdot \| \bm x_\mathrm{eq})$ is PL and LC on $\mathcal{O}_{\infty}(\bm x_0)$, respectively; therefore it is PL and LC on $\mathcal{O}_{T}(\bm x_0) \subseteq \mathcal{O}_{\infty}(\bm x_0)$ as well.
\end{proof}

\subsection{Assumptions on Factorized EPR Bounds}
\label{subsec:assumptions on factorized epr bounds}
We formulate assumptions on bounding EPRs, which is required to state the main theorems. 
Informally, the assumptions say that the EPR can be bounded by the product of two factors: one depending only on concentrations and {the other} on the force norm.
We state the assumptions abstractly to cover the wider class of systems, generalized gradient flows.

% properties of EPRs
We note that the EPR function $\dot \Sigma(\bm x, \bm f) := \expval{\partial_{\bm f} \Psi_{\bm x}^\ast(\bm f), \bm f}$ for separable $\Psi_{\bm x}^\ast$ only depends on $\bm x$ and $\abs{\bm f}$ as $\expval{\partial_{\bm f} \Psi_{\bm x}^\ast(\bm f), \bm f} = \expval{\partial_{\bm f} \Psi_{\bm x}^\ast(\abs{\bm f}), \abs{\bm f}}$. 
One can obtain this equality from the definition of the dissipation function \cref{eq:dissipation function}.
For the equilibrium subspace $\F^\mathrm{eq}$, {the space of force‑norm values} is defined as $\norm{\F^\mathrm{eq}}:= \qty{\norm{\bm f}_2 \mid \bm f \in \F^\mathrm{eq}} = \R_{\geq 0}$.

% assupmtions on EPR bounds
We state assumptions on factorized bounds of EPRs in generalized gradient flows:
% \begin{assumption}[Factorized Lower Bound of EPR]
%     \label{asm: activity-force decomposability of EPR lower bound}
%     There exists a nonnegative function $\low{\B} :\norm{\F^\mathrm{eq}} \to \mathbb{R}_{\geq 0}$, which is positive except for at $0$, $\low{\B}(0) = 0$, and a positive function $\low{\omega}: \X \to \R_{>0}$ such that
%     $
%         \low{\omega}(\bm x) \low{\B}(\norm{\bm f(\bm x)}_2) \leq 
%         % \expval{ \partial_{\bm{f}} \Psi_{\bm{x}}^\ast(\bm{f}(\bm x)), \bm f(\bm x)} 
%         \dot \Sigma(\bm x, \bm f) \textrm{ for any } \bm x \in \X$.
% \end{assumption}
% \begin{assumption}[Factorized Upper Bound of EPR]
%     \label{asm: activity-force decomposability of EPR upper bound}
%     There exists a nonnegative function $\up{\B} :\norm{\F^\mathrm{eq}} \to \mathbb{R}_{\geq 0}$, which is positive except for at $0$, $\up{\B}(0) = 0$, and a positive function $\up{\omega}: \X \to \R_{>0}$ such that
%         % \expval{ \partial_{\bm{f}} \Psi_{\bm{x}}^\ast(\bm{f}(\bm x)), \bm{f}(\bm x)} 
%         $\dot \Sigma(\bm x, \bm f) 
%         \leq \up{\omega}(\bm x) \up{\B}(\norm{\bm f(\bm x)}_2) \textrm{ for any } \bm x \in \X$.
% \end{assumption}
\begin{assumption}[Factorized Bounds of EPR]
    \label{asm: activity-force decomposability of EPR bounds}
    \begin{subasms}
        \item
        \label[assumption]{asm: activity-force decomposability of EPR lower bound}
        % \label{asm: activity-force decomposability of EPR lower bound}
        There exists a nonnegative function $\low{\B} :\norm{\F^\mathrm{eq}} \to \mathbb{R}_{\geq 0}$, which is positive except for at $0$, $\low{\B}(0) = 0$, and a positive function $\low{\omega}: \X \to \R_{>0}$ such that
        $\dot \Sigma(\bm x, \bm f) \geq \low{\omega}(\bm x) \low{\B}(\norm{\bm f(\bm x)}_2)$ for any $\bm x \in \X$,
        \item
        \label[assumption]{asm: activity-force decomposability of EPR upper bound}
        % \label{asm: activity-force decomposability of EPR upper bound}
        There exists a nonnegative function $\up{\B} :\norm{\F^\mathrm{eq}} \to \mathbb{R}_{\geq 0}$, which is positive except for at $0$, $\up{\B}(0) = 0$, and a positive function $\up{\omega}: \X \to \R_{>0}$ such that
        $\dot \Sigma(\bm x, \bm f) \leq \up{\omega}(\bm x) \up{\B}(\norm{\bm f(\bm x)}_2)$ for any $\bm x \in \X$. 
    \end{subasms}
\end{assumption}
% \KS{Make lemmas 3.2 and 3.3 into one (U/L)}
% For a separable dissipation function \cref{eq:separable}, \Cref{asm: activity-force decomposability of EPR upper bound} is satisfied as follows
% Furthermore, if the derivative of the scalar dissipation function, $(\psi^\ast)'$, is convex, then \Cref{asm: activity-force decomposability of EPR lower bound} is also satisfied:
For a separable dissipation function \cref{eq:separable}, \Cref{asm: activity-force decomposability of EPR upper bound} is satisfied; moreover, if the derivative of the scalar dissipation function, $(\psi^\ast)'$, is convex, then \Cref{asm: activity-force decomposability of EPR lower bound} is also satisfied:
% lemmas for separable dissipation functions
\begin{lemma}[Bounds of separable EPR by force norm]
    \label{lemma:bounds of separable EPR}
    \begin{sublems}
        \item \label[lemma]{lemma:upper bounds of separable EPR}
        For a separable dissipation function \cref{eq:dissipation function}, the EPR at $\bm x \in \X, \bm f \in \F^\mathrm{eq}$ is bounded from above as
        \begin{align}
            \label{eq:upper bounds of separable EPR}
            \dot\Sigma(\bm x, \bm f) 
            % = \expval{\bm \omega (\bm x) \circ (\psi^\ast)' (\bm f), \bm f} 
            \leq  \| \bm \omega (\bm x) \|_{2}\norm{\bm f}_2 (\psi^\ast)' (\norm{\bm f}_2).
        \end{align}
        For $M\geq 2$, the equality holds iff $\bm f = \bm 0$. For $M=1$, the equality always holds.
        That is, for a separable dissipation function, \Cref{asm: activity-force decomposability of EPR upper bound} is satisfied by $\up{\omega}(\bm x) = \| \bm \omega (\bm x) \|_{2}$ and $\up{\B}(\abs{f}) = \up{\B}_{\psi^\ast}(\abs{f}) := \abs{f} \cdot \psi^\ast(\abs{f})$.
        % for $\abs{f} \in \norm{\F^\mathrm{eq}}$.
        \item \label[lemma]{lemma:lower bounds of separable EPR}
        For a separable dissipation function \cref{eq:dissipation function} with convex $(\psi^\ast)'$, the EPR at $\bm x \in \X, \bm f \in \F^\mathrm{eq}$ is bounded from below as
        \begin{align}
            \label{eq:lower bounds of separable EPR}
            \norm{\bm \omega (\bm x) }_{-\infty} \norm{\bm f}_2 (\psi^\ast)' \qty(\frac{\norm{\bm f}_2}{\sqrt{M}}) \leq 
            % \expval{\bm \omega (\bm x) \circ (\psi^\ast)' (\bm f), \bm f} = 
            \dot\Sigma(\bm x, \bm f).
        \end{align}
        For $M\geq 2$, the equality holds iff $\bm f = \bm 0$. For $M=1$, the equality always holds.
        That is, for a separable dissipation function with convex $(\psi^\ast)'$, \Cref{asm: activity-force decomposability of EPR lower bound} is satisfied by $\low{\omega}(\bm x) = \| \bm \omega (\bm x) \|_{-\infty}$ and $\low{\B}(\abs{f}) = \low{\B}_{\psi^\ast}(\abs{f}) := \abs{f} \cdot \psi^\ast(\abs{f}/\sqrt{M})$.
        % for $\abs{f} \in \norm{\F^\mathrm{eq}}$.
    \end{sublems}
\end{lemma}
% lemmas for separable dissipation functions
% \begin{lemma}[Upper Bound of separable EPR by force norm]
%     \label{lemma:upper bounds of separable EPR}
%     For a separable dissipation function \cref{eq:dissipation function}, the EPR at $\bm x \in \X, \bm f \in \F^\mathrm{eq}$ is bounded from above as
%     \begin{align}
%         \label{eq:upper bounds of separable EPR}
%         \dot\Sigma(\bm x, \bm f) = \expval{\bm \omega (\bm x) \circ (\psi^\ast)' (\bm f), \bm f} \leq  \| \bm \omega (\bm x) \|_{2}\norm{\bm f}_2 (\psi^\ast)' (\norm{\bm f}_2).
%     \end{align}
%     For $M\geq 2$, the equality holds iff $\bm f = \bm 0$. For $M=1$, the equality always holds.
%     That is, for a separable dissipation function, \Cref{asm: activity-force decomposability of EPR upper bound} is satisfied by $\up{\omega}(\bm x) = \| \bm \omega (\bm x) \|_{2}$ and $\up{\B}(\abs{f}) = \up{\B}_{\psi^\ast}(\abs{f}) := \abs{f} \cdot \psi^\ast(\abs{f})$ for $\abs{f} \in \norm{\F^\mathrm{eq}}$.
% \end{lemma}
\begin{proof}[Proof of \Cref{lemma:lower bounds of separable EPR}]
    For a scalar dissipation function $\psi^\ast$, by the definition of dissipation function \cref{eq:dissipation function}, $(\psi^\ast)'$ is increasing, odd, and $(\psi^\ast)'(0) = 0$. Thus, $f \cdot (\psi^\ast)'(f) = \abs{f} \cdot (\psi^\ast)'(\abs{f})$ holds for $f \in \R$.
    The upper bound is obtained as
    \begin{align}
        \expval{\bm{\omega}(\bm x) \circ(\psi^\ast)'(\abs{\bm f}), \abs{\bm f}}
        &= \sum_{r=1}^M \omega_r(\bm x) \cdot \abs{f_r} \cdot(\psi^\ast)'(\abs{f_r})\\
        \label{eq:first inequality of upper bounds of separable EPR}
        &\leq \sum_{r=1}^M \omega_r(\bm x)  \abs{f_r} (\psi^\ast)'(\norm{\bm f}_2)\\
        \label{eq:second inequality of upper bounds of separable EPR}
        &\leq \norm{\bm \omega (\bm x)}_2 \cdot \norm{\bm f}_2 (\psi^\ast)' (\norm{\bm f}_2).
    \end{align}
    The first inequality \cref{eq:first inequality of upper bounds of separable EPR} hold because $\abs{f_r} \leq \norm{\bm f}_2$ and the monotonicity of $(\psi^\ast)'$. The equality holds if and only if, for any $r$, $|f_r| = \norm{\bm f}_2$.
    For $M\geq 2$, the equality holds if and only if $\bm f = \bm 0$.
    For $M = 1$, the equality always holds.
    The second inequality \cref{eq:second inequality of upper bounds of separable EPR} follows from Cauchy-Schwartz inequality, and equality holds if and only if $\bm \omega (\bm x)$ and $\abs{\bm f}$ are parallel.
\end{proof}
% \begin{lemma}[Lower Bound of separable EPR by force norm] 
%     \label{lemma:lower bounds of separable EPR}
%     For a separable dissipation function \cref{eq:dissipation function} with convex $(\psi^\ast)'$, the EPR at $\bm x \in \X, \bm f \in \F^\mathrm{eq}$ is bounded from below as
%     \begin{align}
%         \label{eq:lower bounds of separable EPR}
%         \norm{\bm \omega (\bm x) }_{-\infty} \norm{\bm f}_2 (\psi^\ast)' \qty(\frac{\norm{\bm f}_2}{\sqrt{M}}) \leq \expval{\bm \omega (\bm x) \circ (\psi^\ast)' (\bm f), \bm f} = \dot\Sigma(\bm x, \bm f).
%     \end{align}
%     For $M\geq 2$, the equality holds iff $\bm f = \bm 0$. For $M=1$, the equality always holds.
%     That is, for a separable dissipation function with convex $(\psi^\ast)'$, \Cref{asm: activity-force decomposability of EPR lower bound} is satisfied by $\low{\omega}(\bm x) = \| \bm \omega (\bm x) \|_{-\infty}$ and $\low{\B}(\abs{f}) = \low{\B}_{\psi^\ast}(\abs{f}) := \abs{f} \cdot \psi^\ast(\abs{f}/\sqrt{M})$  for $\abs{f} \in \norm{\F^\mathrm{eq}}$.
% \end{lemma}
\begin{proof}[Proof of \Cref{lemma:upper bounds of separable EPR}] 
    We assume $\bm f \neq \bm 0$ since the equality in \cref{eq:lower bounds of separable EPR} holds for $\bm f = \bm 0$.
    The lower bound of the EPR is obtained as
    \begin{align}
        \expval{\bm \omega(\bm x) \circ (\psi^\ast)'(\bm f), \bm f} 
        &= \sum_{r=1}^M {\omega_r(\bm x)\abs{f_r}}(\psi^\ast)' (\abs{f_r})\\
        \label{eq:first inequality of lower bounds of separable EPR}
        &\geq \norm{\bm \omega (\bm x) }_{-\infty} \sum_{r=1}^M {\abs{f_r}}(\psi^\ast)' (\abs{f_r})\\
        &= \norm{\bm \omega (\bm x) }_{-\infty} \norm{\bm f}_1 \sum_{r=1}^M \frac{\abs{f_r}}{\norm{\bm f}_1} (\psi^\ast)' (\abs{f_r})\\
        \label{eq:second inequality of lower bounds of separable EPR}
        &\geq \norm{\bm \omega (\bm x) }_{-\infty}\norm{\bm f}_1 (\psi^\ast)' \qty(\frac{\sum_{r = 1}^M \abs{f_r}^2}{\norm{\bm f}_1})\\
        \label{eq:third inequality of lower bounds of separable EPR}
        &\geq \norm{\bm \omega (\bm x) }_{-\infty}\norm{\bm f}_2 (\psi^\ast)' \qty(\frac{\norm{\bm f}_2^2}{\norm{\bm f}_1})\\
        \label{eq:fourth inequality of lower bounds of separable EPR}
        &\geq \norm{\bm \omega (\bm x) }_{-\infty} \norm{\bm f}_2 (\psi^\ast)' \qty(\frac{\norm{\bm f}_2}{\sqrt{M}}).
    \end{align}
    The first inequality \cref{eq:first inequality of lower bounds of separable EPR} follows from $\omega_r(\bm x) \geq \norm{\bm \omega(\bm x)}_{-\infty}$.
    % , and equality hold iff $\omega_r(\bm x)$ takes the same value for every $r$.
    The inequality \cref{eq:second inequality of lower bounds of separable EPR} follows from the convexity of $(\psi^\ast)'$ and equality holds iff $\abs{f_r}$ takes the same value for every $r$.
    The third inequality \cref{eq:third inequality of lower bounds of separable EPR} holds iff $\norm{\bm f}_1 = \norm{\bm f}_2$, i.e. $\bm f$ has exactly one non-zero component.
    The fourth inequality \cref{eq:third inequality of lower bounds of separable EPR} holds by Cauchy-Schwartz inequality: $\norm{\bm f}_1 = \expval{\bm 1, \abs{\bm f}} \leq \sqrt{M} \norm{\bm f}_2$.
    For $M\geq 2$, These equalities \cref{eq:first inequality of lower bounds of separable EPR,eq:second inequality of lower bounds of separable EPR,eq:third inequality of lower bounds of separable EPR,eq:fourth inequality of lower bounds of separable EPR} can not hold simultaneously for any $\bm f\neq \bm 0$. For $M=1$ they always hold.
\end{proof}

% additional assumptions to prove theorem 2
We additionally formulate stronger assumptions on EPR bounds, used in the proofs of \Cref{thm:lower bound under locally LC assumption,thm:upper bound under locally PL assumption}.
\begin{assumption}[{Convexity of Bounding Functions}]
    \label{asm:convexity of bound functions}
    \begin{subasms}
        \item
        \label[assumption]{asm: convexity of EPR lower bound}
        % \label{asm: activity-force decomposability of EPR lower bound}
        \Cref{asm: activity-force decomposability of EPR lower bound} is satisfied and $\low{\B} (\sqrt{\cdot})$ is convex,
        \item
        \label[assumption]{asm: convexity of EPR upper bound}
        % \label{asm: activity-force decomposability of EPR upper bound}
        \Cref{asm: activity-force decomposability of EPR upper bound} is satisfied and $\up{\B} (\sqrt{\cdot})$ is convex. 
    \end{subasms}
\end{assumption}

For a separable dissipation function, the convexity of $(\psi^\ast)'$ also implies \Cref{asm: convexity of EPR lower bound,asm: convexity of EPR upper bound}; the proof is deferred to {\Cref{sec:proof for lemma:sufficient condition to convexity of bounds}}:
\begin{lemma}[Sufficient Condition for \Cref{asm: convexity of EPR lower bound,asm: convexity of EPR upper bound} under Separable Dissipation]
    \label{lemma:sufficient condition to convexity of bounds}
    For a separable dissipation function \cref{eq:dissipation function} with convex $(\psi^\ast)'$, and for functions $\up{\B}_{\psi^\ast}(\abs{f}) = \abs{f} \cdot \psi^\ast(\abs{f})$, $\low{\B}_{\psi^\ast}(\abs{f}) = \abs{f} \cdot \psi^\ast(\abs{f}/\sqrt{M})$, the compositions $\up{\B}_{\psi^\ast} (\sqrt{\cdot})$ and $\low{\B}_{\psi^\ast} (\sqrt{\cdot})$ are convex.
    In other words, $\up{\B} = \up{\B}_{\psi^\ast}$ and $\low{\B} = \low{\B}_{\psi^\ast}$ in \Cref{lemma:lower bounds of separable EPR,lemma:upper bounds of separable EPR} satisfy \Cref{asm: convexity of EPR lower bound,asm: convexity of EPR upper bound}, respectively.
\end{lemma}
In equilibrium CRNs with quadratic and hyperbolic dissipation, the derivatives of the scalar dissipation functions, $\psi_\mathrm{quad}^\ast(f) = f$ and $\psi_\mathrm{hyp}^\ast(f) = 2 \sinh\qty(\frac{f}{2})$, are both convex. Therefore, by \Cref{lemma:sufficient condition to convexity of bounds}, both functions satisfy \Cref{asm: convexity of EPR lower bound,asm: convexity of EPR upper bound}.

\subsection{Main results}
\label{subsec:main results}
We state our main theorems in this section. 
Under a global convexity assumption, the first pair of theorems establishes upper and lower bounds for the Bregman divergence in generalized gradient flows, whereas the second pair provides bounds under stronger local convexity assumptions.
These results give the first quantitative estimates for Bregman divergences in generalized gradient flows.
The proofs of \Cref{thm:upper bound of Bregman divergence in generalized equilibirum flow} are given in \Cref{sec:proof of main results}, while those of \Cref{thm:lower bound of Bregman divergence in generalized equilibirum flow,thm:bounds under locally convexity assumption} are provided in {\Cref{sec:proof for thm:lower bound of Bregman divergence in generalized equilibirum flow,sec:proof for thm:upper bound under locally PL assumption,sec:proof for thm:lower bound under locally LC assumption}}.

We prove that the bounds and convergence rates of the Bregman divergence $D_\phi(\bm x_T \|\bm x_\mathrm{eq})$ are represented by deformed exponential functions \cite{NAUDTS2002323,NAUDTS200432,Naudts2024}, 
% $\low{\B} (\sqrt{\cdot})$ and $\up{\B} (\sqrt{\cdot})$
whose properties are investigated in the field of generalized thermostatics.

Let $\psi: \R_{> 0} \to \R_{> 0}$ be a positive function.
The \textit{deformed logarithm} $\ln_\psi$ for $\psi$ is defined as 
\begin{align}
    \ln_\psi(u) := \int_1^u \frac{\dd u'}{\psi(u')},\, u> 0.
\end{align}
The \textit{deformed exponential function} $\exp_\psi$ is defined as the inverse of $\ln_\psi$.
The conventional logarithmic and exponential functions are recovered when $\psi(u) = u$.
Using the deformed exponential functions with respect to $\low{\B}$ and $\up{\B}$, we can quantify the upper and lower bounds of $D_\phi(\bm x_T \| \bm x_\mathrm{eq})$, respectively.

\subsubsection{Bounds of Bregman Divergence under Global Convexity Assumptions}
% Firstly, w
We state the main theorems on the upper and lower bound for the Bregman divergence under the global convexity assumptions (\Cref{thm:bounds of Bregman divergence in generalized equilibirum flow}). 
Our bounds are determined by the largest and smallest singular values $\sigma_1, \sigma_{\rk(S)}$ of the stoichiometric matrix, the global convexity parameters $m_{\bm x_\mathrm{eq}}, L_{\bm x_\mathrm{eq}}$ of the Bregman divergence $D_\phi(\cdot \| \bm x_\mathrm{eq})$, and the deformed exponentials specified by $\low{\B}$ and $\up{\B}$. 
\begin{theorem}[Bounds of Bregman Divergence in Generalized Gradient Flow under Global Convexity Assumption]
    \label{thm:bounds of Bregman divergence in generalized equilibirum flow}
    Let $\qty{\bm x_t}_{0\leq t\leq T}$ follow the generalized gradient flow \cref{eq:generalized gradient flow}. Let $\bm x_\mathrm{eq}$ be an equilibrium state such that $\qty{\bm x_\mathrm{eq}} = \Stoich(\bm x_0) \cap \partial \phi^\ast[\Equib(\partial \phi(\bm{\hat x}))]$.
    \begin{subthms}
        \item
        \label[theorem]{thm:upper bound of Bregman divergence in generalized equilibirum flow}
        If \Cref{asm: activity-force decomposability of EPR lower bound,asm:convergence to the positive equilibrium} holds for $\low{\B}$ and $\low{\omega}$, and the Bregman divergence $D_\phi(\cdot \| \bm x_\mathrm{eq})$ is $m_{\bm x_\mathrm{eq}}$-PL on $\mathcal{O}_{T}(\bm x_0)$, then $D_\phi(\bm x_T \| \bm x_{\mathrm{eq}})$ is bounded from above as
        \begin{align}
        \label{eq:upper bound of Bregman divergence in generalized equilibirum flow}
            \begin{aligned}
                &D_\phi(\bm x_T \| \bm x_{\mathrm{eq}})\\
                &\leq \frac{\exp_{\low{\B} (\sqrt{\cdot})} \qty[ \ln_{\low{\B} (\sqrt{\cdot})} (2\sigma_{\rk(S)}^2 m_{\bm x_\mathrm{eq}}  D_\phi(\bm x_0 \| \bm x_{\mathrm{eq}}) ) - 2\sigma_{\rk(S)}^2 m_{\bm x_\mathrm{eq}} \low{\Omega}(\bm{x}_{0:T}) ]}{2 \sigma_{\rk(S)}^2 m_{\bm x_\mathrm{eq}}},
            \end{aligned}
        \end{align}
        where $\low{\Omega}(\bm{x}_{0:T}):= \int_0^T \low{\omega}(\bm x_t) \dd t$.
        \item 
        \label[theorem]{thm:lower bound of Bregman divergence in generalized equilibirum flow}
        If \Cref{asm: activity-force decomposability of EPR upper bound} holds for $\up{\B}$ and $\up{\omega}$, and the Bregman divergence $D_\phi(\cdot \| \bm x_\mathrm{eq})$ is $L_{\bm x_\mathrm{eq}}$-LC on $\mathcal{O}_{T}(\bm x_0)$, then $D_\phi(\bm x_T \| \bm x_{\mathrm{eq}})$ is bounded from below as
        \begin{align}
            \label{eq:lower bound of Bregman divergence in generalized equilibirum flow}
            \begin{aligned}
                &D_\phi(\bm x_T \| \bm x_{\mathrm{eq}}) \\
                &\geq \frac{\exp_{\up{\B} (\sqrt{\cdot})} \qty[ \ln_{\up{\B} (\sqrt{\cdot})} (2 \sigma_1^2 L_{\bm x_\mathrm{eq}} D_\phi(\bm x_0 \| \bm x_{\mathrm{eq}}) ) - 2\sigma_1^2 L_{\bm x_\mathrm{eq}} \up{\Omega}(\bm{x}_{0:T}) ]}{2\sigma_1^2 L_{\bm x_\mathrm{eq}} },
            \end{aligned}
        \end{align}
        where $\up{\Omega}(\bm{x}_{0:T}):= \int_0^T \up{\omega}(\bm x_t) \dd t$.
    \end{subthms}
\end{theorem}
The proof of \Cref{thm:upper bound of Bregman divergence in generalized equilibirum flow} is given in \Cref{sec:proof of main results}. 
The complete proof of \Cref{thm:lower bound of Bregman divergence in generalized equilibirum flow} is placed in {\Cref{sec:proof for thm:lower bound of Bregman divergence in generalized equilibirum flow}}, because its idea is almost identical to that of \Cref{thm:upper bound of Bregman divergence in generalized equilibirum flow}; the only essential difference is that the technically harder lower bound on the force norm is replaced by a simpler upper bound.

\begin{remark}[Asymptotic Behavior of Bounds]
We summarize the asymptotic behavior of the bounds as $T\to \infty$.
Loosely speaking, the Bregman divergence $D_{\phi}(\bm{x}_T \| \bm{x}_{\mathrm{eq}})$
asymptotically approaches a deformed exponential of the time-integrated activities.
We now state the asymptotic bounds formally using Vinogradov's notation.
Let $f_T, g_T$ be positive functions with respect to $T \geq  0$. Define $f_T \ll g_T$ iff 
$
    \lim_{T \to \infty} \frac{g_T}{f_T} < \infty
$ holds.
% Let $f_T, g_T$ be positive functions with respect to $T \geq  0$. Define $f_T = O(g_T)$ iff 
% $
%     \lim_{T \to \infty} \frac{g_T}{f_T} < \infty
% $ holds.
% Similarly, we define $f_T = \Omega(g_T)$ \KS{Omega: confusing notation} iff
% $
%     \lim_{T \to \infty} \frac{f_T}{g_T} < \infty
% $
% holds.
% \KS{exclude O and Omega notations, and introduce $\ll$}
The lower and upper activity bounds $\low{\omega}(\bm x_T), \up{\omega}(\bm x_T)$ converge to positive values $\low{\omega}(\bm x_\mathrm{eq}), \up{\omega}(\bm x_\mathrm{eq}) > 0$ as $t \to \infty$, respectively.
Therefore, in $\exp_{\up{\B} (\sqrt{\cdot})} [\cdot ]$ of \cref{eq:lower bound of Bregman divergence in generalized equilibirum flow}, the second term, $- 2\sigma_{\rk(S)}^2 m_{\bm x_\mathrm{eq}} \low{\Omega}(\bm{x}_{0:T})$, dominates the first term because $\low{\Omega}(\bm{x}_{0:T}) = \int_0^T \low{\omega}(\bm{x}_{t}) \dd t$ goes to $+\infty$ as $\bm x_T\to \bm x_\mathrm{eq}$.
Therefore, we asymptotically represent the upper and lower bound as follows, respectively:
% \begin{align}
%     D_\phi(\bm x_T \| \bm x_{\mathrm{eq}}) 
%     &= O\qty( \frac{1}{2\sigma_{\rk(S)}^2 m_{\bm x_\mathrm{eq}} }\exp_{\low{\B} (\sqrt{\cdot})} \qty[  - 2\sigma_{\rk(S)}^2 m_{\bm x_\mathrm{eq}} \low{\Omega}(\bm{x}_{0:T}) ] )\\
%     &= \Omega\qty( \frac{1}{2\sigma_1^2 L_{\bm x_\mathrm{eq}} }\exp_{\up{\B} (\sqrt{\cdot})} \qty[  - 2\sigma_1^2 L_{\bm x_\mathrm{eq}} \up{\Omega}(\bm{x}_{0:T}) ] ),
% \end{align}
\begin{align}
    \exp_{\up{\B} (\sqrt{\cdot})} \qty[  - 2\sigma_1^2 L_{\bm x_\mathrm{eq}} \up{\Omega}(\bm{x}_{0:T}) ] \ll 
    D_\phi(\bm x_T \| \bm x_{\mathrm{eq}}) 
    \ll  \exp_{\low{\B} (\sqrt{\cdot})} \qty[  - 2\sigma_{\rk(S)}^2 m_{\bm x_\mathrm{eq}} \low{\Omega}(\bm{x}_{0:T}) ].
\end{align}
\end{remark}

\subsubsection{Bounds of Bregman Divergence under Local Convexity Assumptions}
We state two additional theorems that provide bounds for Bregman divergence under local convexity assumptions, which are stricter than the global convexity assumptions.
% The crucial difference is that the time-integrated activities are replaced by those modified for local convexity. 
We make the assumptions on EPR bounds stronger to obtain the bounds.

\begin{theorem}[Bounds under Local Convexity Assumption]
    \label{thm:bounds under locally convexity assumption}
    Suppose \Cref{asm:invariance of positive orthant} holds.
    Let $\qty{\bm x_t}_{0\leq t\leq T}$ follow the generalized gradient flow \cref{eq:generalized gradient flow}. Let $\bm x_\mathrm{eq}$ be an equilibrium state such that $\qty{\bm x_\mathrm{eq}} = \Stoich(\bm x_0) \cap \partial \phi^\ast[\Equib(\partial \phi(\bm{\hat x}))]$.
    \begin{subthms}
        \item 
        \label[theorem]{thm:upper bound under locally PL assumption}
        The \emph{effective concentration} of $\bm x$ with respect to an equilibrium state $\bm x_\mathrm{eq}$ is defined as $\bm x^\mathrm{eff}(\bm x, \bm x_\mathrm{eq}) := \partial_{\bm \mu} \phi^\ast [ \partial\phi(\bm x_\mathrm{eq}) + P_{\Equib(\bm{0})^\perp} (\partial_{\bm x} D_\phi (\bm{x} \| \bm x_\mathrm{eq}) ) ]$. 
        % \KS{[x eff is firstly defined]}
        
        If \Cref{asm: convexity of EPR lower bound} holds for $\low{\B}$ and $\low{\omega}$, and $\phi$ is $m_{\bm x_\mathrm{eq}}(\cdot)$-locally PL and $m_{\bm x_\mathrm{eq}}(\cdot) \geq \low{m_{\bm x_\mathrm{eq}}}$ on $\bm x^\mathrm{eff}(\mathcal{O}_{T}(\bm x_0), \bm x_\mathrm{eq}) := \qty{\bm x^\mathrm{eff}(\bm x, \bm x_\mathrm{eq}) \mid \bm x \in \mathcal{O}_{T}(\bm x_0)}$, then the Bregman divergence $D_\phi(\bm x_T \| \bm x_{\mathrm{eq}})$ is bounded from above as
        % \begin{align}
        %     \begin{aligned}
        %         &D_\phi(\bm x_T \| \bm x_{\mathrm{eq}}) \\&\leq \frac{\exp_{\low{\B} (\sqrt{\cdot})} \qty[ \ln_{\low{\B} (\sqrt{\cdot})} (2\sigma_{\rk(S)}^2 \low{m_{\bm x_\mathrm{eq}}}  D_\phi(\bm x_0 \| \bm x_{\mathrm{eq}}) ) - 2\sigma_{\rk(S)}^2 \low{m_{\bm x_\mathrm{eq}}}  \low{\Omega}^{m_{\bm x_\mathrm{eq}}}(\bm{x}_{0:T}) ]}{2 \sigma_{\rk(S)}^2 \low{m_{\bm x_\mathrm{eq}}}},            
        %     \end{aligned}
        % \end{align}
        \begin{multline}
            \shoveleft{D_\phi(\bm x_T \| \bm x_{\mathrm{eq}})} \\\leq \frac{\exp_{\low{\B} (\sqrt{\cdot})} \qty[ \ln_{\low{\B} (\sqrt{\cdot})} (2\sigma_{\rk(S)}^2 \low{m_{\bm x_\mathrm{eq}}}  D_\phi(\bm x_0 \| \bm x_{\mathrm{eq}}) ) - 2\sigma_{\rk(S)}^2 \low{m_{\bm x_\mathrm{eq}}}  \low{\Omega}^{m_{\bm x_\mathrm{eq}}}(\bm{x}_{0:T}) ]}{2 \sigma_{\rk(S)}^2 \low{m_{\bm x_\mathrm{eq}}}},               
        \end{multline}
        where 
        % $\low{m_{\bm x_\mathrm{eq}}} := \min_{\KS{\bm x \in \X(\bm x_0)}} m_{\bm x_\mathrm{eq}}(\bm x) > 0$ and  
        $\low{\Omega}^{m_{\bm x_\mathrm{eq}}}(\bm{x}_{0:T}):= \int_0^T \frac{m_{\bm x_\mathrm{eq}}(\bm x_t^\mathrm{eff})}{\low{m_{\bm x_\mathrm{eq}}}} \cdot \low{\omega}(\bm x_t) \dd t$.
        \item 
        \label[theorem]{thm:lower bound under locally LC assumption}
        If \Cref{asm: convexity of EPR upper bound} holds for $\up{\B}$ and $\up{\omega}$, and $D_\phi(\cdot \| \bm x_\mathrm{eq})$ is $L_{\bm x_\mathrm{eq}}(\cdot)$-locally LC and $L_{\bm x_\mathrm{eq}}(\cdot) \leq \up{L_{\bm x_\mathrm{eq}}}$ on $\mathcal{O}_{T}(\bm x_0)$, then the Bregman divergence $D_\phi(\bm x_t \| \bm x_{\mathrm{eq}})$ is bounded from below as
        % \begin{align}
        %     \begin{aligned}
        %         &D_\phi(\bm x_T \| \bm x_{\mathrm{eq}}) \\&\geq \frac{\exp_{\up{\B} (\sqrt{\cdot})} \qty[ \ln_{\up{\B} (\sqrt{\cdot})} (2\sigma_1^2 \up{L_{\bm x_\mathrm{eq}}}  D_\phi(\bm x_0 \| \bm x_{\mathrm{eq}}) ) - 2\sigma_1^2 \up{L_{\bm x_\mathrm{eq}}}  \up{\Omega}^{L_{\bm x_\mathrm{eq}}}(\bm{x}_{0:T}) ]}{2 \sigma_1^2 \up{L_{\bm x_\mathrm{eq}}}},            
        %     \end{aligned}
        % \end{align}
        \begin{multline}
            \shoveleft{D_\phi(\bm x_T \| \bm x_{\mathrm{eq}})} \\\geq \frac{\exp_{\up{\B} (\sqrt{\cdot})} \qty[ \ln_{\up{\B} (\sqrt{\cdot})} (2\sigma_1^2 \up{L_{\bm x_\mathrm{eq}}}  D_\phi(\bm x_0 \| \bm x_{\mathrm{eq}}) ) - 2\sigma_1^2 \up{L_{\bm x_\mathrm{eq}}}  \up{\Omega}^{L_{\bm x_\mathrm{eq}}}(\bm{x}_{0:T}) ]}{2 \sigma_1^2 \up{L_{\bm x_\mathrm{eq}}}},            
        \end{multline}
        where 
        % $\up{L_{\bm x_\mathrm{eq}}} := \max_{\KS{\bm x \in \X(\bm x_0)}} L_{\bm x_\mathrm{eq}}(\bm x) > 0$ and  
        $\up{\Omega}^{L_{\bm x_\mathrm{eq}}}(\bm{x}_{0:T}):= \int_0^T \frac{L_{\bm x_\mathrm{eq}}(\bm x_t)}{\up{L_{\bm x_\mathrm{eq}}}} \cdot \up{\omega}(\bm x_t) \dd t$.
    \end{subthms}

\end{theorem}
The complete proofs of \Cref{thm:bounds under locally convexity assumption} are placed in {\Cref{sec:proof for thm:upper bound under locally PL assumption,sec:proof for thm:lower bound under locally LC assumption}} because the proofs are similar to that of \Cref{thm:upper bound of Bregman divergence in generalized equilibirum flow}.

\begin{remark}[Asymptotic behavior of bounds (continued)]
We represent the long-time asymptotic behavior of the local convexity bounds as
% \begin{align}
%     D_\phi(\bm x_T \| \bm x_{\mathrm{eq}}) &= O\qty( \frac{1}{2\sigma_{\rk(S)}^2 \low{m_{\bm x_\mathrm{eq}}} }\exp_{\low{\B} (\sqrt{\cdot})} \qty[  - 2\sigma_{\rk(S)}^2 \low{m_{\bm x_\mathrm{eq}}} \cdot \low{\Omega}^{m_{\bm x_\mathrm{eq}}}(\bm{x}_{0:T}) ] )\\&= \Omega\qty( \frac{1}{2\sigma_1^2 \up{L_{\bm x_\mathrm{eq}}} }\exp_{\up{\B} (\sqrt{\cdot})} \qty[  - 2\sigma_1^2 \up{L_{\bm x_\mathrm{eq}}} \cdot \up{\Omega}^{L_{\bm x_\mathrm{eq}}}(\bm{x}_{0:T}) ] ).
% \end{align}
\begin{align}
    \begin{aligned}
        \exp_{\up{\B} (\sqrt{\cdot})} \qty[ - 2\sigma_1^2 \up{L_{\bm x_\mathrm{eq}}} \cdot \up{\Omega}^{L_{\bm x_\mathrm{eq}}}(\bm{x}_{0:T})] &\ll 
    D_\phi(\bm x_T \| \bm x_{\mathrm{eq}}) \\&\ll   \exp_{\up{\B} (\sqrt{\cdot})} \qty[ - 2\sigma_{\rk(S)}^2 \low{m_{\bm x_\mathrm{eq}}} \cdot \low{\Omega}^{m_{\bm x_\mathrm{eq}}}(\bm{x}_{0:T}) ].        
    \end{aligned}
\end{align}
By the conditions in \Cref{thm:bounds under locally convexity assumption}, the resulting bounds are stricter than those in \Cref{thm:bounds of Bregman divergence in generalized equilibirum flow} for $\low{m_{\bm x_\mathrm{eq}}}$ and $\up{L_{\bm x_\mathrm{eq}}}$, respectively.
The global convexity bounds are also recovered by choosing the local convexity parameters to be constants.
Intuitively, the local convexity bounds exploit more information about the convexity than their global counterparts, and therefore yield stricter bounds.
% \KS{difference b/w global \& local}
\end{remark}

\begin{KSblock}
\begin{example}[Bounds in the Gray-Scott model]
\label{eg:bounds in gray scott}
For the Gray--Scott example, the quadratic divergence $D_2$ ($p = 2$ in \cref{eq:p-norm divergence}) is
\begin{align}
D_2(\bm{x}\| \bm x_\mathrm{eq})= \frac{1}{2}\sum_{i=1}^4 x_{\mathrm{eq},i} \qty(\frac{x_i}{{x}_{\mathrm{eq},i}}-1 )^2.
\end{align}
Hence, the local convexity $\rho_{\bm x_\mathrm{eq}}^{\phi_{h_2}}$ is given by
\begin{align}
\rho_{\bm x_\mathrm{eq}}^{\phi_{h_2}}(\bm{x})
=\frac{\norm{\partial_{\bm{x}}D_2(\bm{x}\| \bm x_\mathrm{eq})}_2^2}{2 D_2(\bm{x}\| \bm x_\mathrm{eq})}
=\frac{\sum_{i=1}^4 \qty(\frac{x_i}{x_{\mathrm{eq},i}}-1 )^2}{\sum_{i=1}^4 x_{\mathrm{eq},i} \qty(\frac{x_i}{x_{\mathrm{eq},i}}-1 )^2},
\end{align}
and satisfies
$
({\max_i x_{\mathrm{eq},i}})^{-1} \leq \rho_{\bm x_\mathrm{eq}}^{\phi_{h_2}}(\bm{x})\leq ({\min_i x_{\mathrm{eq},i}})^{-1}
$.
Thus, $({\max_i x_{\mathrm{eq},i}})^{-1}$ and $({\min_i x_{\mathrm{eq},i}})^{-1}$ are global convexity parameters.
In the quadratic dissipation case, the force, the activity and the EPR function are defined by 
\begin{align}
\bm f_{h_2}(\bm x) &= S^\top \partial_{\bm x} D_2(\bm x\| \bm x_\mathrm{eq}) = S^\top (\bm x/ \bm x_\mathrm{eq} - \bm 1),\\
\omega_{h_2,r}(\bm x) &= \frac{j_r(x)}{f_{h_2,r}(x)} = 
\begin{cases}
k_1^+x_{\mathrm{eq},1}x_2^2
= k_1^-x_{\mathrm{eq},2}x_2^2, & r=1,\\
k_2^+x_{\mathrm{eq},2}
= k_2^-x_{\mathrm{eq},3}, & r=2,\\
\KSTwo{k_3^+x_{\mathrm{eq},1}
= k_3^-x_{\mathrm{eq},4}}, & r=3,
\end{cases}\\
\dot{\Sigma}(\bm{x}, \bm f_{h_2})&=\sum_{r=1}^3 \omega_{h_2, r} (\bm{x}) f_{h_2, r}^2.
\end{align}
For
$
\underline{\omega}(\bm x):=\min_r\omega_{h_2, r}(\bm{x}),
\overline{\omega}(\bm x):=\max_r\omega_{h_2, r}(\bm{x}),
\overline{\mathcal{B}}(\abs{f}) = \underline{\mathcal{B}}(\abs{f}) = \abs{f}^2
$, \Cref{asm:convexity of bound functions} is satisfied.

For this system, we present our bounds and sketch their derivation.
Considering the factorized bounds, one has
\begin{align}
\label{eq:activity estimate}
 \low{\omega}(\bm x) \|\bm f_{h_2}(\bm x)\|_2^2 
\leq \dot{\Sigma}(\bm{x}, \bm f_{h_2}(\bm x))
\leq  \up{\omega}(\bm x)  \|\bm f_{h_2}(\bm x)\|_2^2.
\end{align}
One also has the bounds of the force norm by the gradient of the divergence, as
\begin{align}
\label{eq:force-to-gradient estimate}
 \sigma_{\rk(S)}^2 \underline{\omega}(\bm x) \norm{ \partial_{\bm x^\mathrm{eff}}D_2(\bm{x}^\mathrm{eff}\| {\bm{x}}_\mathrm{eq}) }_2^2
\leq \dot{\Sigma}(\bm{x}, \bm f_{h_2}(\bm x))
\leq  \sigma_1^2 \overline{\omega}(\bm x)\norm{ \partial_{\bm x}D_2(\bm{x}\| {\bm{x}}_\mathrm{eq}) }_2^2.
\end{align}
Here, even when $S$ is rank-deficient, a nonzero lower bound can be guaranteed by considering the gradient evaluated at the effective state $\bm x^\mathrm{eff}$.
The details for the lower bound in the genralized gradient flow can be found in \Cref{lemma:bounds of thermodynamic force norm}.
Lastly, using the local convexity $\rho_{\bm x_\mathrm{eq}}^{\phi_{h_2}}(\bm{x})$, one can obtain the bounds of the EPR by $D_2$:
\begin{align}
2 s_{\rk(S)}^2 \rho_{\bm x_\mathrm{eq}}^{\phi_{h_2}} (\bm x^\mathrm{eff}) \underline{\omega}(\bm x) D_2(\bm{x}\| {\bm{x}}_\mathrm{eq})
\leq \dot{\Sigma}(\bm{x}, \bm f_{h_2}(\bm x))
\leq 2 \sigma_1^2 \rho_{\bm x_\mathrm{eq}}^{\phi_{h_2}} (\bm x) \overline{\omega}(\bm x) D_2(\bm{x}\| {\bm{x}}_\mathrm{eq}).
\end{align}
Using the global convexity above, one can also obtain the bounds as
\begin{align}
\label{eq:convexity estimate}
\frac{2 s_{\rk(S)}^2  \underline{\omega}(\bm x)}{{\max_i x_{\mathrm{eq},i}}} D_2(\bm{x}\| {\bm{x}}_\mathrm{eq})
\leq \dot{\Sigma}(\bm{x}, \bm f_{h_2}(\bm x))
\leq \frac{2 \sigma_1^2 \overline{\omega}(\bm x)}{{\min_i x_{\mathrm{eq},i}}} D_2(\bm{x}\| {\bm{x}}_\mathrm{eq}).
\end{align}
The time integral leads the explicit bounds of the divergence with the global convexity:
\begin{align}
    \label{eq:global bounds in Gray-Scott}
    \exp\qty( - \frac{2 \sigma_1^2 }{{\min_i x_{\mathrm{eq},i}}} \int_0^T \overline{\omega}(\bm x_t) \dd t) \leq \frac{D_2(\bm{x}_T\| {\bm{x}}_\mathrm{eq})}{D_2(\bm{x}_0\| {\bm{x}}_\mathrm{eq})} \leq \exp\qty( - \frac{2 s_{\rk(S)}^2  }{{\max_i x_{\mathrm{eq},i}}} \int_0^T \underline{\omega}(\bm x_t) \dd t).
\end{align}

Compared with the convergence theorem for the reversible Gray--Scott model \cite{laurencot2025}, our bounds give a trajectory-wise estimate of relaxation. Laurencot and Walker identify the limiting equilibrium and prove local exponential stability, with the rate determined by the spectrum of the linearized operator. 
By contrast, our estimates bound the divergence itself at each time along the trajectory and express its decay through stoichiometric singular values, reaction activities, and convexity constants; see \cref{eq:activity estimate,eq:force-to-gradient estimate,eq:convexity estimate}. 
In the large-time regime, these bounds yield exponential decay of the divergence, which also implies exponential convergence in norm by local equivalence with the squared Euclidean norm. 
By examining the loss at each approximation step in the EPR bounds, we can identify a posteriori which step is responsible for the looseness of the relaxation bound. 
This diagnostic does not necessarily require the full concentration trajectory: in this Gray--Scott example, the global bounds \cref{eq:global bounds in Gray-Scott} are especially simple, because apart from rate constants and the conserved total concentration, the activity bounds are controlled by the time series of $x_2$.
\end{example}
\end{KSblock}

\subsection{Corollaries for Equilibrium Mass-Action CRNs}
\label{sec:corollaries for equilibrium CRNs}
Here, we present the corollaries of the main theorem for equilibrium mass-action CRNs. Specifically, we look at the explicit functional forms of the KL divergence bounds for CRNs with two dissipation functions: hyperbolic and quadratic.
% \KS{skip the detailed forms of bounds and integrated-activity, just state omega and B}
\KS{We note that both \Cref{asm:invariance of positive orthant,asm:convergence to the positive equilibrium} are satisfied for equilibrium mass-action CRNs with the KL divergence.}

\subsubsection{Explicit KL-divergence Bounds in CRNs with Hyperbolic Dissipation}
First, \Cref{lemma:upper bounds of separable EPR} applies to the hyperbolic dissipation function, since it is separable.
Since $(\psi_\mathrm{hyp}^\ast)'(f) = 2 \sinh(f/2)$ is convex, \Cref{lemma:lower bounds of separable EPR} also applies.
Thus, we obtain the following corollary for EPR bounds in hyperbolic case:
\begin{corollary}[EPR Bounds for Hyperbolic Dissipation]
    \label{cor:EPR bounds for hyperbolic dissipation}
    For the hyperbolic dissipation function with the geometric mean activity $\bm\omega_\mathrm{geo}(\bm x)$ and the cosh scalar dissipation function $\psi_\mathrm{hyp}^\ast(f)$ \cref{eq:hyperbolic dissipation}, \Cref{asm: convexity of EPR upper bound} is satisfied by 
\begin{align}
    \up{\omega}_\mathrm{geo}(\bm x) :=& \| \bm \omega_{\log} (\bm x) \|_2 = \sqrt{\sum_{r = 1}^M {j^+_r(\bm x) \cdot  j^-_r(\bm x)}},\\
    \up{\B}_\mathrm{hyp}(\abs{f}) :=& \abs{f} \cdot \psi^\ast_\mathrm{hyp}(\abs{f}) = 2\abs{f} \cdot \sinh\qty(\frac{\abs{f}}{2}),
\end{align} and \Cref{asm: convexity of EPR lower bound} is satisfied by 
\begin{align}
    \low{\omega}_\mathrm{geo}(\bm x) :=& \| \bm \omega_{\log} (\bm x) \|_{-\infty} = \min_{1\leq r\leq M} \sqrt{ {j^+_r(\bm x) \cdot  j^-_r(\bm x)}},\\
    \low{\B}_\mathrm{hyp}(\abs{f}) :=& \abs{f} \cdot \psi^\ast_\mathrm{hyp}\qty(\frac{\abs{f}}{\sqrt{M}}) = 2\abs{f} \cdot \sinh\qty(\frac{\abs{f}}{2\sqrt{M}}).
\end{align}
    In other words, the inequalities $\low{\dot\Sigma}_\mathrm{hyp}(\bm x, \bm f) \leq \dot \Sigma(\bm x, \bm f) \leq \up{\dot\Sigma}_\mathrm{hyp}(\bm x, \bm f)$ hold for $\bm x \in \X$ and $\bm f \in \F^\mathrm{eq}$, where
    % \begin{align}
    %     \begin{aligned}
    %         \low{\omega}_\mathrm{geo}(\bm x) \cdot \low{\B}_\mathrm{hyp}(\norm{\bm f}_2)
    %         \leq \expval{\bm f, \bm \omega_\mathrm{geo}(\bm x)\circ (\psi_\mathrm{hyp}^\ast)'(\bm f)} 
    %         \leq \up{\omega}_\mathrm{geo}(\bm x) \cdot \up{\B}_\mathrm{hyp}(\norm{\bm f}_2),
    %     \end{aligned}
    % \end{align}
    \begin{align}
        \label{eq:hyperbolic EPR bounds}
        \low{\dot\Sigma}_\mathrm{hyp}(\bm x, \bm f) := \low{\omega}_\mathrm{geo}(\bm x) \cdot \low{\B}_\mathrm{hyp}(\norm{\bm f}_2),\
        \up{\dot\Sigma}_\mathrm{hyp}(\bm x, \bm f) := \up{\omega}_\mathrm{geo}(\bm x) \cdot \up{\B}_\mathrm{hyp}(\norm{\bm f}_2).
    \end{align}
    For $M\geq 2$, the equality hold iff $\bm f = \bm 0$. For $M=1$, the equality always hold.
\end{corollary}
We state the corollary of \Cref{thm:lower bound under locally LC assumption,thm:upper bound under locally PL assumption} with respect to $\rho^\phi_{\bm x_\mathrm{eq}}(\bm x)$ for hyperbolic dissipation.
We can explicitly write the deformed logarithm with respect to $\up{\B}_\mathrm{hyp} (\sqrt{\cdot}) = 2 \sqrt{\cdot } \sinh(\sqrt{\cdot }/2)$ as
\begin{align}
    \label{eq:ln_hyp}
    \begin{aligned}
            \ln_\mathrm{hyp}(u) :=&\, \ln_{\up{\B}_\mathrm{hyp}(\sqrt{\cdot})}(u) 
        = \int_1^u \frac{\dd u'}{2\sqrt{u'} \sinh(\sqrt{u'}/2)}
        % =&\, 2\int_1^u \frac{\dd (\sqrt{u'}/2)}{\sinh(\sqrt{u'}/2)}
        = 2 \ln \qty(\tanh\qty(\frac{\sqrt{u}}{4})) - C.
    \end{aligned}
\end{align}
Here, we use the integral $\int \frac{1}{\sinh y} \dd y = \ln \qty(\tanh\qty(\frac{y}{2}))$ \cite
[\href{https://dlmf.nist.gov/4.40.E4}{(4.40.4)}]{822801}.
$C$ is a constant: $C :=  2\ln \qty(\tanh\qty(\frac{1}{4}))$.
The deformed exponential function is also calculated as
\begin{align}
    \label{eq:exp_hyp}
    \exp_\mathrm{hyp}(v)
    := \exp_{\up{\B}_\mathrm{hyp} (\sqrt{\cdot})}(v) = \qty(4 \tanh^{-1} \qty(\exp \qty(\frac{v + C}{2})))^2.
\end{align}
% The $\ln_{\low{\B}_\mathrm{hyp} (\sqrt{\cdot })}$ and $\exp_{\low{\B}_\mathrm{hyp} (\sqrt{\cdot })}$
The deformed functions for $\low{\B}_\mathrm{hyp}  (\sqrt{\cdot }) = 2 \sqrt{\cdot} \sinh(\sqrt{\cdot }/2\sqrt{M})$ are also calculated as 
\begin{align}
    \label{eq:ln_hyp^m}
    \ln_\mathrm{hyp}^M(u) &:= \ln_{\low{\B}_\mathrm{hyp} (\sqrt{\cdot})}
    = 2 \sqrt{M} \ln \qty(\tanh\qty(\frac{\sqrt{u}}{4\sqrt{M}})) - C',\\
    \label{eq:exp_hyp^m}
    \exp_\mathrm{hyp}^M(v)
    &:=\exp_{\low{\B}_\mathrm{hyp} (\sqrt{\cdot})}
    = \qty(4 \sqrt{M} \tanh^{-1} \qty(\exp \qty(\frac{v + C'}{2\sqrt{M}})))^2,
\end{align}
where $C'$ is a constant. 

\begin{corollary}[Explicit Bounds of KL divergence in Mass-Action CRNs with Hyperbolic Dissipation]
    \label{cor:bounds in hyperbolic mass action CRNs}
    % \KS{\sout{Assume \Cref{asm:invariance of positive orthant} in an equilibrium mass-action CRN.}}
    For an equilibrium mass-action CRN,
    from \Cref{lemma:global convexity of Bregman divergence on solution orbit}, \Cref{thm:bounds under locally convexity assumption} and \Cref{cor:EPR bounds for hyperbolic dissipation}, the divergence $D_\mathrm{KL}(\cdot \| \bm x_\mathrm{eq})$ is $\low{\rho^{\phi_\mathrm{KL}}_{\bm x_\mathrm{eq}}}$-PL and $\up{\rho^{\phi_\mathrm{KL}}_{\bm x_\mathrm{eq}}}$-LC on $\mathcal{O}_{T}(\bm x_0)$ for some $\low{\rho^{\phi_\mathrm{KL}}_{\bm x_\mathrm{eq}}}, \up{\rho^{\phi_\mathrm{KL}}_{\bm x_\mathrm{eq}}}$, and 
    % $\low{D}_\mathrm{hyp, loc}(\up{\Omega}^\mathrm{KL}_{\mathrm{geo}}(\bm x_t)) \leq D_\mathrm{KL}(\bm x_t \| \bm x_\mathrm{eq}) \leq \up{D}_\mathrm{hyp, loc}(\low{\Omega}^\mathrm{KL}_{\mathrm{geo}}(\bm x_t))$
    \begin{align}
        \low{D}_\mathrm{hyp, loc}(\up{\Omega}^\mathrm{KL}_{\mathrm{geo}}(\bm x_t)) \leq D_\mathrm{KL}(\bm x_t \| \bm x_\mathrm{eq}) \leq \up{D}_\mathrm{hyp, loc}(\low{\Omega}^\mathrm{KL}_{\mathrm{geo}}(\bm x_t))
    \end{align}
    holds for the upper and lower bounds,
    \begin{align}
        \label{eq:hyp_loc_upper}
                \up{D}_\mathrm{hyp, loc}(\low{\Omega}_{\mathrm{geo}}^\mathrm{KL}) &:= \frac{\exp_\mathrm{hyp} \qty[ \ln_\mathrm{hyp} (2\sigma_{\rk(S)}^2 \low{\rho^{\phi_\mathrm{KL}}_{\bm x_\mathrm{eq}}}  D_\mathrm{KL}(\bm x_0 \| \bm x_{\mathrm{eq}}) ) - 2\sigma_{\rk(S)}^2 \low{\rho^{\phi_\mathrm{KL}}_{\bm x_\mathrm{eq}}} \low{\Omega}_{\mathrm{geo}}^\mathrm{KL}]}{2 \sigma_{\rk(S)}^2 \low{\rho^{\phi_\mathrm{KL}}_{\bm x_\mathrm{eq}}}},\\
        \label{eq:hyp_loc_lower}
            \low{D}_\mathrm{hyp, loc}(\up{\Omega}_{\mathrm{geo}}^\mathrm{KL}) &:= \frac{\exp_\mathrm{hyp} \qty[ \ln_\mathrm{hyp} (2\sigma_1^2 \up{\rho^{\phi_\mathrm{KL}}_{\bm x_\mathrm{eq}}}  D_\mathrm{KL}(\bm x_0 \| \bm x_{\mathrm{eq}}) ) - 2\sigma_1^2 \up{\rho^{\phi_\mathrm{KL}}_{\bm x_\mathrm{eq}}} \up{\Omega}_{\mathrm{geo}}^\mathrm{KL}]}{2\sigma_1^2 \up{\rho^{\phi_\mathrm{KL}}_{\bm x_\mathrm{eq}}} }.
    \end{align}
    Here, $\low{\Omega}_{\mathrm{geo}}^\mathrm{KL}(\bm x_{0:T})$ and $\up{\Omega}_{\mathrm{geo}}^\mathrm{KL}(\bm{x}_{0:T})$ are defined as
    \begin{align}
            \low{\Omega}_{\mathrm{geo}}^\mathrm{KL}(\bm x_{0:T}) := \int_0^T \frac{\rho^{\phi_\mathrm{KL}}_{\bm x_\mathrm{eq}}(\bm x_\tau^\mathrm{eff})}{\low{\rho^{\phi_\mathrm{KL}}_{\bm x_\mathrm{eq}}}} \cdot \low{\omega}_\mathrm{geo}(\bm x_\tau) \dd \tau,\,
            \up{\Omega}_{\mathrm{geo}}^\mathrm{KL}(\bm x_{0:T}) := \int_0^T \frac{\rho^{\phi_\mathrm{KL}}_{\bm x_\mathrm{eq}}(\bm x_\tau)}{\up{\rho^{\phi_\mathrm{KL}}_{\bm x_\mathrm{eq}}}} \cdot  \up{\omega}_\mathrm{geo}(\bm x_\tau) \dd \tau,       
    \end{align}
    respectively.
\end{corollary}

\subsubsection{Explicit KL-divergence Bounds in CRNs with Quadratic Dissipation}
\Cref{lemma:bounds of separable EPR} can also be applied to the quadratic dissipation case, while stricter bounds can be derived for this setting.
\begin{lemma}[EPR Bounds for Quadratic Dissipation]
    \label{lemma:bounds of quadratic EPRs}
    For the quadratic dissipation function with the logarithmic mean activity $\bm\omega_\mathrm{log}(\bm x)$ and the quadratic scalar dissipation function $\psi_\mathrm{quad}^\ast(f)$ \cref{eq:quadratic dissipation},  \Cref{asm: activity-force decomposability of EPR upper bound} is satisfied by 
\begin{align}
    \up{\omega}_\mathrm{log}(\bm x) &:= \| \bm \omega_{\log} (\bm x) \|_\infty = \max_{1\leq r\leq M} \frac{j^+_r(\bm x) - j^-_r(\bm x)}{ \ln j^+_r(\bm x) - \ln j^-_r(\bm x)},\\
    \up{\B}_\mathrm{quad}(\abs{f}) &:= \abs{f} \cdot \psi^\ast_\mathrm{quad}(\abs{f}) = 2\abs{f}^2,
\end{align} and \Cref{asm: activity-force decomposability of EPR lower bound} is satisfied by 
\begin{align}
    \low{\omega}_\mathrm{log}(\bm x) &:= \| \bm \omega_{\log} (\bm x) \|_{-\infty} = \min_{1\leq r\leq M} \frac{j^+_r(\bm x) - j^-_r(\bm x)}{ \ln j^+_r(\bm x) - \ln j^-_r(\bm x)},\\
    \low{\B}_\mathrm{quad}(\abs{f}) &:= \abs{f} \cdot \psi^\ast_\mathrm{quad}\qty(\abs{f}) = 2\abs{f}^2.
\end{align}
    Especially, the following inequalities hold for $\bm x \in \X$ and $\bm f \in \F^\mathrm{eq}$,
    \begin{align}
        \norm{\bm \omega_\mathrm{log}(\bm x)}_{-\infty} \frac{\norm{\bm f}_2^2}{\sqrt{M}} \leq \low{\dot\Sigma}_\mathrm{quad}(\bm x, \bm f)
        \leq \Sigma(\bm x, \bm f)
        \leq \up{\dot\Sigma}_\mathrm{quad}(\bm x, \bm f)
        \leq \norm{\bm \omega_\mathrm{log}(\bm x)}_{2} {\norm{\bm f}_2^2},
    \end{align}
    where the bounds $\low{\dot\Sigma}_\mathrm{quad}, \up{\dot\Sigma}_\mathrm{quad}$ are defined as
    \begin{align}
        \label{eq:quadratic EPR bounds}
        \low{\dot\Sigma}_\mathrm{quad}(\bm x, \bm f) := \low{\omega}_\mathrm{log}(\bm x) \cdot \low{\B}_\mathrm{quad}(\norm{\bm f}_2),\
        \up{\dot\Sigma}_\mathrm{quad}(\bm x, \bm f) := \up{\omega}_\mathrm{log}(\bm x) \cdot \up{\B}_\mathrm{quad}(\norm{\bm f}_2).
    \end{align}
    The second and third inequalities hold if and only if $\omega_r(\bm x)$ takes the same value for $r \in \mathrm{supp}(\bm f)$.
\end{lemma}
\begin{proof}[Proof of \Cref{lemma:bounds of quadratic EPRs}]
    The first and fourth inequalities follow from $1 \leq \sqrt{M}$ and the monotonic decreasing property of $\norm{\cdot}_p$ with respect to $1\leq p \leq \infty$.
    The second and third inequalities are obtained as
    \begin{align}
        \expval{\bm \omega_\mathrm{log} (\bm x) \circ \bm f, \bm f}
        &= \sum_{r \in \mathrm{supp}(\bm f)} (\omega_\mathrm{log})_r(\bm x)) \abs{f_r}^2
        \geq \norm{\bm \omega_\mathrm{log}(\bm x)}_{-\infty} \sum_{r \in \mathrm{supp}(\bm f)} \abs{f_r}^2,\\
        \expval{\bm \omega_\mathrm{log} (\bm x) \circ \bm f, \bm f}
        &= \sum_{r \in \mathrm{supp}(\bm f)} (\omega_\mathrm{log})_r(\bm x) \abs{f_r}^2
        \leq \norm{\bm \omega_\mathrm{log}(\bm x)}_{\infty} \sum_{r \in \mathrm{supp}(\bm f)} \abs{f_r}^2.
    \end{align}
    Each $\abs{f_r}$ is positive for $r \in\mathrm{supp}(\bm f)$, thus the equalities hold iff $\omega_r(\bm x)$ takes the same value for $r \in \mathrm{supp}(\bm f)$.
\end{proof}
We state the corollary of \Cref{thm:bounds under locally convexity assumption} with respect to $\rho^\phi_{\bm x_\mathrm{eq}}(\bm x)$;
\begin{corollary}[Explicit Bounds of KL divergence in Mass-Action CRNs with Quadratic Dissipation]
    \label{cor:bounds in quadratic mass action CRNs}
    % \KS{\sout{Assume \Cref{asm:invariance of positive orthant} in an equilibrium mass-action CRN.}}
    For an equilibrium mass-action CRN,
    from \Cref{lemma:global convexity of Bregman divergence on solution orbit}, \Cref{thm:bounds under locally convexity assumption} and \Cref{lemma:bounds of quadratic EPRs}, the divergence $D_\mathrm{KL}(\cdot \| \bm x_\mathrm{eq})$ is $\low{\rho^{\phi_\mathrm{KL}}_{\bm x_\mathrm{eq}}}$-PL and $\up{\rho^{\phi_\mathrm{KL}}_{\bm x_\mathrm{eq}}}$-LC on $\mathcal{O}_{T}(\bm x_0)$ for some $\low{\rho^{\phi_\mathrm{KL}}_{\bm x_\mathrm{eq}}}, \up{\rho^{\phi_\mathrm{KL}}_{\bm x_\mathrm{eq}}}$, and 
    % $\low{D}_\mathrm{quad, loc}(\up{\Omega}^\mathrm{KL}_{\mathrm{log}}(\bm x_t)) \leq D_\mathrm{KL}(\bm x_t \| \bm x_\mathrm{eq}) \leq \up{D}_\mathrm{quad, loc}(\low{\Omega}^\mathrm{KL}_{\mathrm{log}}(\bm x_t))$ 
    \begin{align}
        \low{D}_\mathrm{quad, loc}(\up{\Omega}^\mathrm{KL}_{\mathrm{log}}(\bm x_t)) \leq D_\mathrm{KL}(\bm x_t \| \bm x_\mathrm{eq}) \leq \up{D}_\mathrm{quad, loc}(\low{\Omega}^\mathrm{KL}_{\mathrm{log}}(\bm x_t))
    \end{align}
    holds for the upper and lower bounds,
    \begin{align}
        \label{eq:quad_loc_upper}
                \up{D}_\mathrm{quad, loc}(\low{\Omega}_{\mathrm{log}}^\mathrm{KL}(\bm x_{0:T}))
                % & \frac{\exp \qty[ \ln (2\sigma_{\rk(S)}^2 \low{\rho^{\phi_\mathrm{KL}}_{\bm x_\mathrm{eq}}}  D_\mathrm{KL}(\bm x_0 \| \bm x_{\mathrm{eq}}) ) - 2\sigma_{\rk(S)}^2 \low{\rho^{\phi_\mathrm{KL}}_{\bm x_\mathrm{eq}}} \low{\Omega}_{\mathrm{log}}^\mathrm{KL}(\bm{x}_{0:T}) ]}{2 \sigma_{\rk(S)}^2 \low{\rho^{\phi_\mathrm{KL}}_{\bm x_\mathrm{eq}}}}\\
        &:=  D_\mathrm{KL}(\bm x_0 \| \bm x_{\mathrm{eq}}) \cdot \exp \qty[  - 2\sigma_{\rk(S)}^2 \low{\rho^{\phi_\mathrm{KL}}_{\bm x_\mathrm{eq}}} \low{\Omega}_{\mathrm{log}}^\mathrm{KL}(\bm{x}_{0:T}) ],\\
        \label{eq:quad_loc_lower}
            \low{D}_\mathrm{quad, loc}(\up{\Omega}_{\mathrm{log}}^\mathrm{KL}(\bm x_{0:T}))
            % & \frac{\exp \qty[ \ln (2\sigma_1^2 \up{\rho^{\phi_\mathrm{KL}}_{\bm x_\mathrm{eq}}}  D_\mathrm{KL}(\bm x_0 \| \bm x_{\mathrm{eq}}) ) - 2\sigma_1^2 \up{\rho^{\phi_\mathrm{KL}}_{\bm x_\mathrm{eq}}} \up{\Omega}_{\mathrm{log}}^\mathrm{KL}(\bm{x}_{0:T}) ]}{2\sigma_1^2 \up{\rho^{\phi_\mathrm{KL}}_{\bm x_\mathrm{eq}}} }\\
        &:=  D_\mathrm{KL}(\bm x_0 \| \bm x_{\mathrm{eq}}) \cdot  \exp \qty[  - 2\sigma_1^2 \up{\rho^{\phi_\mathrm{KL}}_{\bm x_\mathrm{eq}}} \,\up{\Omega}_{\mathrm{log}}^\mathrm{KL}(\bm{x}_{0:T}) ].
    \end{align}
    Here, $\low{\Omega}_{\mathrm{log}}^\mathrm{KL}(\bm x_{0:T})$ and $\up{\Omega}_{\mathrm{log}}^\mathrm{KL}(\bm{x}_{0:T})$ are defined as
    \begin{align}
            \low{\Omega}_{\mathrm{log}}^\mathrm{KL}(\bm x_{0:T}) := \int_0^T \frac{\rho^{\phi_\mathrm{KL}}_{\bm x_\mathrm{eq}}(\bm x_\tau^\mathrm{eff})}{\low{\rho^{\phi_\mathrm{KL}}_{\bm x_\mathrm{eq}}}} \cdot \low{\omega}_\mathrm{log}(\bm x_\tau) \dd \tau, \,
            \up{\Omega}_{\mathrm{log}}^\mathrm{KL}(\bm x_{0:T}) := \int_0^T \frac{\rho^{\phi_\mathrm{KL}}_{\bm x_\mathrm{eq}}(\bm x_\tau)}{\up{\rho^{\phi_\mathrm{KL}}_{\bm x_\mathrm{eq}}}} \cdot  \up{\omega}_\mathrm{log}(\bm x_\tau) \dd \tau,       
    \end{align}
    respectively.
\end{corollary}

\subsection{Proofs of the Main result}
% \subsection{Proofs of the Main result \Cref{thm:upper bound of Bregman divergence in generalized equilibirum flow}}
\label{sec:proof of main results}
In this section, we provide the proof of \Cref{thm:upper bound of Bregman divergence in generalized equilibirum flow}, which provide upper bounds for the Bregman divergence. The proofs of remaining results, \Cref{thm:lower bound of Bregman divergence in generalized equilibirum flow,thm:bounds under locally convexity assumption}, are largely analogous to those of the upper bounds, except for the force-norm bounds. 

First, we sketch the idea of the proof: using the EPR bound \Cref{asm: activity-force decomposability of EPR lower bound}, the force bound \Cref{lemma:lower bounds of thermodynamic force norm}, and the convexity condition, we bound the time derivative of the Bregman divergence by the product of a $D_\phi$-only term and an activity term.
By separating variables and integrating them, the resulting integral can be expressed via a deformed logarithm; applying its inverse yields an explicit bound.

Let us begin with the lemma on the force norm bounds.
We only provide a proof of \Cref{lemma:lower bounds of thermodynamic force norm} for the lower bound in the main text; that for the upper bound is elementary and is deferred to {\Cref{sec:proof for lemma:upper bounds of thermodynamic force norm}}.
One nontrivial point, compared with the upper-bound case, is that we orthogonally decompose the gradient of the divergence to remove its $\Ker S^\top$ component, which is used to obtain a nonzero lower bound.

\begin{lemma}[Bounds of separable EPR by force norm]
    \label{lemma:bounds of thermodynamic force norm}
    Consider the state $\bm x \in \Stoich(\bm x_\mathrm{eq})$ and the force $\bm f(\bm x) = S^\top \partial D_\phi(\bm x \| \bm x_\mathrm{eq})$ defined in \cref{equilibirum force with x_eq}.
    \begin{sublems}
        \item \label[lemma]{lemma:lower bounds of thermodynamic force norm}
        The following inequality hold:
        \begin{align}
        \label{eq:lower bound of thermodynamic force norm by gradient}
            \sigma_{\rk(S)} \|\partial_{\bm x} D_\phi(\bm{x}^\mathrm{eff} \| \bm x_\mathrm{eq})\|_2 \leq \|\bm f(\bm x)\|_2,
        \end{align}
        where $\bm{{x}}^\mathrm{eff} = \bm{{x}}^\mathrm{eff}(\bm x, \bm x_\mathrm{eq}) = \partial_{\bm \mu} \phi^\ast [ \partial\phi(\bm x_\mathrm{eq}) + P_{\Equib(\bm{0})^\perp} (\partial_{\bm x} D_\phi (\bm{x} \| \bm x_\mathrm{eq}) ) ]$ and $P_{\Equib(\bm{0})^\perp}$ is the orthogonal projection from $\R^N$ to $\Equib(\bm{0})^\perp = \Im S$.
        The equality in \cref{eq:lower bound of thermodynamic force norm by gradient} holds iff $\sigma_1 = \cdots = \sigma_{\rk(S)}$.

        If $\bm x \in \Stoich(\bm x_\mathrm{eq})$ and $D_\phi(\cdot \| \bm x_\mathrm{eq})$ is $m_{\bm x_\mathrm{eq}}(\cdot)$-locally PL at $\bm x^\mathrm{eff}$, then the following inequality hold:
        \begin{align}
        \label{eq:lower bound of thermodynamic force norm by convexity}
            \sqrt{2 \sigma_{\rk(S)}^2 m_{\bm x_\mathrm{eq}}(\bm{x}^\mathrm{eff})  D_\phi(\bm{x} \|\bm x_\mathrm{eq}) } \leq \|\bm f(\bm x)\|_2.
        \end{align}
        The equality in \cref{eq:lower bound of thermodynamic force norm by convexity} holds iff $\sigma_1 = \cdots = \sigma_{\rk(S)}$, $\partial D_\phi(\bm x \| \bm x_\mathrm{eq}) \in (\Ker S^\top)^\perp$, and $m_{\bm x_\mathrm{eq}}(\cdot)$-locally PL equality holds at $\bm x^\mathrm{eff}$.
        \item \label[lemma]{lemma:upper bounds of thermodynamic force norm}
        The following inequality hold:
        \begin{align}
        \label{eq:upper bound of thermodynamic force norm by gradient}
             \|\bm f(\bm x)\|_2 \leq \sigma_1 \|\partial_{\bm x} D_\phi(\bm{x} \| \bm x_\mathrm{eq})\|_2.
        \end{align}
        The equality holds iff $\sigma_1 = \cdots = \sigma_{\rk(S)}$.

        If $\bm x \in \Stoich(\bm x_\mathrm{eq})$ and $D_\phi(\cdot \| \bm x_\mathrm{eq})$ is $L_{\bm x_\mathrm{eq}}(\cdot)$-locally LC at $\bm x$, then the following inequality hold:
        \begin{align}
        \label{eq:upper bound of thermodynamic force norm by convexity}
            \|\bm f(\bm x)\|_2\leq \sqrt{2 \sigma_{1}^2 L_{\bm x_\mathrm{eq}}(\bm{x})  D_\phi(\bm{x} \|\bm x_\mathrm{eq}) }.
        \end{align}
        The equality holds iff $\sigma_1 = \cdots = \sigma_{\rk(S)}$ and $L_{\bm x_\mathrm{eq}}(\cdot)$-local LC equality
        % in \cref{eq:local Lipschitz continuity} 
        holds at $\bm x$.
    \end{sublems}
\end{lemma}

\begin{proof}[Proof of \Cref{lemma:lower bounds of thermodynamic force norm}]
    The potential space $\M - \bm{\tilde{\mu}}$ shifted by a fixed vector $\bm{\tilde\mu}$ is also a linear subspace. It is orthogonally decomposed as 
$\M - \bm{\tilde{\mu}}  =  \Equib(\bm {0})^\perp \oplus \Equib(\bm {0}) =  (\Ker S^\top)^\perp \oplus (\mathrm{Ker\,} S^\top)$ with the standard inner product $\expval{\cdot, \cdot}$.
Denote the orthogonal projection from $\M - \bm{\tilde{\mu}}$ to $\Equib(\bm {0})^\perp$ as $P_{\Equib(\bm {0})^\perp}$.
Using the orthogonal basis $\qty{\bm{\hat v}_i \mid i = 1, \cdots, N}$ of $\M$, defined by the stoichiometric matrix $S$ in \Cref{subsubsec:strucural properties of CRNs}, the image of $P_{\Equib(\bm {0})^\perp}$ is written as 
\begin{align}
    P_{\Equib(\bm {0})^\perp} \qty(\bm \mu - \bm{\tilde{\mu}}) = \sum_{i=1}^{\rk(S)} \expval{\bm \mu - \bm{\tilde{\mu}}, \bm{\hat v}_i}  \bm{\hat v}_i, \, \bm \mu = \sum_{i=1}^{N} \expval{\bm \mu - \bm{\tilde{\mu}}, \bm{\hat v}_i} \in \M.
\end{align}
Note that the following relation holds for any $\bm \mu \in \M$:
\begin{align}
    \begin{aligned}
        \norm{ \bm \mu - \bm{\tilde{\mu}} }_2^2 
        = \norm{ P_{\Equib(\bm {0})^\perp} \qty(\bm \mu - \bm{\tilde{\mu}}) }_2^2 + \norm{ \bm \mu - \bm{\tilde{\mu}} - P_{\Equib(\bm {0})^\perp} \qty(\bm \mu - \bm{\tilde{\mu}}) }_2^2 
        \geq \norm{ P_{\Equib(\bm {0})^\perp} \qty(\bm \mu - \bm{\tilde{\mu}}) }_2^2.
    \end{aligned}
\end{align}

We first prove the inequality \cref{eq:lower bound of thermodynamic force norm by gradient}.
The 2-norm of $\bm f(\bm x) = S^\top \partial D_\phi(\bm x \| \bm x_\mathrm{eq})$ is bounded from below by the norm of $\partial D_\phi(\bm x^\mathrm{eff} \| \bm x_\mathrm{eq})$, as
\begin{align}
    % \begin{aligned}
        \label{eq:1st}
        \norm{\bm f(\bm x)}_2^2 
        &= \norm{ S^\top \qty(  \sum_{i=1}^{N} \expval{\partial D_\phi(\bm x \| \bm x_\mathrm{eq}), \bm{\hat v}_i} \bm{\hat v}_i) }_2^2
        % \label{eq:2nd}
        = \norm{  \sum_{i=1}^{\rk(S)} \expval{\partial D_\phi(\bm x \| \bm x_\mathrm{eq}), \bm{\hat v}_i} \sigma_i \bm{\hat u}_i }_2^2\\
        \label{eq:3rd}
        &= \sum_{i=1}^{\rk(S)} \expval{\partial D_\phi(\bm x \| \bm x_\mathrm{eq}), \bm{\hat v}_i}^2 \sigma_i^2\\
        \label{eq:4th}
        &\geq \sigma_{\rk(S)}^2 \sum_{i=1}^{\rk(S)} \expval{\partial D_\phi(\bm x \| \bm x_\mathrm{eq}), \bm{\hat v}_i}^2\\
        \label{eq:5th}
        &= \sigma_{\rk(S)}^2 \norm{ P_{\Equib(\bm {0})^\perp} \qty[\partial D_\phi(\bm x \| \bm x_\mathrm{eq})] }_2^2
        % \label{eq:6th}
        = \sigma_{\rk(S)}^2 \norm{\partial D_\phi(\bm x^\mathrm{eff} \| \bm x_\mathrm{eq})}_2^2.
    % \end{aligned}
\end{align}
The transformation from \cref{eq:1st} to \cref{eq:3rd} follows from \cref{eq:basis by stoichiometric matrix}.
The transformation from \cref{eq:3rd} to \cref{eq:4th} is because $\sigma_{\rk(S)}$ is the minimum positive singular value of $S$, and the equality holds iff $\sigma_1 = \cdots = \sigma_{\rk(S)}$. The transformation from \cref{eq:4th} to \cref{eq:5th} follows from the orthonormality of $\qty{\bm{\hat v}_i}$.

We note the information-geometric relation among $\bm x$, $\bm x_\mathrm{eq}$, and $\bm x^\mathrm{eff}$, which is used in the next paragraph.
For $\bm x^\mathrm{eff}$, the following equality holds:
\begin{align}
    \label{eq:orthogonality of effective concentration}
    \expval{\partial\phi(\bm x) - \partial\phi(\bm{ x}^\mathrm{eff}), \bm x - \bm x_\mathrm{eq}}
    = \expval{\partial D_\phi(\bm{x} \| \bm x_\mathrm{eq}) - P_{\Equib(\bm{0})^\perp}[\partial D_\phi(\bm x \| \bm x_\mathrm{eq})], \bm x_t - \bm x_\mathrm{eq}}=0.
\end{align}
The last equation holds because $\partial D_\phi(\bm{x} \| \bm x_\mathrm{eq}) - P_{\Equib(\bm{0})^\perp}[\partial D_\phi(\bm x \| \bm x_\mathrm{eq})] \in \Equib(\bm 0) = \Ker S^\top$ and $\bm x - \bm x_\mathrm{eq} \in \Im S$.
The generalized Pythagorean relation holds for the Bregman divergence $D_\phi$ as
\begin{align}
    \label{eq:Pythagorean relation of effective concentration}
    D_\phi(\bm{x}^\mathrm{eff} \| \bm{x}_\mathrm{eq}) = D_\phi(\bm{x}^\mathrm{eff} \| \bm{x}) + D_\phi(\bm{x} \| \bm{x}_\mathrm{eq}),
\end{align}
which is equivalent to \cref{eq:orthogonality of effective concentration}.
Therefore, $\qty{\bm x} = \Stoich(\bm x_\mathrm{eq}) \cap \partial\phi^\ast(\Equib(\partial\phi(\bm x^\mathrm{eff})))$ holds for $\bm x$, $\bm x_\mathrm{eq}$, and $\bm x^\mathrm{eff}$ (see also \Cref{fig:information-geometric relation}).

% The 2-norms of the gradients relate as
% \begin{align}
%     \norm{\partial D_\phi(\bm{x}_t \| \bm x_\mathrm{eq})}_2^2 = \norm{\partial D_\phi(\bm{x}_t \| \bm{\tilde x}_t)}_2^2 + \norm{\partial D_\phi(\bm{\tilde x}_t \| \bm x_\mathrm{eq})}_2^2.
% \end{align}

\begin{figure}
    \centering
    \includegraphics[width=.8\linewidth]{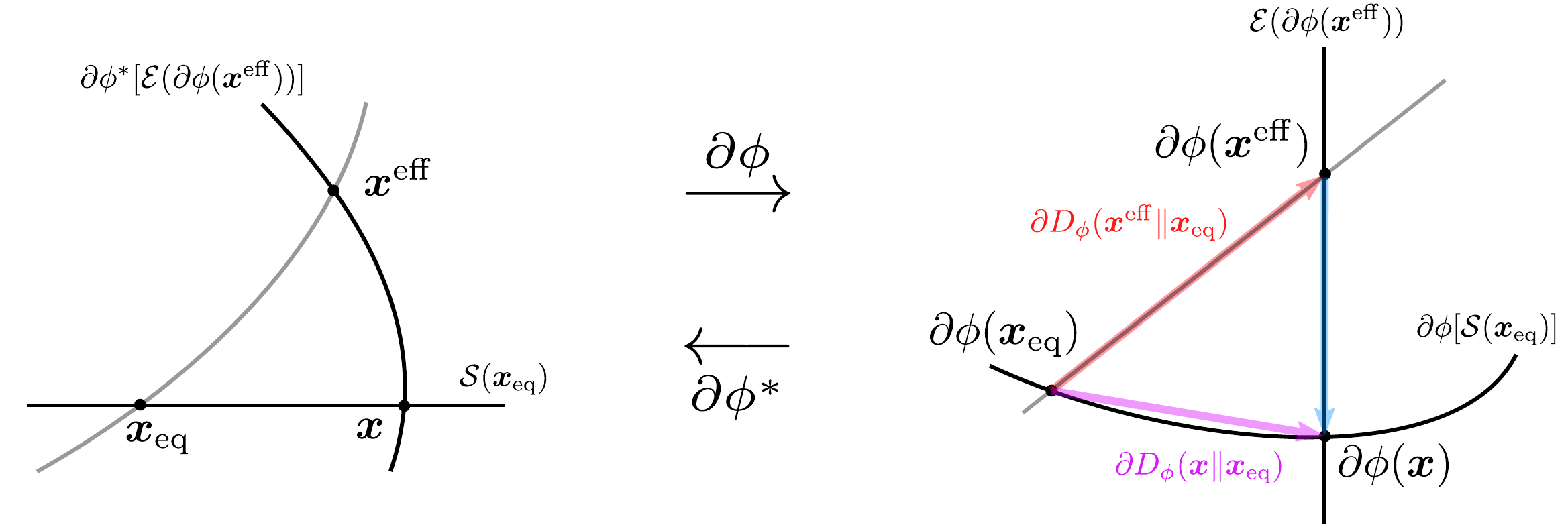}
    \caption{The diagram of the information-geometric relation between $\bm x$, $\bm x_\mathrm{eq}$, and $\bm x^\mathrm{eff}$. The intersection of the stoichiometric subspace $\Stoich(\bm x_\mathrm{eq})$ and the equilibrium submanifold $\partial\phi^\ast[\Equib(\partial \phi(\bm x^\mathrm{eff}))]$ contains the unique point $\bm x$. The same relation between the stoichiometric submanifold $\partial\phi[\Stoich(\bm x_0)]$ and the equilibrium subspace $\Equib(\partial\phi(\bm x^\mathrm{eff}))$ also holds in the potential space $\M$.
    Both $\partial D_\phi(\bm x\| \bm x_\mathrm{eq})$ (purple vector) and $\partial D_\phi(\bm x^\mathrm{eff}\| \bm x_\mathrm{eq})$ (red vector) are mapped to $f(\bm x)$ by $S^\top$.}
    \label{fig:information-geometric relation}
\end{figure}

Then we prove the inequality \cref{eq:lower bound of thermodynamic force norm by convexity}.
If $\phi$ is locally-LC with $L_{\bm x_\mathrm{eq}}(\cdot)$, combining the definition of locally LC \cref{eq:local Lipschitz continuity} and \cref{eq:lower bound of thermodynamic force norm by gradient},
one gets a nonzero lower bound as
\begin{align}
    \label{eq:1st of local bound}
    \norm{\bm f (\bm x)}_2 &\geq \sigma_{\rk(S)} \norm{\partial D_\phi(\bm {x}^\mathrm{eff} \|\bm x_\mathrm{eq})}_2 \\
    \label{eq:2nd of local bound}
    &\geq \sqrt{2 \sigma_{\rk(S)}^2 m_{\bm x_\mathrm{eq}}(\bm x^\mathrm{eff})  D_\phi(\bm{x}^\mathrm{eff} \|\bm x_\mathrm{eq}) } \\
    \label{eq:3rd of local bound}
    &\geq \sqrt{2 \sigma_{\rk(S)}^2 m_{\bm x_\mathrm{eq}}(\bm x^\mathrm{eff})  D_\phi(\bm{x} \|\bm x_\mathrm{eq}) }.
\end{align}
The equality of \cref{eq:1st of local bound} holds iff $\sigma_1=\cdots=\sigma_{\rk S}$.
The equality of \cref{eq:2nd of local bound} holds iff the equality in locally-PL condition \cref{eq:local Polyak-Lojasiewicz} holds. The inequality \cref{eq:3rd of local bound} holds because \cref{eq:Pythagorean relation of effective concentration} and $D_\phi(\bm{x}^\mathrm{eff} \| \bm{x}) \geq 0$. The last equality holds iff $D_\phi(\bm x^\mathrm{eff} \| \bm x) = 0$, which is equivalent to $\partial D_\phi(\bm x \| \bm x_\mathrm{eq}) \in (\Ker S^\top)^\perp$.
\end{proof}

\begin{proof}[Proof of \Cref{thm:upper bound of Bregman divergence in generalized equilibirum flow}]
    For $t \in [0,T]$, the time derivative of $D_\phi(\bm x_t \| \bm x_\mathrm{eq})$  is bounded from above as
    \begin{align}
        \label{eq:global upper 1}
        \dv{}{t} D_\phi(\bm x_t \| \bm x_\mathrm{eq}) 
        &= - \expval{ \partial_{\bm{f}} \Psi_{\bm{x}_t}(\bm{f}(\bm x_t)), \bm{f}(\bm x_t) }\\
        \label{eq:global upper 2}
        &\leq - \low{\omega}(\bm x_t) \low{\B}(\norm{\bm f(\bm x_t)}_2)\\
        \label{eq:global upper 3}
        &\leq - \low{\omega}(\bm x_t) \low{\B}\qty(\sqrt{2\sigma_{\rk(S)}^2 m_{\bm x_\mathrm{eq}}  D_\phi(\bm{x}_t \|\bm x_\mathrm{eq})}).
    \end{align}
    The transformation from \cref{eq:global upper 1} to \cref{eq:global upper 2} follows from \Cref{asm: activity-force decomposability of EPR lower bound}.
    The transformation from \cref{eq:global upper 2} to \cref{eq:global upper 3} follows from \Cref{lemma:lower bounds of thermodynamic force norm}, which also holds in the globally convex setting.

    From \Cref{asm: activity-force decomposability of EPR lower bound}, $\low{\B}$ is positive in $(0,\infty)$ and $2\sigma_{\rk(S)}^2 m_{\bm x_\mathrm{eq}} D_\phi(\bm{x}_t \|\bm x^\ast)$ is positive.
    Thus, 
    \begin{align}
        \frac{2\sigma_{\rk(S)}^2 m_{\bm x_\mathrm{eq}} }{ \qty( \low{\B}(\sqrt{\cdot}) )\qty(2\sigma_{\rk(S)}^2 m_{\bm x_\mathrm{eq}}  D_\phi (\bm{x}_t \|\bm x_{\mathrm{eq}})) } \dv{}{t} D_\phi(\bm x_t \| \bm x_{\mathrm{eq}}) \leq - 2\sigma_{\rk(S)}^2 m_{\bm x_\mathrm{eq}} \low{\omega}(\bm x_t).
    \end{align}
    Integrating with respect to the time $t$ from $0$ to $T$ due to the monotonic non-increasing property of $D_\phi(\bm x_t \| \bm x_\mathrm{eq})$, we obtain:
    \begin{align}
        \begin{aligned}
            % \ln_{\low{\B} (\sqrt{\cdot})} (2\sigma_{\rk(S)}^2 m_{\bm x_\mathrm{eq}}  D_\phi(\bm x_T \| \bm x_{\mathrm{eq}}) ) - \ln_{\low{\B} (\sqrt{\cdot})} (2\sigma_{\rk(S)}^2 m_{\bm x_\mathrm{eq}}  D_\phi(\bm x_0 \| \bm x_{\mathrm{eq}}) ) 
            \qty[\ln_{\low{\B} (\sqrt{\cdot})} (2\sigma_{\rk(S)}^2 m_{\bm x_\mathrm{eq}}  D_\phi(\bm x_t \| \bm x_{\mathrm{eq}}) )]_{t = 0}^T
            &\leq -2\sigma_{\rk(S)}^2 m_{\bm x_\mathrm{eq}}  \int_0^T \low{\omega}(\bm x_t)  \dd t.
            % \\&= - 2\sigma_{\rk(S)}^2 m_{\bm x_\mathrm{eq}} \low{\Omega}(\bm{x}_{0:T}).
        \end{aligned}
    \end{align}
    % where $\ln_{\psi}(u) := \int_1^u \frac{1}{\psi(u')} \dd{u'}$ is the deformed logarithm for a positive function $\psi$ and $\Omega(\bm{x}_{0:T}):= \int_0^T \low{\omega}(\bm x_t)$ is the time-integrated lower bound of activities
    % $\ln_{\low{\B} (\sqrt{\cdot})}$ is invertible
    Since $\exp_{\low{\B} (\sqrt{\cdot})}$ is  monotonically increasing, we obtain the target inequality \cref{eq:upper bound of Bregman divergence in generalized equilibirum flow}.
    % \begin{align}
    %     D_\phi(\bm x_T \| \bm x_{\mathrm{eq}}) \leq \frac{\exp_{\low{\B} (\sqrt{\cdot})} \qty[ \ln_{\low{\B} (\sqrt{\cdot})} (2\sigma_{\rk(S)}^2 m_{\bm x_\mathrm{eq}}  D_\phi(\bm x_0 \| \bm x_{\mathrm{eq}}) ) - 2\sigma_{\rk(S)}^2 m_{\bm x_\mathrm{eq}} \low{\Omega}(\bm{x}_{0:T}) ]}{2\sigma_{\rk(S)}^2 m_{\bm x_\mathrm{eq}} }.
    % \end{align}
\end{proof}

\section{Simulation Results}
We consider some CRNs to evaluate the performance of the bounds given in \Cref{thm:bounds of Bregman divergence in generalized equilibirum flow,thm:bounds under locally convexity assumption}.

\begin{KSblock}
\subsection{\KS{Bounds in Homogeneous Reversible Gray-Scott model}}
We first consider the homogeneous reversible Gray-Scott model in \Cref{eg:gray scott model,eg:bounds in gray scott}. 

\begin{KSTwoblock}
We summarize the numerical setting. We fix the positive detailed-balance equilibrium $\hat{\bm x}=(1.2,1.0,0.9,0.9)^\top$ and choose the initial condition $\bm x_0=(1.4,0.7,0.9,1.0)^\top$. The forward rate constants are $(k_1^+,k_2^+,k_3^+)=(0.20,0.18,0.15)$, and detailed balance determines the backward rate constants as $(k_1^-,k_2^-,k_3^-)=(0.24,0.20,0.20)$. Since $\bm 1^\top\bm x_0=\bm 1^\top\hat{\bm x}=4$, the reference equilibrium $\bm x_{\mathrm{eq}}=\hat{\bm x}$ lies in the stoichiometric class of $\bm x_0$. The mass-action rate equation is integrated on $0\le t\le 15$ with fixed step size $\Delta t=0.005$ by the explicit fourth-order Runge--Kutta method.

Along the resulting trajectory, we evaluate $D_2(\bm x_t\|\bm x_{\mathrm{eq}})$ together with its global and local upper and lower bounds (\Cref{fig:gray scott}(a)). The bounds enclose $D_2(\bm x_t\|\bm x_{\mathrm{eq}})$ at all plotted times. The divergence decreases monotonically from $6.72\times10^{-2}$ to $5.92\times10^{-4}$. Its semi-log curve becomes less steep with time. Near the final plotted time, the local and global upper bounds are approximately $4.3$ and $5.7$ times $D_2$, respectively, whereas the local and global lower bounds are approximately $1.7\times10^8$ and $4.7\times10^8$ times smaller than $D_2$, respectively.

We also evaluate the corresponding EPR bounds by following the successive estimates in \cref{eq:activity estimate,eq:force-to-gradient estimate,eq:convexity estimate} (\Cref{fig:gray scott}(b)). The gaps between these estimates provide a posteriori diagnostic of how looseness accumulates at each step. Near the final plotted time, the largest gap in the upper EPR estimate occurs between \cref{eq:activity estimate} and \cref{eq:force-to-gradient estimate}, where the upper estimate increases by a factor of approximately $5.8$. The same step has little effect on the lower estimate at that time. Thus, the main source of looseness depends on whether the upper or lower EPR bound is considered.
\end{KSTwoblock}

\begin{figure}
    \centering
    \includegraphics[width=\linewidth]{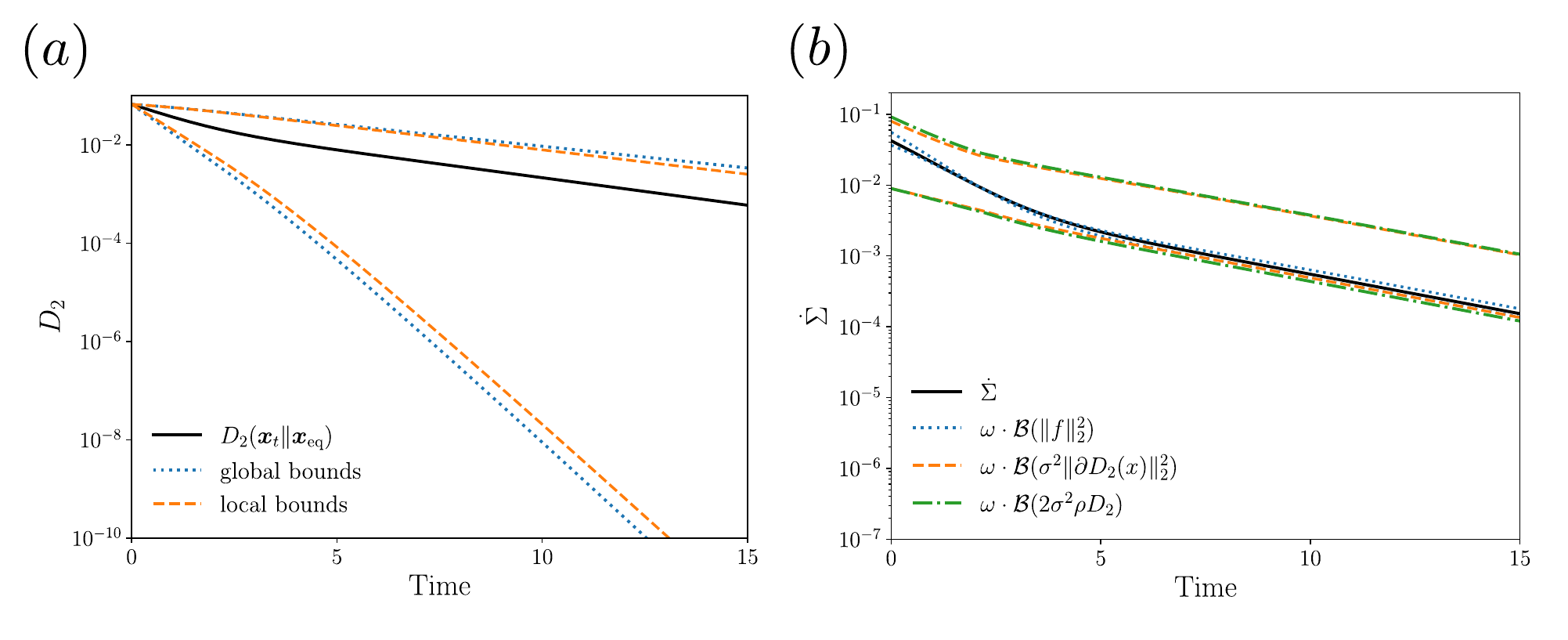}
    \caption{
    \KS{Time evolution of the quadratic divergence, the EPR, and their bounds in the homogeneous reversible Gray--Scott model.
    (a) The quadratic divergence $D_2(\bm{x}_t\|\bm{x}_{\mathrm{eq}})$ (black solid curve) and its upper and lower bounds obtained using global convexity (blue dotted curves) and local convexity (orange dashed curves), as given in \Cref{thm:bounds of Bregman divergence in generalized equilibirum flow,thm:bounds under locally convexity assumption}.
    (b) The EPR $\dot{\Sigma}$ (black solid curve) and its successive upper and lower bounds obtained from the quadratic dissipation estimate \cref{eq:activity estimate} (blue dotted curves), the force-to-gradient estimate \cref{eq:force-to-gradient estimate} (orange dashed curves), and the convexity estimate \cref{eq:convexity estimate} (green dash-dotted curves).}
    }
    \label{fig:gray scott}
\end{figure}
\end{KSblock}

\subsection{Bounds Characterizing Slow Relaxation in Catalytic CRNs}
\label{sec:convergence analysis and characterization of slow relaxation in mass-action reaction systems}

%% Awazu&Kaneko model

Additionally, we characterize the slow relaxation in equilibrium CRNs through our bounds of the thermodynamic quantities.
% and divergence.

We use the catalytic CRNs \cite{Awazu2009}, which involve the 2nd-order catalytic reactions,
\begin{align}
    X_i + X_c \rightleftharpoons X_j + X_c,
\end{align}
and follow mass-action kinetics.
% $X_i + X_c \rightleftharpoons X_r + X_c \, (1\leq i,j,c \leq n)$.
It is reported by Awazu and Kaneko \cite{Awazu2009} that the deviation from the equilibrium, defined as $C(t) := \frac{\expval{\bm x_t - \bm x_\mathrm{eq}, \bm x_0 - \bm x_\mathrm{eq}} }{\|\bm x_0 - \bm x_\mathrm{eq}\|^2}$, exhibits some plateaus over in the time courses. 
% The paper \cite{Awazu2009} states that the emergence of the plateaus is attributed to the kinetic constraints, particularly negative correlations between chemical species.
We utilize the CRN1,
\begin{align}
\begin{aligned}
    \label{eq:CRN1}
    X_0 + X_4 \rightleftharpoons X_1 + X_4,\\
    X_0 + X_1 \rightleftharpoons X_2 + X_1,\\
    X_0 + X_2 \rightleftharpoons X_3 + X_2,\\
    X_0 + X_1 \rightleftharpoons X_4 + X_1.
\end{aligned}
\end{align}
% and the CRN2,
% \begin{align}
% \begin{aligned}
%     \label{eq:CRN2}
%     X_0 + X_2 \rightleftharpoons X_1 + X_2,\\
%     X_0 + X_4 \rightleftharpoons X_2 + X_4,\\
%     X_0 + X_1 \rightleftharpoons X_3 + X_1,\\
%     X_0 + X_1 \rightleftharpoons X_4 + X_1,   
% \end{aligned}
% \end{align}
% both of 
\KS{
We prescribe the system parameter $\hat{\bm x}$ by $\hat x_i=e^{-16i/N}$ for $i=0,\ldots,N-1$ with $N=5$, following the parametrization used in \cite{Awazu2009}. The rate constant of the reaction $X_i+X_c\to X_j+X_c$ is chosen as $k_{i,j}=\min\qty{1,\hat x_j/\hat x_i}$. We set the initial condition to $\bm x_0=(1,1,1,1,1)^\top$. Since the total concentration is conserved, the equilibrium state in the stoichiometric subspace of $\bm x_0$ is obtained by rescaling $\hat{\bm x}$ as $\bm x_{\mathrm{eq}}=(|\norm{\bm x_0}_1/\norm{\hat{\bm x}}_1) \cdot \hat{\bm x}$. The trajectories are computed on the time interval $0\le t\le 10^5$ by an explicit fourth-order Runge--Kutta scheme with a fixed time step $\Delta t=0.05$.}
The properties noted in this section are also observed for the other CRN, which has the same stoichiometric matrix as CRN1 \cref{eq:CRN1}, with different reaction catalysts. See \Cref{sec:simulation results for the other plateau-exhibiting CRN}.
% The time evolutions of the concentration in CRN1 are shown in \Cref{fig:numerical simulation for slow relation} (a).

%% observation of divergence and bounds
We use the KL divergence $D_\mathrm{KL}(\bm x_t \| \bm x_\mathrm{eq})$ instead of $C(t)$ to quantify how close the concentration $\bm x_t$ is to the equilibrium state $\bm x_\mathrm{eq}$.
In CRN1 \ref{eq:CRN1}, the KL divergence is also observed to exhibit plateaus over time (\Cref{fig:numerical simulation for slow relation}(a)). 
The upper bounds ($\up{D}_\mathrm{quad, glob}$, $\up{D}_\mathrm{quad, loc}$, $\up{D}_\mathrm{hyp, glob}$, and $\up{D}_\mathrm{hyp, loc}$) are $10^0-10^3$ larger than the divergence.
On the other hand, compared to the upper bounds, the lower bounds under the local-LC assumption are observed to rapidly decrease in the $\log$-$\log$ plot (\Cref{fig:numerical simulation for slow relation}(b)).

In this system, the upper bound is likely to be close to the original divergence because of the following two reasons:
\begin{itemize}
    \item We show the time series of the force norm and its bounds in \Cref{fig:numerical simulation for slow relation}(c). The lower bound observed to be closer to the original norm than the upper bound. One potential reason is that the singular values of the CRN1 \cref{eq:CRN1}, $\sigma_1 = \sqrt{5}$ and $\sigma_2 = \sigma_3 = \sigma_4 = 1$, are biased toward the minimal singular value $1$.
    \item The EPR bounds as functions of $\bm \omega(\bm x)$ and $\bm f(\bm x)$ are observed in \Cref{fig:numerical simulation for slow relation}(d).
    The lower bounds are closer to the original EPR than the upper bounds.
\end{itemize}

\KS{
As in the Gray--Scott example, we also examined the three-step bounds on the EPR with quadratic dissipation in the Awazu--Kaneko model.
Here, the dominant loss occurs in the activity-estimate step (\Cref{fig:EPR quad stages}).
This differs from the Gray--Scott one: in the Awazu--Kaneko model, the activities are separated over much broader scales, and this heterogeneity is associated with slow relaxation, observed as plateaus in the KL divergence.
In particular, the heterogeneity of activities affects the KL lower bound mainly by worsening the EPR upper bound, rather than the EPR lower bound (\Cref{fig:numerical simulation for slow relation} (b)).
Thus, the EPR bounds help identify the main source of looseness in each system.
}

We can asymptotically explain that the lower bound become too small in reaction systems with varying timescales. 
More generally, for the quadratic dissipation, the ratio between the upper and lower divergence bounds depends exponentially on $\up{\Omega}_\mathrm{log} - \low{\Omega}_\mathrm{log} = \int_0^t [\max_r \omega_r - \min_r \omega_r] \dd \tau$ as $t \to \infty$: $\up{D}_\mathrm{quad}/\low{D}_{\mathrm{quad}} \propto \exp(-(\low{\Omega}_\mathrm{log} - \up{\Omega}_\mathrm{log})) = \exp(\up{\Omega}_\mathrm{log} - \low{\Omega}_\mathrm{log})$.
Therefore, if the orders of the activities differ significantly, both bounds should not be close to the original divergence.

%% quadratic vs. hyperbolic bounds
We compare the hyperbolic and quadratic bounds in the CRN.
The quadratic upper bounds outperform the hyperbolic ones in the global and local cases.
On the other hand, more quantitatively, their differences are less than one-tenth of one percent.
% For the lower bound, while their decreasing speeds are different, both bounds rapidly become small to the divergence.
Both the quadratic and hyperbolic lower bounds decrease rapidly compared to the original divergence.

%%convexity and plateaus
Finally, we note the relationship between convexity and slow relaxation.
\KS{While the validity of the upper bounds is guaranteed by our theoretical results, their agreement with the observed plateau-like behavior is a numerical observation.}
Under local convexity conditions,\KSsout{ plateaus also emerge in the time evolution of the upper bounds} \KS{the upper bounds numerically track the plateau-like behavior}.
In contrast, no plateau is observed for the upper bounds under global convexity.\KSsout{ This result indicates that it is necessary to consider the state-dependent convexity of the thermodynamic function in order to capture plateaus.} \KS{Because the local bounds retain state-dependent convexity evaluated along the actual trajectory, this agreement should be interpreted as an a posteriori diagnostic feature rather than as an independent prediction of plateau onset, height, or duration.}
Furthermore, recalling the form of the upper bound under\KSsout{ locally PL} \KS{the local PL condition},\KSsout{ it can be said that the information on the minimum eigenvalue of $S$, the minimum activities, and the local convexity is sufficient to explain the formation of plateaus, at least for this system.} \KS{the minimum nonzero singular value of $S$, the minimum activities, and the local-convexity factor provide a diagnostic decomposition of the plateau-like relaxation observed along the computed trajectory, at least for this system.}

\begin{figure}[htbp]
    \centering
    \includegraphics[width=\linewidth]{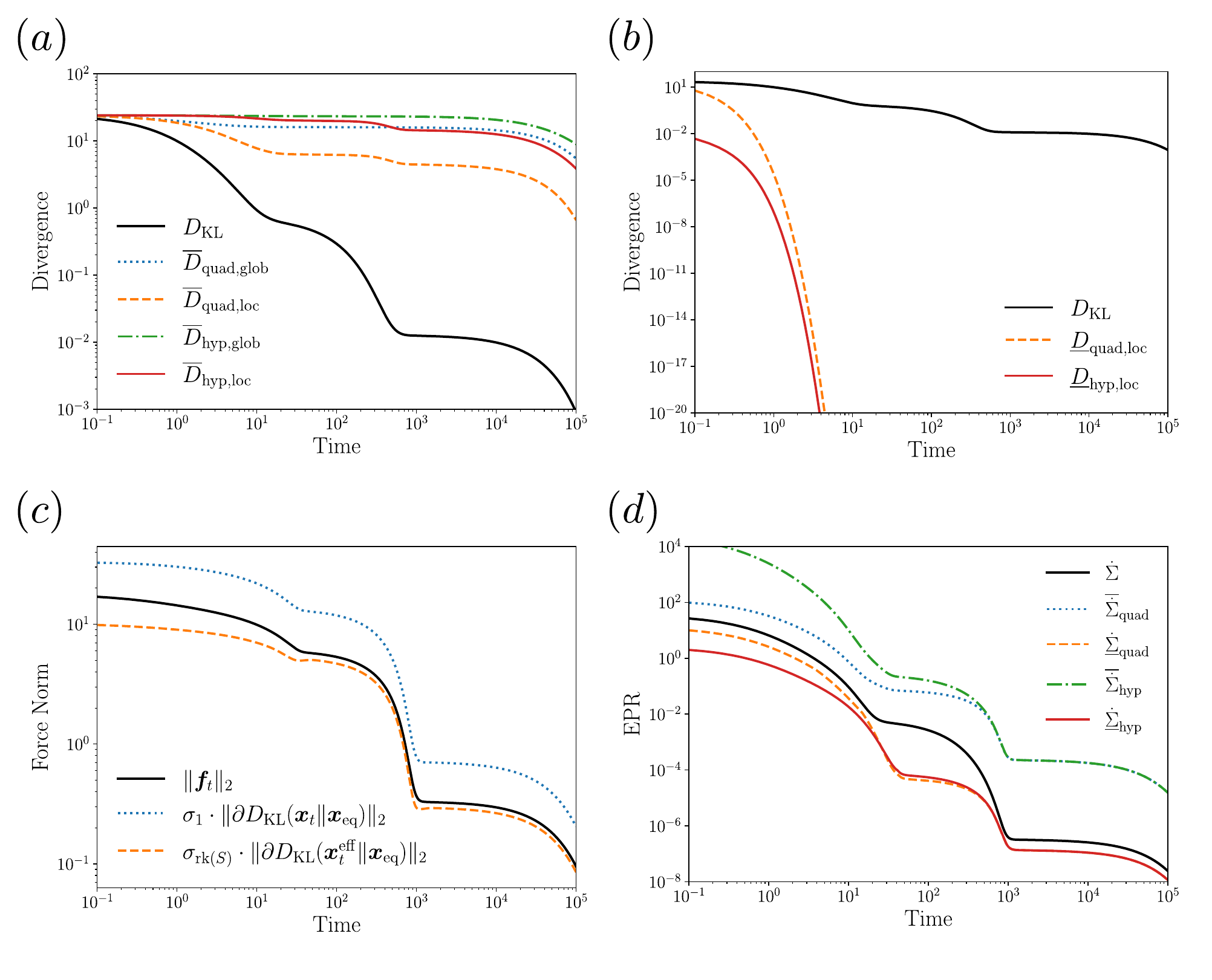}
    \caption{The time courses of the generalized KL divergence and the corresponding bounds in the CRN1 \cref{eq:CRN1} $(0\leq t\leq T = 10^5)$. (a) The time courses of the divergence $D_\mathrm{KL}(\bm x_t \| \bm x_\mathrm{eq})$ and its bounds. We show the local upper bounds $\up{D}_\mathrm{hyp,loc}$ \cref{eq:hyp_loc_upper}, $\up{D}_\mathrm{quad,loc}$ \cref{eq:quad_loc_upper} for $\rho_{\bm x_\mathrm{eq}}^\mathrm{KL}(\bm x)$, and the global upper bounds $\up{D}_\mathrm{hyp,glob}, \up{D}_\mathrm{quad,glob}$ for $\low{\rho}_{\bm x_\mathrm{eq}}^\mathrm{KL} = \min_{\bm x\in \mathcal{O}_{10^5}(\bm x_0)} \rho_{\bm x_\mathrm{eq}}^\mathrm{KL}(\bm x)$. 
    % We exploit locally-PL and locally-LC conditions with respect to $m_{\bm x_\mathrm{eq}}(\bm x) = L_{\bm x_\mathrm{eq}}(\bm x) = \rho_{\bm x_\mathrm{eq}}^\mathrm{KL}(\bm x)$, and the global convexity parameters $\low{m_{\bm x_\mathrm{eq}}} = \min_{\bm x\in \mathcal{O}_{10^5}(\bm x_0)} \rho_{\bm x_\mathrm{eq}}^\mathrm{KL}(\bm x)$ and $\up{L_{\bm x_\mathrm{eq}}}  = \max_{\bm x\in \mathcal{O}_{10^5}(\bm x_0)} \rho_{\bm x_\mathrm{eq}}^\mathrm{KL}(\bm x)$, respectively.
    (b) The time courses of the KL divergence, the local lower bounds $\low{D}_\mathrm{hyp,loc} \cref{eq:hyp_loc_lower}, \low{D}_\mathrm{quad,loc}$ \cref{eq:hyp_loc_lower} for $\rho_{\bm x_\mathrm{eq}}^\mathrm{KL}(\bm x)$. 
    (c) The time courses of the force norm $\norm{\bm f (\bm x_t)}_2$, its upper bound \cref{eq:upper bound of thermodynamic force norm by gradient} and lower bound \cref{eq:lower bound of thermodynamic force norm by gradient}. 
    (d) The time courses of the EPR $\dot\Sigma(\bm x_t, \bm f(\bm x_t))$, the quadratic upper bound $\up{\dot \Sigma}_\mathrm{quad}$ and lower bound  $\low{\dot \Sigma}_\mathrm{quad}$ \cref{eq:quadratic EPR bounds}, and the hyperbolic upper bound $\up{\dot \Sigma}_\mathrm{hyp}$ and lower bound $\low{\dot \Sigma}_\mathrm{hyp}$ \cref{eq:hyperbolic EPR bounds}.}
    \label{fig:numerical simulation for slow relation}
\end{figure}

\begin{figure}
    \centering
    \includegraphics[width=.55\linewidth]{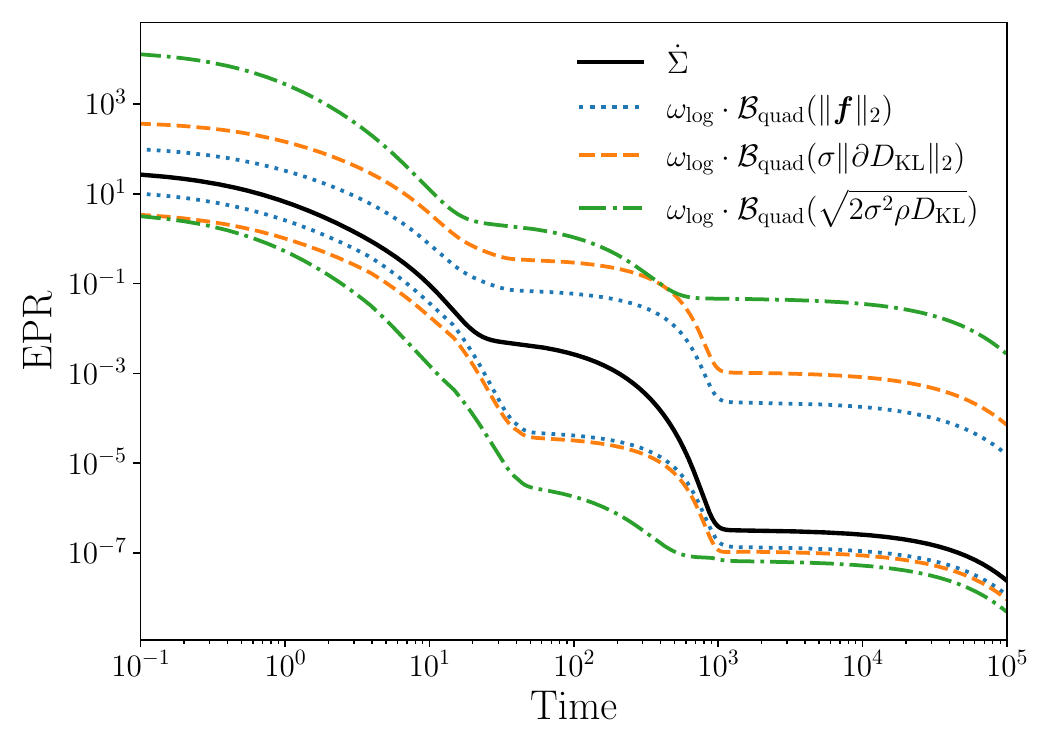}
    \caption{\KS{The time course of the EPR, $\dot{\Sigma}$, and its successive bounds. The green dotted curves show the quadratic dissipation bounds in \Cref{eq:quadratic EPR bounds} evaluated at $\bm f=\bm f(\bm x)$; the orange dashed curves are obtained by bounding the force norm in terms of the KL-divergence gradient norm; and the green dash-dotted curves follow by further bounding the gradient norm using the global convexity.}}
    \label{fig:EPR quad stages}
\end{figure}

\begin{KSblock}
\subsection{Upper Bounds on the Relaxation of Fractions in the Hatakeyama-Kaneko MWC model}
\label{sec:bounds in MWC model}
As a more realistic example, we consider an enzymatic-kinetically constrained variant  of the Monod--Wyman--Changeux (MWC) model, a classical equilibrium model for concerted allosteric regulation.
In this model, allostery is represented by a coupling between the modification of multiple binding sites and a tense--relaxed conformational transition of a protein.
While this coupling is traditionally encoded in the equilibrium landscape, Hatakeyama and Kaneko introduced an additional kinetic effect by incorporating enzymatic reactions through Michaelis--Menten-type enzyme-binding probabilities.
They showed that competition for a limited amount of enzyme can transiently suppress modification reactions and induce slow relaxation with plateaus in two observables, such as the tense-state fraction and the unmodified-site fraction.

In their ODE model, the chemical species $T_i$ and $R_i$ represent proteins in the tense and relaxed conformations, respectively, with $i$ modified sites.
The CRN consists of $2(m+1)$ species and $3m+1$ reversible reaction pairs.
Specifically, it contains modification/demodification reactions
$T_i \rightleftharpoons T_{i+1}$ and $R_i \rightleftharpoons R_{i+1}$ $(i=0,\ldots,m-1)$, and conformational switching reactions
$T_i \rightleftharpoons R_i$ $(i=0,\ldots,m)$.
Here, $m$ denotes the maximum number of modification sites.
The concentration vector is
$
\bm x=(x_{T_0},\ldots,x_{T_m},x_{R_0},\ldots,x_{R_m})^\top .
$
The enzymatic effect is incorporated into the effective modification and demodification rates through Michaelis--Menten-type factors, rather than by treating the enzyme as an independent dynamical species.
The detailed kinetics, parameter values, initial condition, and numerical implementation are provided in \Cref{sec:numerical settings in MWC model}.

For this system, our bounds also provide estimates on the relaxation of biologically interpretable observables.
We consider two fractions: the fraction of tense-state molecules,
$
\mathcal{T}(\bm x)
:=
{\sum_i x_{T_i}}/{\|\bm x\|_1},
$
and the fraction of unmodified sites,
$
\mathcal{U}(\bm x)
:=
{\sum_i (m-i)(x_{T_i}+x_{R_i})}/{(m\|\bm x\|_1)}
$, where $\norm{\bm x}_1$ denotes the total number of molecules: $\|\bm x\|_1 = \sum_i (x_{T_i} + x_{R_i})$.
Their relaxation is quantified by the KL divergence between Bernoulli distributions associated with the current and equilibrium values of each fraction.
For $a,a'\in (0,1)$, we define
\begin{align}
D_{\mathrm{KL}}
\left(
\mathrm{Ber}(a)
\middle|
\mathrm{Ber}(a')
\right)
:=
a\ln\frac{a}{a'}
+
(1-a)\ln\frac{1-a}{1-a'}.
\end{align}
We then evaluate $D_{\mathrm{KL}}
\qty(
\mathrm{Ber}(\mathcal{T}(\bm x_t))
\|
\mathrm{Ber}(\mathcal{T}(\bm x_{\mathrm{eq}}))
)$ and $D_{\mathrm{KL}}
\qty(
\mathrm{Ber}(\mathcal{U}(\bm x_t))
\|
\mathrm{Ber}(\mathcal{U}(\bm x_{\mathrm{eq}}))
)$.
Although the Bernoulli KL divergences are not standard observables in MWC models, they bound the squared distances of the fractions from the equilibrium ones through Pinsker's inequality, namely, $2 \cdot \abs{\mathcal{T}(\bm x_t)-\mathcal{T}(\bm x_{\mathrm{eq}})}^2$ and $2 \cdot \abs{\mathcal{U}(\bm x_t)-\mathcal{U}(\bm x_{\mathrm{eq}})}^2$, respectively.
Moreover, we find that both Bernoulli KL divergences are bounded above by the KL divergence of the CRN:
\begin{align}
D_{\mathrm{KL}}
\qty(
\mathrm{Ber}(\mathcal{T}(\bm x_t))
\|
\mathrm{Ber}(\mathcal{T}(\bm x_{\mathrm{eq}}))
)
\leq
\frac{D_{\mathrm{KL}}(\bm x_t|\bm x_{\mathrm{eq}})}{\norm{\bm x_0}_1},\\
D_{\mathrm{KL}}
\qty(
\mathrm{Ber}(\mathcal{U}(\bm x_t))
\|
\mathrm{Ber}(\mathcal{U}(\bm x_{\mathrm{eq}}))
)
\leq
\frac{D_{\mathrm{KL}}(\bm x_t|\bm x_{\mathrm{eq}})}{\norm{\bm x_0}_1}.
\end{align}
Thus, our upper bounds on the full KL divergence also yield upper bounds on the relaxation of these fractions.
The proof follows directly from the log-sum inequality and is given in \Cref{sec:proof of upper bounds in MWC model}.

We numerically confirm the inequalities and the resulting bounds for the MWC model.
The Bernoulli KL divergences of the two fractions remain below the KL divergence and its upper bounds along the trajectory (\Cref{fig:MWC bounds}).
Although conservative, these estimates decrease for $t\geq 10^2$ and capture the qualitative relaxation of the two observables.
The decomposition of the EPR bound shows that the activity estimate is the main source of looseness in the present example, and this loss may become more pronounced in more complex regimes, such as the $m=6$ systems studied in the original work; see \Cref{sec:MWC EPR bounds main,sec:MWC bounds original}

\begin{figure}
    \centering
    \includegraphics[width=.55\linewidth]{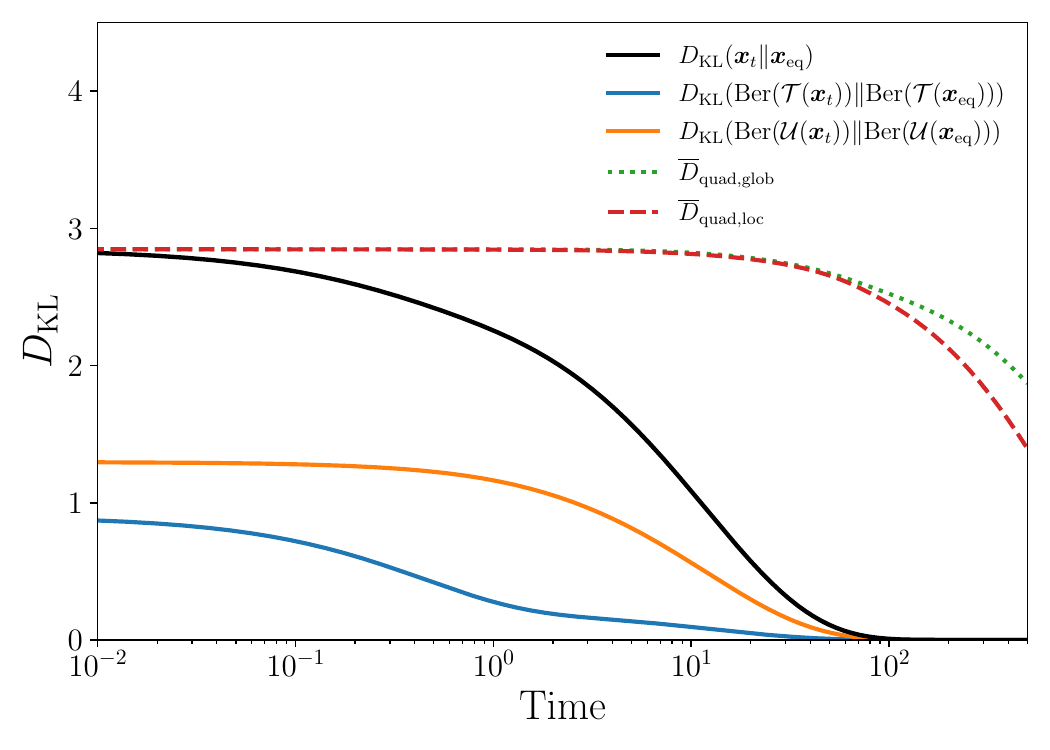}
    \caption{
    \KS{
    The time series of the KL divergence $D_{\mathrm{KL}}(\bm x_t|\bm x_{\mathrm{eq}})$, the Bernoulli KL divergences $D_{\mathrm{KL}}
\qty(
\mathrm{Ber}(\mathcal{T}(\bm x_t))
\|
\mathrm{Ber}(\mathcal{T}(\bm x_{\mathrm{eq}}))
)$ and $D_{\mathrm{KL}}
\qty(
\mathrm{Ber}(\mathcal{U}(\bm x_t))
\|
\mathrm{Ber}(\mathcal{U}(\bm x_{\mathrm{eq}}))
)$, and the global and local upper bounds $\overline{D}_{\mathrm{quad,glob}}, \overline{D}_{\mathrm{quad,loc}}$ on the KL divergence in the Hatakeyama--Kaneko MWC model. Detailed numerical conditions are provided in \Cref{sec:numerical settings in MWC model}.
    }}
    \label{fig:MWC bounds}
\end{figure}

\end{KSblock}

\section{Conclusion and Discussion}
\label{sec:Conclusion and Discussion}
In this paper, we have analytically derived \KS{trajectory-dependent} upper/lower bounds, characterized by the stoichiometric matrix, thermodynamic functions, and dissipative structure of the CRN dynamics, and applied them to the slowly relaxing CRNs.
To evaluate the divergence from above and below, we have exploited generalized gradient flow structure of equilibrium CRNs and quantified the convexity of thermodynamic functions.
% In particular, our upper bounds not only provide quantitative limits for relaxation but also reveal the importance of local convexity in thermodynamic functions for relaxation plateaus.
In particular, our upper bounds not only provide \KS{trajectory-dependent, a posteriori} quantitative limits for relaxation but also clarify the role of local convexity in thermodynamic functions in the emergence of relaxation plateaus. 
To the best of our knowledge, this is the first quantitative\KSsout{ analytical} \KS{diagnostic} characterization of such plateaus within an analytic framework, marking a pioneering step toward an analytical theory of slow relaxation.
% \KS{[stressed important results of the plateaus]}

\KS{
We compare our force-space formulation with reaction trajectory approaches \cite{Wang2020, Liu2025}. 
In such approaches, the variables are the time-integrated reaction fluxes $\bm R_t$, and the concentration is represented as $\bm x_t = \bm x_0 + S\bm R_t$.
The dynamics can then be viewed as a gradient flow of the free energy with respect to $\bm R_t$. The reaction-trajectory viewpoint is closely related to our force-space formulation:
by the chain rule, the variation of the free energy with respect to $\bm R_t$ is given by $S^\top \partial_{\bm x} D_\phi(\bm x \| \bm x_\mathrm{eq})$, which corresponds to the force $\bm f$. 
In this work, we adopt the force-space formulation because the explicit expression $\bm f=S^\top \partial_{\bm x}D_\phi$ makes the roles of stoichiometry and free-energy gradients transparent, which is essential for deriving our convex-analytic bounds.
}

% A notable difference in the convergence rate of CRNs compared to convex optimization methods, such as Gradient Descent \cite{garrigos2024} or Gradient Flow \cite{ang2020convergence}, is that the time $t$ is replaced by the time-integrated activities. 
Unlike the convergence rate of convex optimization methods like Gradient Descent \cite{Nesterov2018,Polyak2021,Boyd_Vandenberghe_2004}, the CRN convergence depends on time-integrated activities instead of time.
This difference comes from the multi-timescale nature of chemical reactions in CRNs. Timescale separation (i.e. fast-slow separation or quasi-steady-state approximation) in multi-timescale systems treats the fast variables as constants and reduces them to the system that is closed with respect to the slow variables \cite{Strogatz2019}. This may be represented by taking the limit $\low{\Omega} \to 0$ or $\up{\Omega} \to \infty$ within our bounds. 
\KSsout{Consequently, for the CRNs, our lower bounds provide less information than our upper ones.} 
\KS{Our stagewise numerical analysis shows that, in the biologically motivated Awazu--Kaneko and Hatakeyama--Kaneko models, replacing reaction-wise activities spanning widely separated scales by their minimum or maximum is the dominant source of looseness. Possible improvements include applying the bounds after an appropriate timescale reduction or exploiting model-specific estimates that avoid such extremal activity bounds; developing these refinements is left for future work.}

% \ch{In our method, we translate the cause of the convergence of the KL divergence to the width of Activity. Note, on the other hand, that no explanation is given as to \textbf{why} Activity varies.}

% \KSsout{We note the potential difficulty in testing the bounds in CRNs. 
% That is, the stoichiometric conditions of CRNs can be controlled in numerical experiments, but it is difficult to control their kinetics.
% In general, in physically and chemically natural kinetics---such as mass-action kinetics, Michaelis-Menten kinetics \cite{michel2013}, and Hill-type kinetics \cite{Weiss1997,Hernandez2021}---the reaction rate $j$ is proportional to the stoichiometry-dependent term: $j^\pm(\bm x) = g(\bm x) \cdot \prod_i x_i^{s_{ij}^\pm}$ \cite{Kobayashi2023}.
% For systems that have the same stoichiometric matrix, we can consider various kinetics by the difference in the catalytic mechanisms, which is corresponding to the change in $g(\bm x)$.
% For example, we compare the two CRNs with the same stoichiometry but with the different catalysts in \Cref{sec:simulation results for the other plateau-exhibiting CRN}. 
% Conversely, we cannot naturally set CRNs with different kinetics and the same stoichiometry, due to the stoichiometry-dependence of kinetics.}
\KS{
A practical difficulty in systematically assessing our bounds is that stoichiometric and kinetic effects cannot be varied independently in physically and chemically natural CRNs. Although the stoichiometric matrix can be prescribed directly in numerical experiments, the concentration dependence of the reaction rates is partly constrained by the stoichiometry. For many common kinetic laws, including mass-action, Michaelis--Menten \cite{michel2013}, and Hill-type kinetics \cite{Weiss1997,Hernandez2021}, the forward and backward rates take the form
$j_j^\pm(\bm x)=g_j^\pm(\bm x)\prod_i x_i^{s_{ij}^\pm}$,
where the monomial factor is determined by the reactant stoichiometry and $g_j^\pm(\bm x)$ represents additional kinetic effects arising from catalytic or regulatory mechanisms \cite{Kobayashi2023}. 
Thus, the kinetics cannot be varied arbitrarily while the stoichiometry is kept fixed. Nevertheless, kinetic effects can be partially isolated by comparing CRNs with the same stoichiometric matrix but different catalytic factors $g_j^\pm(\bm x)$. Following this strategy, in \Cref{sec:simulation results for the other plateau-exhibiting CRN}, we compare two CRNs with identical stoichiometry but different catalysts.
}

We discuss two \KS{theoretical} directions for improving the performance of divergence bounds \KS{in equilibrium systems}.
One is to extend our approach to a wider class of dissipative structures.
Dissipative structures are not limited to quadratic or hyperbolic forms; they can be formulated for infinitely many types of reaction activities \cite{Nagayama2025}.
% Specifically, it is necessary to explore alternative lower bounds for non-convex dissipation functions because \Cref{lemma:bounds of EPR} is not directly applicable to that class.
Another potential direction is to investigate bounds by norms with more favorable properties than the $p$-norms. For example, the Luxemburg norm and the Orlicz norm \cite{krasnoselskij1960convex} 
% --- both defined for Young functions, a class of convex functions that diverge rapidly enough at infinity --- 
are candidates for norms compatible to dissipative structures.

\KS{
We emphasize that the present analysis is restricted to equilibrium CRNs: our bounds concern the KL divergence, or its Bregman-divergence analogue in generalized gradient flows, and do not directly cover CRNs with genuine nonequilibrium stationary states.
Nevertheless, the form of the estimates, based on reaction activities, singular values of the stoichiometric matrix, and convexity of the relevant divergence, suggests possible extensions to nonequilibrium steady state systems with analogous relaxation structures.
For example, in time-homogeneous Markov jump processes with a nonequilibrium stationary distribution, the excess or nonadiabatic entropy production is related to the decay of the KL divergence from the stationary distribution \cite{Hatano2001,Esposito2010,Ge2010}.
Complex-balanced mass-action systems provide another direction, since they generalize detailed-balanced CRNs while retaining a KL-type Lyapunov function \cite{Horn1972,Feinberg_2019}.
Investigating whether such systems admit analogous activity--singular-value--convexity bounds is left for future work.
}

% \clearpage

% \newpage

% \appendix

\section*{Acknowledgments}
We are deeply grateful to Matthias Liero for valuable discussions and for sharing his expertise on convergence in reaction–diffusion systems. 
The first author used generative AI tools—ChatGPT (OpenAI) and DeepL Translator (DeepL SE)—to assist with language editing (translation, spelling, grammar, and general style).

\bibliographystyle{siamplain}
\bibliography{reference}
\end{document}

\maketitle
\section{Standard Notations}
\label{sec:standard notations}
In this article, we use the following notation and standard facts on norms and convexity:

For $\bm a \in \R^N$ and $g: \R \to \R$, the vector $\bm g(\bm a)$ represents $(\bm g(\bm a))_i := g(a_i)$.
 Denote the absolute value function by $\abs{\cdot}: \R \to \R_{\geq 0}$.
 Denote the $p$-norm of $\bm a$ ($1\leq p < \infty$) on $\R^N$ as
$
    \norm{\bm a}_p := \qty(\abs{a_1}^p + \abs{a_2}^p + \cdots + \abs{a_n}^p)^{{1}/{p}}
$.
The
$\infty$-norm is defined as
$
    \norm{\bm a}_\infty := \lim_{p\to \infty} \norm{\bm a}_p = \max_{1 \leq i \leq n} \abs{a_i}.
$
The $p$-norm is non-increasing with respect to $p$.
It means that, for any $\bm a \in \R^N$ and $1 \leq i \leq n$, the following inequalities hold:
$\abs{a_i} \leq \norm{\bm a}_\infty \leq \norm{\bm a}_2 \leq \norm{\bm a}_1$.
 For $\bm a \in \R_{>0}^n$, the minimum element of $\abs{\bm a}$ is represented as
$
    \min_{1\leq i\leq n} a_i = \lim_{p\to -\infty} \norm{\bm a}_p,
$
which is formally denoted as $\norm{\bm a}_{-\infty}$. 
We note that $\norm{\cdot }_{-\infty}$ is not a norm.

A subset $\X \subseteq \R^N$ is called convex if, for any $\bm x, \bm y \in \mathcal{X}$ and $\lambda \in [0,1]$, $(1-\lambda)\bm{x} + \lambda \bm{y}$ is in $\X$. 
For a convex subset $\X$ of $\R^N$, a function $\phi : \X \to \R$ is called convex if, for any $\lambda \in [0,1]$ and any $\bm x, \bm y \in \mathcal{X}$, the following condition holds:
$
    \phi((1-\lambda)\bm{x} + \lambda \bm{y}) \leq (1-\lambda)\phi(\bm x) + \lambda\phi(\bm y).
$
For a differentiable function $\phi:\X \to \R$, the following conditions (i)-(iii) are equivalent: 
(i) $\phi$ is convex,
(ii) for any $\bm x, \bm y \in \X$, $\phi(\bm x) \geq \phi(\bm y) + \expval{\partial\phi(\bm y), \bm x - \bm y }$ holds,
(iii) for any $\bm x \in \X$, $\partial^2\phi(\bm x) \succeq O$.
 A function $\phi: \X \to \R$ is called strictly convex if, for any $\lambda \in (0,1)$ and any $\bm x, \bm y \in \mathcal{X}$ such that $\bm x \neq \bm y$, the following condition holds:
$
    \phi((1-\lambda)\bm{x} + \lambda \bm{y}) < (1-\lambda)\phi(\bm x) + \lambda\phi(\bm y).
$
For a differentiable function $\phi:\X \to \R$, the following conditions (i)-(iii) are equivalent:
(i) $\phi$ is strictly convex,
(ii) for any $\bm x, \bm y \in \X$ such that $\bm x \neq \bm y$, $\phi(\bm x) > \phi(\bm y) + \expval{\partial\phi(\bm y), \bm x - \bm y }$ holds,
(iii) for any $\bm x \in \X$, $\partial^2\phi(\bm x) \succ O$.

\begin{KSblock}
\section{Finite-Time Positivity of mass-action CRNs}
\label{sec:finite-time positivity of mass-action CRNs}
We explicitly prove that, for any positive initial condition $\bm x_0 \in \X = \R_{> 0}^N$,
the solution $\bm x_t$ of the CRE with mass-action kinetics remains in $\X$ on every finite time
interval $[0,T]$.

For each species $i$, the CRE \cref{eq:CRE} can be written as
\begin{align}
\dot{x}_{t,i} = b_i(\bm{x}_t)-a_i(\bm{x}_t)x_{t,i}, 
\end{align} 
where 
\begin{align} 
b_i(\bm{x}) &:= \sum_{r:S_{ir}<0} (-S_{ir})j_r^+(\bm{x}) + \sum_{r:S_{ir}>0} S_{ir}j_r^-(\bm{x}), \\
a_i(\bm{x}) &:= \sum_{r:S_{ir}>0} S_{ir}k_r^+\bm{x}^{\bm{s}_r^+-\bm e_i} + \sum_{r:S_{ir}<0} (-S_{ir})k_r^-\bm{x}^{\bm{s}_r^--\bm e_i}. 
\end{align}
and $\bm e_i\in\mathbb{R}^N$ denotes the $i$th standard basis vector, so that 
\begin{align} \bm{x}^{\bm{s}_r^\pm-\bm e_i} = x_i^{s_{ir}^\pm-1} \prod_{\ell\neq i}x_\ell^{s_{\ell r}^\pm}. \end{align} For every term appearing in $a_i$, the corresponding stoichiometric coefficient satisfies $s_{ir}^\pm\geq1$, and hence all exponents in $\bm{s}_r^\pm-\bm{e}_i$ are nonnegative integers.
Since $\bm{x}_t\in\mathbb{R}_{\geq0}^N$ \cite[Lemma 3.7.1]{Feinberg_2019} and $k_r^\pm>0$, every term defining $a_i(\bm{x}_t)$ and $b_i(\bm{x}_t)$ is nonnegative: $a_i(\bm{x}_t),b_i(\bm{x}_t)\geq0$.

Since $\bm{x}_t$ is continuous on $[0,T]$,
\begin{align}
M_{i,T}:=\max_{0\leq t\leq T}a_i(\bm{x}_t)<\infty.
\end{align}
Therefore,
\begin{align}
\dot{x}_{t,i}\geq -M_{i,T}x_{t,i}.
\end{align}
By Gronwall's inequality,
\begin{align}
x_{t,i}\geq x_{0,i}e^{-M_{i,T}t}>0
\qquad
\text{for all }t\in[0,T].
\end{align}
Thus, $\bm{x}_0\in\R_{>0}^N$ implies
$\bm{x}_t\in\R_{>0}^N$ for all $t\in[0,T]$.

Finite-time positivity does not imply persistence, since some
concentrations may still approach zero as $t\to\infty$.
    
\end{KSblock}

\section{Deferred Proofs}
Here, we provide complete proofs for those results stated in the main article whose proofs were deferred.
\subsection{Proof for \Cref{lemma:sufficient condition to convexity of bounds}}
\label{sec:proof for lemma:sufficient condition to convexity of bounds}
\begin{proof}[Proof of \Cref{lemma:sufficient condition to convexity of bounds}]
We prove $\dv[2]{\up{\B}_{\psi^\ast}(\sqrt{a})}{a} \geq 0$ and $\dv[2]{\low{\B}_{\psi^\ast}(\sqrt{a})}{a} \geq 0$ for any $a > 0$.
The second derivatives of $\up{\B}_{\psi^\ast}(\sqrt{a})$ for $a > 0$ is directly calculated as
\begin{align}
    \label{eq:second derivative}
    \dv[2]{\up{\B}_{\psi^\ast}(\sqrt{a})}{a} &= \frac{\sqrt{a} \cdot (\psi^\ast)''(\sqrt{a}) - (\psi^\ast)'(\sqrt{a})}{4a\sqrt{a}} + \frac{(\psi^\ast)'''(\sqrt{a})}{4\sqrt{a}}.
\end{align}
    The first term on the right-hand side of \cref{eq:second derivative} is nonnegative by the convexity of $(\psi^\ast)'$ and $(\psi^\ast)'(0)=0$:
    \begin{align}
        0=(\psi^\ast)'(0)\ge (\psi^\ast)'(\sqrt{a})+(\psi^\ast)''(\sqrt{a})(0-\sqrt{a}),        
    \end{align}
    which implies $\sqrt{a}(\psi^\ast)''(\sqrt{a})-(\psi^\ast)'(\sqrt{a})\ge 0$.
    Moreover, since $(\psi^\ast)'$ is convex, we have $(\psi^\ast)'''(\sqrt{a})\ge 0$, and hence the second term on the right-hand side of \cref{eq:second derivative} is also nonnegative.
    Therefore, $\dv[2]{\up{\B}_{\psi^\ast}(\sqrt{a})}{a} \geq 0$ holds for any $a > 0$.

Also, for any $a > 0$, we obtain $\dv[2]{\low{\B}_{\psi^\ast}(\sqrt{a})}{a} = \frac{1}{M \sqrt{M}} \cdot \dv[2]{\up{\B}_{\psi^\ast}\qty(\sqrt{a/M})}{(a/M)} \geq 0$.
\end{proof}

\subsection{Proof of \Cref{thm:lower bound of Bregman divergence in generalized equilibirum flow}}
\label{sec:proof for thm:lower bound of Bregman divergence in generalized equilibirum flow}
The proof is almost the same as that of \Cref{thm:upper bound of Bregman divergence in generalized equilibirum flow}. 
The difference is that we use an assumption and a lemma for inequalities in the opposite direction to those in \Cref{thm:upper bound of Bregman divergence in generalized equilibirum flow}.
\begin{proof}[Proof of \Cref{thm:lower bound of Bregman divergence in generalized equilibirum flow}]
    For $t \in [0,T]$, the time derivative of $D_\phi(\bm x_t \| \bm x_\mathrm{eq})$  is bounded from below as
    \begin{align}
        \label{eq:global lower 1}
        \dv{}{t} D_\phi(\bm x_t \| \bm x_\mathrm{eq}) 
        &= - \expval{ \partial_{\bm{f}} \Psi_{\bm{x}_t}(\bm{f}(\bm x_t)), \bm{f}(\bm x_t) }\\
        \label{eq:global lower 2}
        &\geq - \up{\omega}(\bm x_t) \up{\B}(\norm{\bm f(\bm x_t)}_2)\\
        \label{eq:global lower 3}
        &\geq - \up{\omega}(\bm x_t) \up{\B}\qty(\sqrt{2\sigma_{1}^2 L_{\bm x_\mathrm{eq}}  D_\phi(\bm{x}_t \|\bm x_\mathrm{eq})}).
    \end{align}
    The transformation from \cref{eq:global lower 1} to \cref{eq:global lower 2} follows from \Cref{asm: activity-force decomposability of EPR upper bound}, and the transformation from \cref{eq:global lower 2} to \cref{eq:global lower 3} follows from \Cref{lemma:upper bounds of thermodynamic force norm}, respectively.
    From \Cref{asm: activity-force decomposability of EPR upper bound}, $\up{\B}$ is positive in $(0,\infty)$ and $2\sigma_{1}^2 L_{\bm x_\mathrm{eq}} D_\phi(\bm{x}_t \|\bm x^\ast)$ is positive.
    Thus, 
    \begin{align}
        \frac{2\sigma_{1}^2 L_{\bm x_\mathrm{eq}} }{ \qty( \up{\B}(\sqrt{\cdot}) )\qty(2\sigma_{1}^2 L_{\bm x_\mathrm{eq}}  D_\phi (\bm{x}_t \|\bm x_{\mathrm{eq}})) } \dv{}{t} D_\phi(\bm x_t \| \bm x_{\mathrm{eq}}) \geq - 2\sigma_{1}^2 L_{\bm x_\mathrm{eq}} \up{\omega}(\bm x_t).
    \end{align}
    Integrating with respect to the time $t$ from $0$ to $T$ due to the monotonic non-increasing property of $D_\phi(\bm x_t \| \bm x_\mathrm{eq})$, we obtain:
    \begin{align}
        \begin{aligned}
            % \ln_{\low{\B} (\sqrt{\cdot})} (2\sigma_{\rk(S)}^2 m_{\bm x_\mathrm{eq}}  D_\phi(\bm x_T \| \bm x_{\mathrm{eq}}) ) - \ln_{\low{\B} (\sqrt{\cdot})} (2\sigma_{\rk(S)}^2 m_{\bm x_\mathrm{eq}}  D_\phi(\bm x_0 \| \bm x_{\mathrm{eq}}) ) 
            \qty[\ln_{\up{\B} (\sqrt{\cdot})} (2\sigma_{1}^2 L_{\bm x_\mathrm{eq}}  D_\phi(\bm x_t \| \bm x_{\mathrm{eq}}) )]_{t = 0}^T
            &\geq -2\sigma_{1}^2 L_{\bm x_\mathrm{eq}}  \int_0^T \up{\omega}(\bm x_t)  \dd t.
            % \\&= - 2\sigma_{\rk(S)}^2 m_{\bm x_\mathrm{eq}} \up{\Omega}(\bm{x}_{0:T}).
        \end{aligned}
    \end{align}
    % where $\ln_{\psi}(u) := \int_1^u \frac{1}{\psi(u')} \dd{u'}$ is the deformed logarithm for a positive function $\psi$ and $\Omega(\bm{x}_{0:T}):= \int_0^T \up{\omega}(\bm x_t)$ is the time-integrated lower bound of activities
    % $\ln_{\low{\B} (\sqrt{\cdot})}$ is invertible
    Since $\exp_{\up{\B} (\sqrt{\cdot})}$ is  monotonically increasing, we obtain the target inequality \cref{eq:lower bound of Bregman divergence in generalized equilibirum flow}.
    % \begin{align}
    %     D_\phi(\bm x_T \| \bm x_{\mathrm{eq}}) \leq \frac{\exp_{\up{\B} (\sqrt{\cdot})} \qty[ \ln_{\up{\B} (\sqrt{\cdot})} (2\sigma_{\rk(S)}^2 m_{\bm x_\mathrm{eq}}  D_\phi(\bm x_0 \| \bm x_{\mathrm{eq}}) ) - 2\sigma_{\rk(S)}^2 m_{\bm x_\mathrm{eq}} \up{\Omega}(\bm{x}_{0:T}) ]}{2\sigma_{\rk(S)}^2 m_{\bm x_\mathrm{eq}} }.
    % \end{align}
\end{proof}

\subsection{Proof of \Cref{thm:upper bound under locally PL assumption}}
\label{sec:proof for thm:upper bound under locally PL assumption}
Fully exploiting the local convexity condition, we establish an another upper bound on the Bregman divergence; the proof is almost similar to that of \Cref{thm:upper bound of Bregman divergence in generalized equilibirum flow}, except in \cref{eq:separate local convexity and divergence}, where \Cref{asm: convexity of EPR lower bound} is used to separate the $\bm x$-dependent convexity parameter from the $D_\phi$-dependent term.
% \KS{This is a somehow similar calculation to the proof of theorem 3.5., so you might consider putting in into the SM as well.}

\begin{proof}[Proof of \Cref{thm:upper bound under locally PL assumption}]
    The time derivative of the divergence is bounded from above as
    \begin{align}
        \label{eq:local upper 1}
        \dv{}{t} D_\phi(\bm x_t \| \bm x_\mathrm{eq}) 
        &= - \expval{ \partial_{\bm{f}} \Psi_{\bm{x}_t}(\bm{f}(\bm x_t)), \bm{f}(\bm x_t) }\\
        \label{eq:local upper 2}
        &\leq - \low{\omega}(\bm x_t) \low{\B}(\norm{\bm f (\bm x_t)}_2)\\
        \label{eq:local upper 3}
        &\leq - \low{\omega}(\bm x_t) \low{\B}\qty(\sqrt{2\sigma_1^2 m_{\bm x_\mathrm{eq}}(\bm x_t^\mathrm{eff})  D_\phi(\bm{x}_t \|\bm x_\mathrm{eq})}).
    \end{align}
    It can be transformed in the same way as \cref{eq:global upper 1,eq:global upper 2,eq:global upper 3}, except that for \cref{eq:global upper 3} we use the local  convexity condition rather than the global one.

    Because of the convexity of $\low{\B}(\sqrt{\cdot})$ (\Cref{asm: convexity of EPR lower bound}) and $\low{\B}(\sqrt{0}) = 0$, for $a \geq 1$ and $x \in \R_{\geq 0}$, we obtain $(\low{\B}(\sqrt{\cdot}))(ax) \geq a(\low{\B}(\sqrt{\cdot}))(x)$ as 
    \begin{align}
    \begin{aligned}
        \low{\B}(\sqrt{ax}) &= a \cdot \qty{ (1 - a^{-1}) \cdot \low{\B}(\sqrt{0}) + a^{-1} \cdot \low{\B}(\sqrt{ax})} \\
        &\geq a \cdot \low{\B}(\sqrt{(1 - a^{-1}) 0  + a^{-1}ax}) = a \low{\B}(x).
    \end{aligned}
    \end{align}
    From $m_{\bm x_\mathrm{eq}}(\bm x_t^\text{eff}) \geq \low{m_{\bm x_\mathrm{eq}}}$, the following inequality holds:
    \begin{align}
        \label{eq:separate local convexity and divergence}
        \low{\B}\qty(\sqrt{2\sigma_1^2 m_{\bm x_\mathrm{eq}}(\bm x_t^\mathrm{eff})  D_\phi(\bm{x}_t \|\bm x_\mathrm{eq})}) \geq \frac{m_{\bm x_\mathrm{eq}}(\bm x_t^\mathrm{eff})}{\low{m_{\bm x_\mathrm{eq}}}}\low{\B}\qty(\sqrt{2\sigma_1^2 \low{m_{\bm x_\mathrm{eq}}}  D_\phi(\bm{x}_t \|\bm x_\mathrm{eq})}).
    \end{align}
    Therefore, the following inequality holds:
    \begin{align}
        \dv{}{t} D_\phi(\bm x_t \| \bm x_\mathrm{eq}) \leq - \frac{m_{\bm x_\mathrm{eq}}(\bm x_t^\mathrm{eff})}{\low{m_{\bm x_\mathrm{eq}}}} \cdot \low{\omega}(\bm x_t) \cdot \low{\B}\qty(\sqrt{2\sigma_1^2 \low{m_{\bm x_\mathrm{eq}}}  D_\phi(\bm{x}_t \|\bm x_\mathrm{eq})})
    \end{align}
    From the same calculation as in \Cref{thm:upper bound of Bregman divergence in generalized equilibirum flow} by replacing $m_{\bm x_\mathrm{eq}}$ and $\bm \omega (\bm x)$ by $\low{m_{\bm x_\mathrm{eq}}}$ and $\frac{m_{\bm x_\mathrm{eq}}(\bm x_t^\mathrm{eff})}{\low{m_{\bm x_\mathrm{eq}}}} \cdot \low{\omega}(\bm x_t)$, respectively, we can obtain the upper bound,
    \begin{align}
        D_\phi(\bm x_T \| \bm x_{\mathrm{eq}}) \leq \frac{\exp_{\low{\B} (\sqrt{\cdot})} \qty[ \ln_{\low{\B} (\sqrt{\cdot})} (2\sigma_{\rk{S}}^2 \low{m_{\bm x_\mathrm{eq}}}  D_\phi(\bm x_0 \| \bm x_{\mathrm{eq}}) ) - 2\sigma_{\rk{S}}^2 \low{m_{\bm x_\mathrm{eq}}} \low{\Omega}^{m_{\bm x_\mathrm{eq}}}(\bm{x}_{0:T}) ]}{2 \sigma_{\rk{S}}^2 \low{m_{\bm x_\mathrm{eq}}}},
    \end{align}
    where $\low{\Omega}^{m_{7\bm x_\mathrm{eq}}}(\bm{x}_{0:T}):= \int_0^T \frac{m_{\bm x_\mathrm{eq}}(\bm x_t^\mathrm{eff})}{\low{m_{\bm x_\mathrm{eq}}}} \cdot \low{\omega}(\bm x_t) \dd t$.
\end{proof}

\subsection{Proof of \Cref{thm:lower bound under locally LC assumption}}
\label{sec:proof for thm:lower bound under locally LC assumption}
The proof is analogous to that of \Cref{thm:upper bound under locally PL assumption}.
\begin{proof}[Proof of \Cref{thm:lower bound under locally LC assumption}]
    The time derivative of the divergence is bounded from below as
    \begin{align}
        \label{eq:local lower 1}
        \dv{}{t} D_\phi(\bm x_t \| \bm x_\mathrm{eq}) 
        &= - \expval{ \partial_{\bm{f}} \Psi_{\bm{x}_t}(\bm{f}(\bm x_t)), \bm{f}(\bm x_t) }\\
        \label{eq:local lower 2}
        &\geq - \up{\omega}(\bm x_t) \up{\B}(\norm{\bm f (\bm x_t)}_2)\\
        \label{eq:local lower 3}
        &\geq - \up{\omega}(\bm x_t) \up{\B}\qty(\sqrt{2\sigma_1^2 L_{\bm x_\mathrm{eq}}(\bm x_t)  D_\phi(\bm{x}_t \|\bm x_\mathrm{eq})}).
    \end{align}
    It can be transformed in the same way as \cref{eq:global lower 1,eq:global lower 2,eq:global lower 3}, except that for \cref{eq:global lower 3} we use the local  convexity condition rather than the global one.

    Because of the convexity of $\up{\B}(\sqrt{\cdot})$ (\Cref{asm: convexity of EPR upper bound}) and $\up{\B}(\sqrt{0}) = 0$, for $a \leq 1$ and $x \in \R_{\geq 0}$, we obtain $\up{\B}(\sqrt{ax}) \leq a\up{\B}(\sqrt{x})$ as $\up{\B}(\sqrt{ax}) = \up{\B}(\sqrt{(1 -a)0 + a x})) \leq  (1 - a) \cdot \up{\B}(\sqrt{0}) + a  \up{\B}(\sqrt{x})  = a \up{\B}(x)$.
    From $L_{\bm x_\mathrm{eq}}(\bm x_t^\text{eff}) \leq \up{L_{\bm x_\mathrm{eq}}}$, the following inequality holds:
    \begin{align}
        % \label{eq:separate local convexity and divergence2}
        \up{\B}\qty(\sqrt{2\sigma_1^2 L_{\bm x_\mathrm{eq}}(\bm x_t)  D_\phi(\bm{x}_t \|\bm x_\mathrm{eq})}) \leq \frac{L_{\bm x_\mathrm{eq}}(\bm x_t)}{\up{L_{\bm x_\mathrm{eq}}}}\up{\B}\qty(\sqrt{2\sigma_1^2 \up{L_{\bm x_\mathrm{eq}}}  D_\phi(\bm{x}_t \|\bm x_\mathrm{eq})}).
    \end{align}
    Therefore, the following inequality holds:
    \begin{align}
        \dv{}{t} D_\phi(\bm x_t \| \bm x_\mathrm{eq}) \geq - \frac{L_{\bm x_\mathrm{eq}}(\bm x_t)}{\up{L_{\bm x_\mathrm{eq}}}} \cdot \up{\omega}(\bm x_t) \cdot \up{\B}\qty(\sqrt{2\sigma_1^2 \up{L_{\bm x_\mathrm{eq}}}  D_\phi(\bm{x}_t \|\bm x_\mathrm{eq})})
    \end{align}
    From the same calculation as in \Cref{thm:lower bound of Bregman divergence in generalized equilibirum flow} by replacing $L_{\bm x_\mathrm{eq}}$ and $\bm \omega (\bm x)$ by $\up{L_{\bm x_\mathrm{eq}}}$ and $\frac{L_{\bm x_\mathrm{eq}}(\bm x_t)}{\up{L_{\bm x_\mathrm{eq}}}} \cdot \up{\omega}(\bm x_t)$, respectively, we can obtain the lower bound,
    \begin{align}
        D_\phi(\bm x_T \| \bm x_{\mathrm{eq}}) \geq \frac{\exp_{\up{\B} (\sqrt{\cdot})} \qty[ \ln_{\up{\B} (\sqrt{\cdot})} (2\sigma_{\rk{S}}^2 \up{L_{\bm x_\mathrm{eq}}}  D_\phi(\bm x_0 \| \bm x_{\mathrm{eq}}) ) - 2\sigma_{\rk{S}}^2 \up{L_{\bm x_\mathrm{eq}}} \up{\Omega}^{L_{\bm x_\mathrm{eq}}}(\bm{x}_{0:T}) ]}{2 \sigma_{\rk{S}}^2 \up{L_{\bm x_\mathrm{eq}}}},
    \end{align}
    where $\up{\Omega}^{L_{\bm x_\mathrm{eq}}}(\bm{x}_{0:T}):= \int_0^T \frac{L_{\bm x_\mathrm{eq}}(\bm x_t)}{\up{L_{\bm x_\mathrm{eq}}}} \cdot \up{\omega}(\bm x_t) \dd t$.
\end{proof}

\subsection{Proof of \Cref{lemma:upper bounds of thermodynamic force norm}}
\label{sec:proof for lemma:upper bounds of thermodynamic force norm}
\begin{proof}[Proof of \Cref{lemma:upper bounds of thermodynamic force norm}]
Let $\qty{\bm{\hat v}_i \mid i = 1, \cdots, N}$ be a basis determined by SVD of the stoichimetric matrix $S$ (see \Cref{subsubsec:strucural properties of CRNs}). The gradient $\partial_{\bm x} D_\phi(\bm x \| \bm x_\mathrm{eq})$ is uniquely decomposed as 
\begin{align}
    \partial_{\bm x} D_\phi(\bm x \| \bm x_\mathrm{eq}) = \sum_{i=1}^N \expval{\partial_{\bm x} D_\phi(\bm x \| \bm x_\mathrm{eq}), \bm{\hat v}_i} \bm{\hat v}_i.
\end{align}
The norm of $\bm f(\bm x) = S^\top \partial_{\bm x} D_\phi(\bm x \| \bm x_\mathrm{eq})$ is bounded from above as
\begin{align}
    \norm{\bm f(\bm x)}^2_2
    &= \norm{\sum_{i=1}^N \expval{\partial_{\bm x} D_\phi(\bm x \| \bm x_\mathrm{eq}), \bm{\hat v}_i} S^\top \bm{\hat v}_i}_2^2\\
    &= \norm{\sum_{i=1}^{\rk{S}} \expval{\partial_{\bm x} D_\phi(\bm x \| \bm x_\mathrm{eq}), \bm{\hat v}_i} \sigma_i \bm{\hat u}_i}_2^2\\
    &\leq \sigma_1^2 \norm{\sum_{i=1}^{\rk{S}} \expval{\partial_{\bm x} D_\phi(\bm x \| \bm x_\mathrm{eq}), \bm{\hat v}_i} \bm{\hat u}_i}_2^2\\
    &= \sigma_1^2 \sum_{i=1}^{\rk{S}} \expval{\partial_{\bm x} D_\phi(\bm x \| \bm x_\mathrm{eq}), \bm{\hat v}_i}^2\\
    &\leq \sigma_1^2 \norm{ \partial_{\bm x} D_\phi(\bm x \| \bm x_\mathrm{eq})}_2^2.
\end{align}
Thus, $\norm{\bm f(\bm x)}_2 \leq \sigma_1 \norm{\partial_{\bm x} D_\phi(\bm x \| \bm x_\mathrm{eq})}_2$ holds.
The equality holds iff $\sigma_1 = \cdots = \sigma_N > 0$.
Additionally, if $D_\phi(\cdot \| \bm x_\mathrm{eq})$ is $L_{\bm x_\mathrm{eq}}(\cdot )$-locally PL at $\bm x$, then $\sigma_1 \norm{\partial_{\bm x} D_\phi(\bm x \| \bm x_\mathrm{eq})}_2 \leq \sqrt{2 \sigma_1^2 L_{\bm x_\mathrm{eq}}(\bm x ) D_\phi(\bm x \| \bm x_\mathrm{eq})}$ holds.
The equality holds iff $L_{\bm x_\mathrm{eq}}(\bm x) = \rho^\phi_{\bm x_\mathrm{eq}}(\bm x)$.    
\end{proof}

\section{Simulation Results for the \KS{\sout{other} another} plateau-exhibiting CRN}
\label{sec:simulation results for the other plateau-exhibiting CRN}
The properties described in \Cref{sec:convergence analysis and characterization of slow relaxation in mass-action reaction systems} are also observed in a mass-action CRN, denoted CRN2, whose reactions are as follows:
\begin{align}
\begin{aligned}
\label{eq:CRN2}
    X_0 + X_2 \rightleftharpoons X_1 + X_2,\\
    X_0 + X_4 \rightleftharpoons X_2 + X_4,\\
    X_0 + X_1 \rightleftharpoons X_3 + X_1,\\
    X_0 + X_1 \rightleftharpoons X_4 + X_1.
\end{aligned}
\end{align}
This CRN was also studied in \cite{Awazu2009}.
Note that CRN2 \cref{eq:CRN2} has the same stoichiometric matrix as CRN1 \cref{eq:CRN1}, but with different catalytic species.
Time evolutions of the KL divergences, their upper and lower bounds, the force norm and its bounds, and the EPR and its bounds are shown in \Cref{fig:numerical simulation for slow relation 2}.
As in CRN1, the lower bounds for the force norm and the EPR are much tighter than the corresponding upper bounds.
We further confirm that, in CRN2, only the local bounds (not the global bounds) exhibit plateaus in their time courses.
\KS{
Finally, as in CRN1, the successive EPR estimates indicate that the dominant loss of tightness arises at the step involving the lower bound on the reaction activities in the CRN2 (\Cref{fig:EPR quad stages 2}).
}

\begin{figure}[htbp]
    \centering
    \includegraphics[width=.8\linewidth]{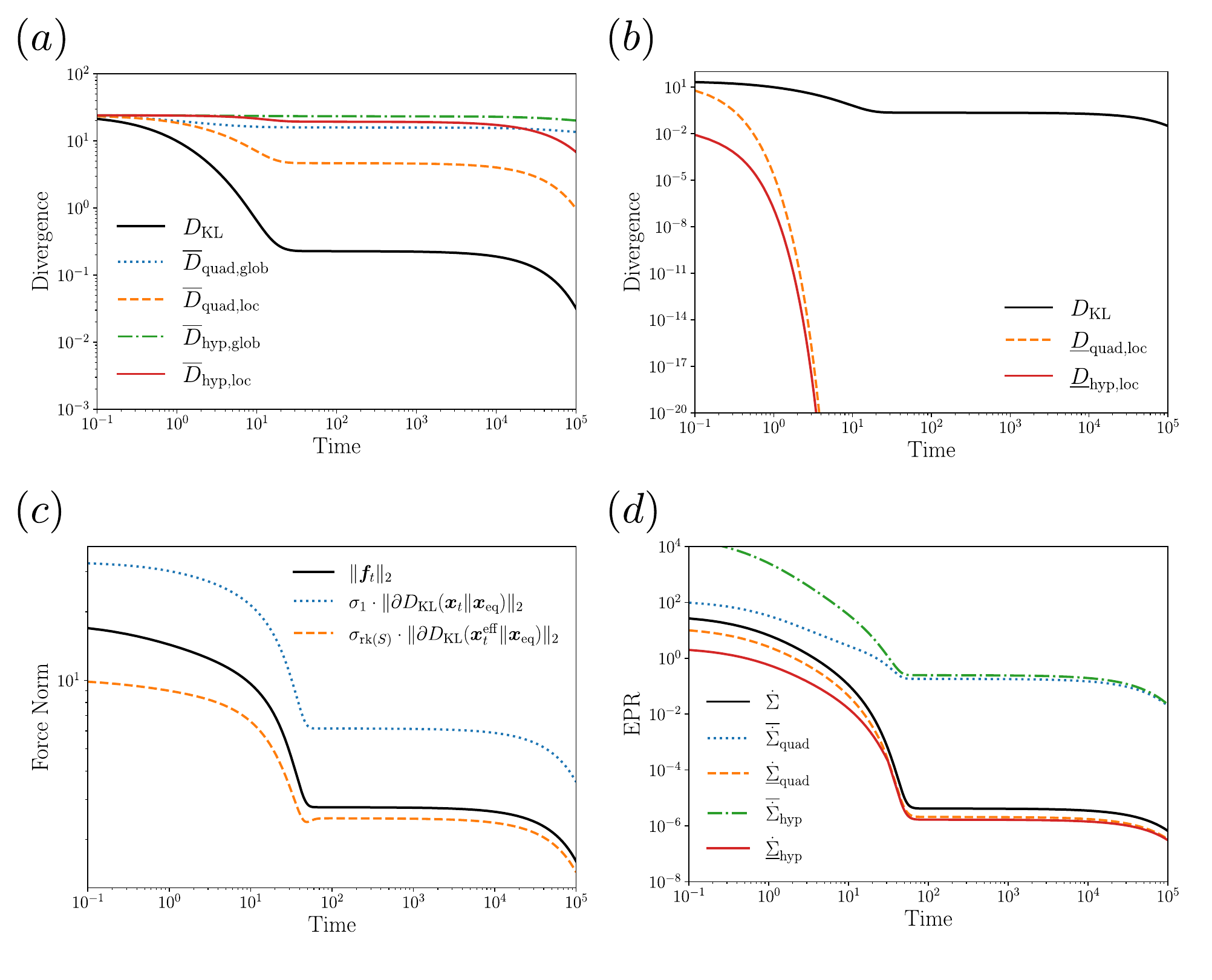}
    \caption{The time courses of the generalized KL divergence and the corresponding bounds in the CRN2 \cref{eq:CRN2} $(0\leq t\leq T = 10^5)$. (a) The time courses of the divergence $D_\mathrm{KL}(\bm x_t \| \bm x_\mathrm{eq})$ and its bounds. We show the local upper bounds $\up{D}_\mathrm{hyp,loc}$ \cref{eq:hyp_loc_upper}, $\up{D}_\mathrm{quad,loc}$ \cref{eq:quad_loc_upper} for $\rho_{\bm x_\mathrm{eq}}^\mathrm{KL}(\bm x)$, and the global upper bounds $\up{D}_\mathrm{hyp,glob}, \up{D}_\mathrm{quad,glob}$ for $\low{\rho}_{\bm x_\mathrm{eq}}^\mathrm{KL} = \min_{\bm x\in \mathcal{O}_{10^5}(\bm x_0)} \rho_{\bm x_\mathrm{eq}}^\mathrm{KL}(\bm x)$. 
    % We exploit locally-PL and locally-LC conditions with respect to $m_{\bm x_\mathrm{eq}}(\bm x) = L_{\bm x_\mathrm{eq}}(\bm x) = \rho_{\bm x_\mathrm{eq}}^\mathrm{KL}(\bm x)$, and the global convexity parameters $\low{m_{\bm x_\mathrm{eq}}} = \min_{\bm x\in \mathcal{O}_{10^5}(\bm x_0)} \rho_{\bm x_\mathrm{eq}}^\mathrm{KL}(\bm x)$ and $\up{L_{\bm x_\mathrm{eq}}}  = \max_{\bm x\in \mathcal{O}_{10^5}(\bm x_0)} \rho_{\bm x_\mathrm{eq}}^\mathrm{KL}(\bm x)$, respectively.
    (b) The time courses of the KL divergence, the local lower bounds $\low{D}_\mathrm{hyp,loc} \cref{eq:hyp_loc_lower}, \low{D}_\mathrm{quad,loc}$ \cref{eq:hyp_loc_lower} for $\rho_{\bm x_\mathrm{eq}}^\mathrm{KL}(\bm x)$. 
    (c) The time courses of the force norm $\norm{\bm f (\bm x_t)}_2$, its upper bound \cref{eq:upper bound of thermodynamic force norm by gradient} (blue dotted curve) and lower bound \cref{eq:lower bound of thermodynamic force norm by gradient} (orange dashed curve). 
    (d) The time courses of the EPR $\dot\Sigma(\bm x_t, \bm f(\bm x_t))$, the quadratic upper bound $\up{\dot \Sigma}_\mathrm{quad}$ and lower bound  $\low{\dot \Sigma}_\mathrm{quad}$ \cref{eq:quadratic EPR bounds}, and the hyperbolic upper bound $\up{\dot \Sigma}_\mathrm{hyp}$ and lower bound $\low{\dot \Sigma}_\mathrm{hyp}$ \cref{eq:hyperbolic EPR bounds}.}
    \label{fig:numerical simulation for slow relation 2}
\end{figure}

\begin{figure}
    \centering
    \includegraphics[width=.55\linewidth]{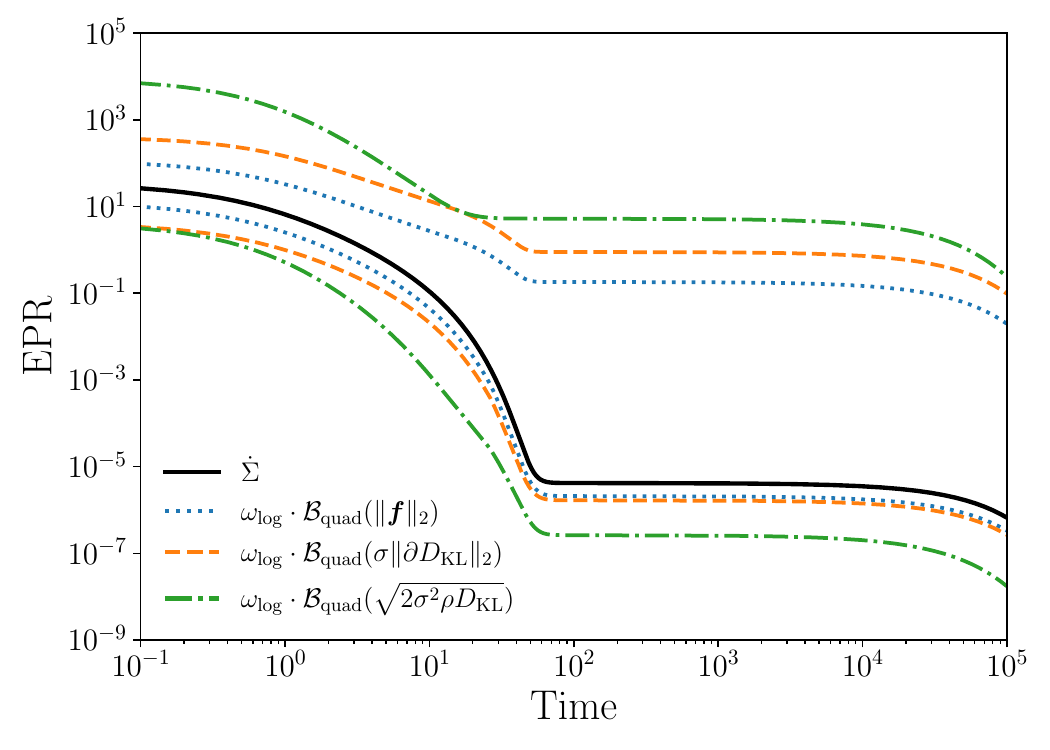}
    \caption{\KS{The time course of the EPR, $\dot{\Sigma}$, and its successive bounds in the CRN2 \cref{eq:CRN2}. The blue dotted curves show the quadratic dissipation bounds in \Cref{eq:quadratic EPR bounds} evaluated at $\bm f=\bm f(\bm x)$; the orange dashed curves are obtained by bounding the force norm in terms of the KL-divergence gradient norm; and the green dash-dotted curves follow by further bounding the gradient norm using the global convexity.}}
    \label{fig:EPR quad stages 2}
\end{figure}

\begin{KSblock}
\section{Supplementary Material for the results on Hatakeyama-Kaneko model}
\subsection{Details of model and numerical settings}
\label{sec:numerical settings in MWC model}

In \Cref{sec:bounds in MWC model}, we consider the deterministic model introduced by Hatakeyama and Kaneko \cite{HatakeyamaKaneko2020}, which extends the classical Monod--Wyman--Changeux model by incorporating enzyme-catalyzed modification reactions. The state variable is
\begin{align}
\bm{x}
= (x_{T_0},\dots,x_{T_m},x_{R_0},\dots,x_{R_m})^\top
\in \mathbb{R}_{>0}^{2(m+1)},
\end{align}
where $x_{T_j}$ and $x_{R_j}$ denote the concentrations of proteins in the tense and relaxed conformations, respectively, with $j$ modified sites. The reduced CRN consists of the modification reactions
\begin{align}
T_i \rightleftharpoons T_{i+1},
\qquad
R_i \rightleftharpoons R_{i+1}
\qquad (i=0,\dots,m-1),
\end{align}
and the conformational transitions
\begin{align}
T_i \rightleftharpoons R_i
\qquad (i=0,\dots,m).
\end{align}
The corresponding forward and backward reaction rates are
\begin{align}
j_{T,i}^+(\bm{x})
&=
(m-i)k_{T,+}(\bm{x})x_{T_i},
&
j_{T,i}^-(\bm{x})
&=
(i+1)k_{T,-}(\bm{x})x_{T_{i+1}},
\\
j_{R,i}^+(\bm{x})
&=
(m-i)k_{R,+}(\bm{x})x_{R_i},
&
j_{R,i}^-(\bm{x})
&=
(i+1)k_{R,-}(\bm{x})x_{R_{i+1}},
\\
j_{TR,i}^+(\bm{x})
&=
k_{TR,i}x_{T_i},
&
j_{TR,i}^-(\bm{x})
&=
k_{RT,i}x_{R_i}.
\end{align}
Here $i=0,\dots,m-1$ in the first two lines and $i=0,\dots,m$ in the last line. The dynamics is given by \Cref{eq:CRE}.
The modification reactions are catalyzed by a common enzyme of finite total concentration. Consequently, the effective rate constants depend on the probabilities that the enzyme is bound to a protein in the tense or relaxed conformation:
\begin{align}
k_{T,+}(\bm{x})
&=
\frac{P_{b,T}(\bm{x})}{\tau_m}
e^{-\beta(\varepsilon_{R,0}-\varepsilon_{T,0})},
&
k_{T,-}(\bm{x})
&=
\frac{P_{b,T}(\bm{x})}{\tau_m}
e^{-\beta(\varepsilon_{R,0}-\varepsilon_{T,1})},
\\
k_{R,+}(\bm{x})
&=
\frac{P_{b,R}(\bm{x})}{\tau_m},
&
k_{R,-}(\bm{x})
&=
\frac{P_{b,R}(\bm{x})}{\tau_m}
e^{-\beta(\varepsilon_{R,0}-\varepsilon_{R,1})}.
\end{align}
Here $\tau_m$ denotes the characteristic timescale of the modification reactions.
Following the quasi-equilibrium treatment of enzyme binding in \cite{HatakeyamaKaneko2020}, the binding probabilities are expressed in terms of the free-enzyme concentration as
\begin{align}
P_{b,T}(\bm{x})
&=
\frac{E_{\mathrm{free}}(\bm{x})}
{e^{-\beta\varepsilon_{b,T}}+E_{\mathrm{free}}(\bm{x})},
&
P_{b,R}(\bm{x})
&=
\frac{E_{\mathrm{free}}(\bm{x})}
{e^{-\beta\varepsilon_{b,R}}+E_{\mathrm{free}}(\bm{x})},
\end{align}
where $\varepsilon_{b,T}$ and $\varepsilon_{b,R}$ characterize the enzyme-binding affinities of the tense and relaxed conformations, respectively. The free-enzyme concentration is not treated as an independent dynamical variable. Instead, for each $\bm{x}$, it is determined by conservation of the total enzyme concentration:
\begin{align}
E_{\mathrm{tot}}
=
E_{\mathrm{free}}(\bm{x})
+
\sum_{i=0}^m x_{T_j}
\frac{E_{\mathrm{free}}(\bm{x})}
{e^{-\beta\varepsilon_{b,T}}+E_{\mathrm{free}}(\bm{x})}
+
\sum_{i=0}^m x_{R_j}
\frac{E_{\mathrm{free}}(\bm{x})}
{e^{-\beta\varepsilon_{b,R}}+E_{\mathrm{free}}(\bm{x})}.
\label{eq:free-enzyme-conservation}
\end{align}
The right-hand side of \cref{eq:free-enzyme-conservation} is strictly increasing as a function of $E_{\mathrm{free}}\geq 0$ and ranges from $0$ to infinity. Hence, the equation uniquely determines a positive state-dependent concentration $E_{\mathrm{free}}(\bm{x})$. 
Following the Supplemental Material of Hatakeyama and Kaneko \cite{HatakeyamaKaneko2020}, our numerical calculations neglect the free-enzyme term in \cref{eq:free-enzyme-conservation}, assuming that $E_{\mathrm{free}}$ is small compared with the concentrations of enzyme bound to the proteins:
\begin{align}
E_{\mathrm{tot}}
\simeq
\sum_{i=0}^m x_{T_j}
\frac{E_{\mathrm{free}}(\bm{x})}
{e^{-\beta\varepsilon_{b,T}}+E_{\mathrm{free}}(\bm{x})}
+
\sum_{i=0}^m x_{R_j}
\frac{E_{\mathrm{free}}(\bm{x})}
{e^{-\beta\varepsilon_{b,R}}+E_{\mathrm{free}}(\bm{x})}.
% \label{eq:free-enzyme-conservation}
\end{align}
Under this approximation, the enzyme-conservation relation reduces to a quadratic equation for $E_{\mathrm{free}}(\bm{x})$, whose positive root expresses the free-enzyme concentration explicitly in terms of the model parameters and the instantaneous state $\bm{x}$.
The conformational transition rates are determined by the energy difference
\begin{align}
\Delta E_i
&=
i(\varepsilon_{R,1}-\varepsilon_{T,1})
+
(m-i)(\varepsilon_{R,0}-\varepsilon_{T,0}),
\end{align}
as
\begin{align}
k_{TR,i}
&=
\frac{1}{\tau_s}\min\qty{1,e^{-\beta\Delta E_i}},
&
k_{RT,i}
&=
\frac{1}{\tau_s}\min\qty{1,e^{\beta\Delta E_i}},
\end{align}
where $\tau_s$ denotes the characteristic timescale of the conformational transitions.

We next specify the parameter sets and numerical implementation. For the main numerical experiment, we use
\begin{align}
&m=2,\quad
\beta=1.25,\quad
(\varepsilon_{T,0},\varepsilon_{T,1})=(2.7,2.3),
\nonumber\\
&(\varepsilon_{R,0},\varepsilon_{R,1})=(3.05,1.95),\quad
(\varepsilon_{b,T},\varepsilon_{b,R})=(0,4.0),
\nonumber\\
&E_{\mathrm{tot}}=0.3,\quad
(\tau_s,\tau_m)=(1,10),
\end{align}
and choose a strictly positive initial perturbation with size $\delta=10^{-6}$. 
In the original Hatakeyama-Kaneko setting, the initial state is concentrated at the fully unmodified tense state. Since our theoretical framework is formulated on the strictly positive orthant, we instead use a positive perturbation of that profile:
\begin{align}
x_{T_0}(0)&= 1 -(2m+1)\delta,\\
x_{T_i}(0)&=\delta \qquad (i=1,\dots,m),\\
x_{R_i}(0)&=\delta \qquad (i=0,\dots,m),
\end{align}
with $\delta=10^{-6}$. Note that the total concentration is conserved over time: $\norm{\bm x_t}_1 = \norm{\bm x_0}_1 = 1$.
The ODE is integrated on $0\le t\le 500$.
The ODE is solved using \texttt{scipy.integrate.solve\_ivp} with the adaptive stiff method \texttt{BDF}
The relative and absolute tolerances are set to $10^{-8}$ and $10^{-10}$, respectively. 
The solution is recorded on a hybrid observation grid with minimum spacing $dt_{\mathrm{obs}}=10^{-3}$, using fine linear sampling near the initial time and logarithmic sampling over longer times.
% Near equilibrium, the reaction force $F_r=\log(j_r^+/j_r^-)$ is evaluated with a stabilized log-ratio, and the logarithmic activity $\omega_{\log,r}=(j_r^+-j_r^-)/F_r$ is computed using an arithmetic-mean fallback when the log-ratio becomes numerically ill-conditioned.

In the following section,
as the baseline of parameters used in the original Hatakeyama--Kaneko parameter regime, we also consider
\begin{align}
&m=6,\quad
\beta=1.75,\quad
(\varepsilon_{T,0},\varepsilon_{T,1})=(3,2),
\nonumber\\
&(\varepsilon_{R,0},\varepsilon_{R,1})=(4,1),\quad
(\varepsilon_{b,T},\varepsilon_{b,R})=(0,10),
\nonumber\\
&E_{\mathrm{tot}}=0.2,\quad
(\tau_s,\tau_m)=(1,10),
\end{align}
with the same initial condition with $\delta=10^{-6}$, integrated on $0\le t\le 10^5$.

\subsection{EPR bounds for the parameter set used in the main text}
\label{sec:MWC EPR bounds main} 
As in the Gray--Scott and Awazu--Kaneko models, we evaluate the successive EPR bounds for the parameter set used in the main text.
As shown in \Cref{fig:MWC EPR bounds}, the dominant loss occurs at the activity-estimate step, indicating that the separation of scales among the reaction activities is the main source of looseness.
This behavior is similar to that observed in the Awazu--Kaneko model, whereas in the Gray--Scott model the largest loss arises in the force-to-gradient estimate.
In the present MWC model, this activity heterogeneity arises from the transient suppression of specific modification reactions due to competition for a limited amount of enzyme, thereby accounting for the looseness of the final relaxation bound.

\begin{figure}
    \centering
    \includegraphics[width=.55\linewidth]{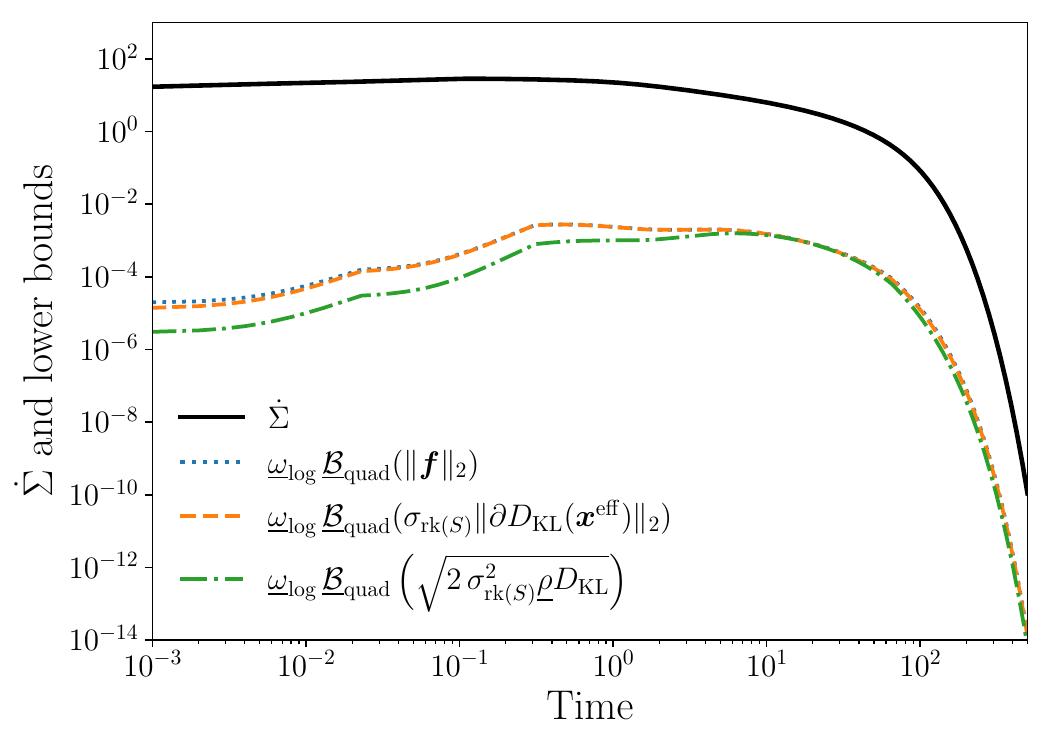}
    \caption{
    \KS{
    Time evolution of the EPR $\dot{\Sigma}$ and its successive lower bounds in the Hatakeyama--Kaneko MWC model for the parameter set used in the main text.
    The blue dotted curve shows the quadratic dissipation bound evaluated at $\bm f=\bm f(\bm x)$; the orange dashed curve is obtained by bounding the force norm in terms of the KL-divergence gradient norm at the effective concentration; and the green dash-dotted curve follows by further applying the convexity estimate.
    }}
    \label{fig:MWC EPR bounds}
\end{figure}

\subsection{Bounds for the original Hatakeyama--Kaneko parameters}
\label{sec:MWC bounds original} 
We also evaluate the relaxation bounds using the original Hatakeyama--Kaneko parameter set. The Bernoulli KL divergences associated with the tense-state fraction and the unmodified-site fraction remain below the full KL divergence, consistently with the inequalities proved in the main text. The global and local upper bounds are more conservative than those for the parameter set used in the main text, reflecting the stronger separation of reaction timescales in the original model (\Cref{fig:MWC bounds 2}). 

\begin{figure}
    \centering
    \includegraphics[width=.55\linewidth]{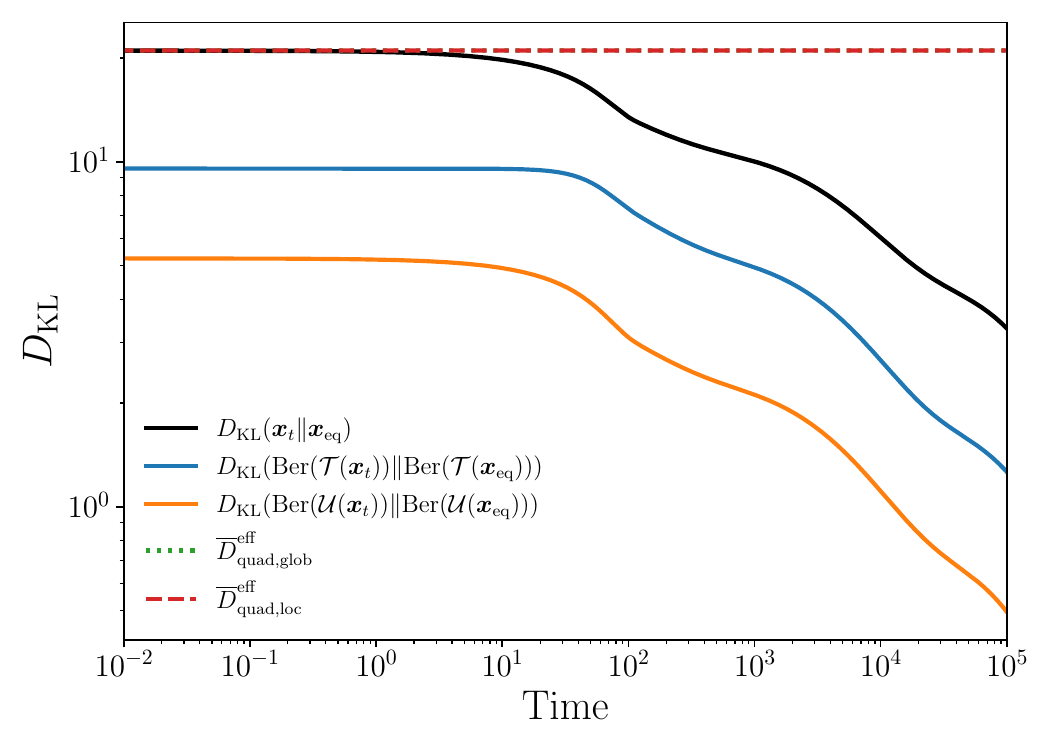}
    \caption{
    \KS{
    The time series of the KL divergence $D_{\mathrm{KL}}(\bm x_t|\bm x_{\mathrm{eq}})$, the Bernoulli KL divergences $D_{\mathrm{KL}}
\qty(
\mathrm{Ber}(\mathcal{T}(\bm x_t))
\|
\mathrm{Ber}(\mathcal{T}(\bm x_{\mathrm{eq}}))
)$ and $D_{\mathrm{KL}}
\qty(
\mathrm{Ber}(\mathcal{U}(\bm x_t))
\|
\mathrm{Ber}(\mathcal{U}(\bm x_{\mathrm{eq}}))
)$, and the global and local upper bounds $\overline{D}_{\mathrm{quad,glob}}, \overline{D}_{\mathrm{quad,loc}}$ on the KL divergence in the Hatakeyama--Kaneko MWC model with the original parameters. 
    }}
    \label{fig:MWC bounds 2}
\end{figure}

The decomposition of the corresponding EPR bounds again shows that the activity estimate is the dominant source of looseness (\Cref{fig:MWC EPR bounds 2}). The deterioration is more pronounced than for the parameter set used in the main text, consistently with the broader separation of scales among the reaction activities. 

\begin{figure}
    \centering
    \includegraphics[width=.55\linewidth]{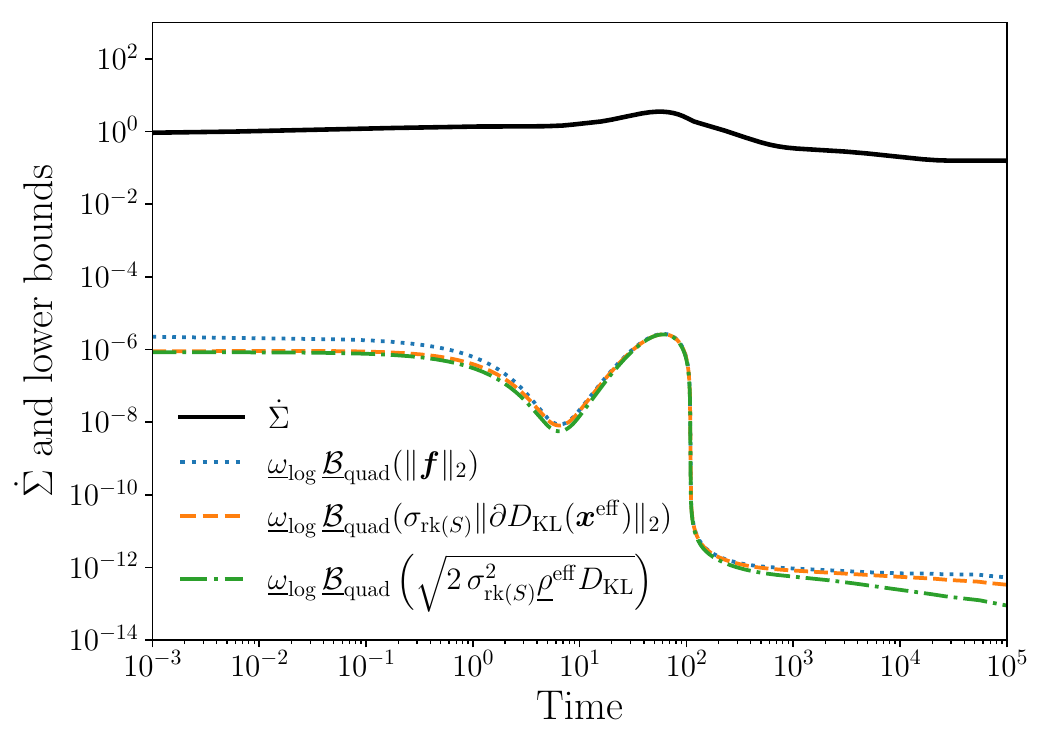}
    \caption{
    \KS{
    The time series of the KL divergence $D_{\mathrm{KL}}(\bm x_t|\bm x_{\mathrm{eq}})$, the Bernoulli KL divergences $D_{\mathrm{KL}}
\qty(
\mathrm{Ber}(\mathcal{T}(\bm x_t))
\|
\mathrm{Ber}(\mathcal{T}(\bm x_{\mathrm{eq}}))
)$ and $D_{\mathrm{KL}}
\qty(
\mathrm{Ber}(\mathcal{U}(\bm x_t))
\|
\mathrm{Ber}(\mathcal{U}(\bm x_{\mathrm{eq}}))
)$, and the global and local upper bounds $\overline{D}_{\mathrm{quad,glob}}, \overline{D}_{\mathrm{quad,loc}}$ on the KL divergence in the Hatakeyama--Kaneko MWC model. Detailed numerical conditions are provided in \Cref{sec:numerical settings in MWC model}.
    }}
    \label{fig:MWC EPR bounds 2}
\end{figure}

\subsection{KL divergence in CRN bounds Bernolli KL divergence of fractions from above}
\label{sec:proof of upper bounds in MWC model}
We prove that the KL divergence
$D_{\mathrm{KL}}(\bm{x}\|\bm{x}_{\mathrm{eq}})$
bounds the two Bernoulli KL divergences
$
D_{\mathrm{KL}}\!\left(
\operatorname{Ber}(\mathcal{T}(\bm{x}))
\,\middle\|\,
\operatorname{Ber}(\mathcal{T}(\bm{x}_{\mathrm{eq}}))
\right)
$
and
$
D_{\mathrm{KL}}\!\left(
\operatorname{Ber}(\mathcal{U}(\bm{x}))
\,\middle\|\,
\operatorname{Ber}(\mathcal{U}(\bm{x}_{\mathrm{eq}}))
\right)
$
for any
$\bm{x},\bm{x}_{\mathrm{eq}}\in\Stoich(\bm{x}_0)$.
Since the total concentration is conserved on
$\Stoich(\bm{x}_0)$ in the Hatakeyama--Kaneko MWC model,
$
\norm{\bm{x}}_1
=
\norm{\bm{x}_{\mathrm{eq}}}_1
=
\norm{\bm{x}_0}_1.
$

We use the log-sum inequality \cite[Section 2.7]{Cover2005}: for positive numbers
$a_1,\ldots,a_n$ and $b_1,\ldots,b_n$,
\begin{align}
\sum_{i=1}^n a_i\ln\frac{a_i}{b_i}
\geq
\left(\sum_{i=1}^n a_i\right)
\ln
\frac{\sum_{i=1}^n a_i}{\sum_{i=1}^n b_i}.
\end{align}
Equivalently, in terms of the generalized scalar KL divergence
$
D_{\mathrm{KL}}(a\|b)
:=
a\ln\frac{a}{b}-a+b,
$
it reads
\begin{align}
D_{\mathrm{KL}}\!\left(
\sum_{i=1}^n a_i
\,\middle\|\,
\sum_{i=1}^n b_i
\right)
\leq
\sum_{i=1}^n
D_{\mathrm{KL}}(a_i\|b_i).
\end{align}
Equality holds if and only if $a_i/b_i$ is independent of $i$.

For
$\mathcal{T}(\bm{x})
=
\sum_i x_{T_i}/\norm{\bm{x}}_1$,
we apply the log-sum inequality separately to the $\mathcal{T}$ and $1 - \mathcal{T}$
components:
\begin{align}
\begin{aligned}
&
D_{\mathrm{KL}}\!\left(
\operatorname{Ber}(\mathcal{T}(\bm{x}))
\,\middle\|\,
\operatorname{Ber}(\mathcal{T}(\bm{x}_{\mathrm{eq}}))
\right)
\\
&\coloneqq
\mathcal{T}(\bm{x})
\ln
\frac{\mathcal{T}(\bm{x})}
     {\mathcal{T}(\bm{x}_{\mathrm{eq}})}
+
\bigl(1-\mathcal{T}(\bm{x})\bigr)
\ln
\frac{1-\mathcal{T}(\bm{x})}
     {1-\mathcal{T}(\bm{x}_{\mathrm{eq}})}
\\
&=
\frac{1}{\norm{\bm{x}}_1}
\Biggl[
D_{\mathrm{KL}}\!\left(
\sum_i x_{T_i}
\,\middle\|\,
\sum_i x_{\mathrm{eq},T_i}
\right)
+
D_{\mathrm{KL}}\!\left(
\sum_i x_{R_i}
\,\middle\|\,
\sum_i x_{\mathrm{eq},R_i}
\right)
\Biggr]
\\
&\leq
\frac{1}{\norm{\bm{x}}_1}
\sum_i
\Biggl[
D_{\mathrm{KL}}\!\left(
x_{T_i}\,\middle\|\,x_{\mathrm{eq},T_i}
\right)
+
D_{\mathrm{KL}}\!\left(
x_{R_i}\,\middle\|\,x_{\mathrm{eq},R_i}
\right)
\Biggr]
\\
&=
\frac{1}{\norm{\bm{x}}_1}
D_{\mathrm{KL}}\!\left(
\bm{x}\,\middle\|\,\bm{x}_{\mathrm{eq}}
\right).
\end{aligned}
\end{align}
Equality holds if and only if
$x_{T_i}/x_{\mathrm{eq},T_i}$ is independent of $i$ and
$x_{R_i}/x_{\mathrm{eq},R_i}$ is independent of $i$.

Similarly, applying the weighted log-sum inequality to
$\mathcal{U}$ gives
\begin{align}
\begin{aligned}
&
D_{\mathrm{KL}}\!\left(
\operatorname{Ber}(\mathcal{U}(\bm{x}))
\,\middle\|\,
\operatorname{Ber}(\mathcal{U}(\bm{x}_{\mathrm{eq}}))
\right)
\\
&\coloneqq
\mathcal{U}(\bm{x})
\ln
\frac{\mathcal{U}(\bm{x})}
     {\mathcal{U}(\bm{x}_{\mathrm{eq}})}
+
\bigl(1-\mathcal{U}(\bm{x})\bigr)
\ln
\frac{1-\mathcal{U}(\bm{x})}
     {1-\mathcal{U}(\bm{x}_{\mathrm{eq}})}
\\
&=
\frac{1}{m\norm{\bm{x}}_1}
\Biggl[
D_{\mathrm{KL}}\!\left(
\sum_i (m-i)\bigl(x_{T_i}+x_{R_i}\bigr)
\,\middle\|\,
\sum_i (m-i)
\bigl(x_{\mathrm{eq},T_i}+x_{\mathrm{eq},R_i}\bigr)
\right)
\\
&\hspace{4em}
+
D_{\mathrm{KL}}\!\left(
\sum_i i\bigl(x_{T_i}+x_{R_i}\bigr)
\,\middle\|\,
\sum_i i
\bigl(x_{\mathrm{eq},T_i}+x_{\mathrm{eq},R_i}\bigr)
\right)
\Biggr]
\\
&\leq
\frac{1}{m\norm{\bm{x}}_1}
\sum_i
\Biggl[
(m-i)
D_{\mathrm{KL}}\!\left(
x_{T_i}\,\middle\|\,x_{\mathrm{eq},T_i}
\right)
+
(m-i)
D_{\mathrm{KL}}\!\left(
x_{R_i}\,\middle\|\,x_{\mathrm{eq},R_i}
\right)
\\
&\hspace{10em}
+
i
D_{\mathrm{KL}}\!\left(
x_{T_i}\,\middle\|\,x_{\mathrm{eq},T_i}
\right)
+
i
D_{\mathrm{KL}}\!\left(
x_{R_i}\,\middle\|\,x_{\mathrm{eq},R_i}
\right)
\Biggr]
\\
&=
\frac{1}{m\norm{\bm{x}}_1}
\sum_i
\Biggl[
m
D_{\mathrm{KL}}\!\left(
x_{T_i}\,\middle\|\,x_{\mathrm{eq},T_i}
\right)
+
m
D_{\mathrm{KL}}\!\left(
x_{R_i}\,\middle\|\,x_{\mathrm{eq},R_i}
\right)
\Biggr]
\\
&=
\frac{1}{\norm{\bm{x}}_1}
D_{\mathrm{KL}}\!\left(
\bm{x}\,\middle\|\,\bm{x}_{\mathrm{eq}}
\right).
\end{aligned}
\end{align}
\end{KSblock}

\bibliographystyle{siamplain}
\bibliography{reference}

%% file: ex_shared.tex
\usepackage{lipsum}
\usepackage{amsfonts}
\usepackage{graphicx}
\usepackage{epstopdf}
\usepackage{algorithmic}
\ifpdf
  \DeclareGraphicsExtensions{.eps,.pdf,.png,.jpg}
\else
  \DeclareGraphicsExtensions{.eps}
\fi
\usepackage{amsmath}
\usepackage{amssymb}
\usepackage{mathrsfs}
\usepackage{mathtools}
\usepackage{physics}
\usepackage{enumerate}
\usepackage[sort]{cite} 
\usepackage{bm}
\usepackage{lmodern}
\usepackage{enumitem}
\usepackage[normalem]{ulem}
\usepackage{thm-restate}
\usepackage{comment}
\usepackage{amsopn}

\usepackage{color}%

\renewcommand*{\thetheorem}{\thesection.\arabic{theorem}}
\newtheorem{assumption}{Assumption}
\newtheorem{example}{Example}

\newsiamremark{remark}{Remark}
\newsiamremark{hypothesis}{Hypothesis}
\crefname{hypothesis}{Hypothesis}{Hypotheses}
\newsiamthm{claim}{Claim}
\crefname{assumption}{assumption}{assumptions}

\headers{Convex Analysis of Relaxation Dynamics in CRN}{K. Sugie, D. Loutchko, T. J. Kobayashi.}

\title{Convex Analysis of Relaxation Dynamics in Chemical Reaction Networks and Generalized Gradient Flows\thanks{Submitted to the editors \today.
\funding{This work was supported by JST CREST (Grant Numbers JPMJCR2011 and JPMJCR25Q2), and JSPS KAKENHI (Grant Number  25H01365). The first author is financially supported by JST SPRING (Grant Number  JPMJSP2108).}}}

\author{Keisuke Sugie\thanks{Department of Mathematical Informatics, Graduate School of Information Science and Technology, The University of Tokyo, 7-3-1 Hongo, Bunkyo-ku, Tokyo 113-8656, Japan
  (\email{sgekisk@sat.t.u-tokyo.ac.jp}).}
\and Dimitri Loutchko\thanks{Institute of Industrial Science, The University of Tokyo, 4-6-1, Komaba, Meguro-ku, Tokyo 153-8505 Japan
  (\email{dimitri@sat.t.u-tokyo.ac.jp}, \email{tetsuya@sat.t.u-tokyo.ac.jp}).}
\and Tetsuya J. Kobayashi\footnotemark[3]}

\newif\ifcomments
\commentsfalse

\ifcomments
  \newcommand{\KS}[1]{\textcolor{magenta}{#1}}
  \newcommand{\KSsout}[1]{\textcolor{magenta}{\sout{#1}}}
  \newenvironment{KSblock}
    {\begingroup\color{magenta}}
    {\endgroup}
\else
  \newcommand{\KS}[1]{#1}
  \newcommand{\KSsout}[1]{}
  \newenvironment{KSblock}
    {}
    {}
\fi

\newif\ifcommentsKSTwo
\commentsKSTwofalse

\ifcommentsKSTwo
  \newcommand{\KSTwo}[1]{\textcolor{blue}{#1}}
  \newcommand{\KSTwosout}[1]{\textcolor{blue}{\sout{#1}}}
  \newenvironment{KSTwoblock}
    {\begingroup\color{blue}}
    {\endgroup}
\else
  \newcommand{\KSTwo}[1]{#1}
  \newcommand{\KSTwosout}[1]{}
  \newenvironment{KSTwoblock}
    {}
    {}
\fi

\newcommand{\X}{\mathcal{X}}
\newcommand{\M}{\mathcal{M}}

\newcommand{\F}{\mathcal{F}}
\newcommand{\Stoich}{\mathcal{S}}
\newcommand{\Equib}{\mathcal{E}}
\newcommand{\Spc}{\mathcal{I}}

\newcommand{\Z}{\mathbb{Z}}
\newcommand{\R}{\mathbb{R}}

\newcommand{\B}{\mathcal{B}}

\DeclareMathOperator{\rk}{rk}
\DeclareMathOperator{\Ker}{Ker}
\DeclareMathOperator{\diag}{diag}

\newcommand{\low}[1]{\underline{#1}}
\newcommand{\up}[1]{\overline{#1}}

\DeclareRobustCommand{\UL}[1]{\ifodd\numexpr #1\relax U\else L\fi}

\newenvironment{subasms}{%
    \begin{enumerate}[label=(\UL{\theenumi}), ref=\theassumption(\UL{\theenumi}), font=\normalfont\bfseries, leftmargin=*, labelindent=0pt, labelsep=.5em]
}{%
    \end{enumerate}
}%
\newenvironment{subthms}{%
    \begin{enumerate}[label=(\UL{\theenumi}), ref=\thetheorem(\UL{\theenumi}), font=\normalfont\bfseries, leftmargin=*, labelindent=0pt, labelsep=.5em]
}{%
    \end{enumerate}
}%
\newenvironment{sublems}{%
    \begin{enumerate}[label=(\UL{\theenumi}), ref=\thelemma(\UL{\theenumi}), font=\normalfont\bfseries, leftmargin=*, labelindent=0pt, labelsep=.5em]
}{%
    \end{enumerate}
}%

\makeatletter
\newcommand*{\addFileDependency}[1]{% argument=file name and extension
  \typeout{(#1)}% latexmk will find this if $recorder=0 (however, in that case, it will ignore #1 if it is a .aux or .pdf file etc and it exists! if it doesn't exist, it will appear in the list of dependents regardless)
  \@addtofilelist{#1}% if you want it to appear in \listfiles, not really necessary and latexmk doesn't use this
  \IfFileExists{#1}{}{\typeout{No file #1.}}% latexmk will find this message if #1 doesn't exist (yet)
}
\makeatother

\newcommand*{\myexternaldocument}[1]{%
    \externaldocument{#1}%
    \addFileDependency{#1.tex}%
    \addFileDependency{#1.aux}%
}